\documentclass{article}
\usepackage[letterpaper,top=2cm,bottom=2cm,left=3cm,right=3cm,marginparwidth=1.75cm]{geometry}

\def\imghome{./figures}

\usepackage[utf8]{inputenc} 
\usepackage[T1]{fontenc}    
\usepackage{amsfonts, amsmath, amssymb, amsthm, bm, bbm}
\usepackage{mathtools}
\usepackage{thmtools, thm-restate}
\usepackage{braket}
\usepackage{stmaryrd} 
\usepackage{hyperref} 
\hypersetup{colorlinks=true,allcolors=blue}

\usepackage{url}            
\usepackage{booktabs}       
\usepackage{nicefrac}       
\usepackage{microtype}      
\usepackage{xcolor}         
\usepackage{xstring}
\usepackage{subfigure}
\usepackage{authblk}

\def\01{\{0,1\}}

\newcommand{\Prob}{{\mathbf{Pr}}}
\newcommand{\tinyspace}{\mspace{1mu}}

\newcommand{\norm}[1]{\left\lVert\tinyspace#1\tinyspace\right\rVert}

\newcommand{\abs}[1]{\left\lvert\tinyspace #1 \tinyspace\right\rvert}

\newcommand{\tr}{\operatorname{tr}}
\newcommand{\trnorm}[1]{\norm{#1}_{\tr}}

\def\({\left(}
\def\){\right)}

\def\complex{\mathbb{C}}
\def\real{\mathbb{R}}

\def\<{\langle}
\def\>{\rangle}

\def\E{\mathcal{E}}

\def\R{\mathcal{R}}

\def\1{\mathbbm{1}}
\def\EXP{\mathbb{E}}

\theoremstyle{plain}
\newtheorem{theorem}{Theorem}[section]
\newtheorem{lemma}[theorem]{Lemma}
\newtheorem{corollary}[theorem]{Corollary}
\newtheorem{proposition}[theorem]{Proposition}

\theoremstyle{definition}
\newtheorem{definition}[theorem]{Definition}

\theoremstyle{remark}
\newtheorem{remark}{Remark}

\definecolor{tab-blue}{HTML}{1F77B4}
\definecolor{tab-green}{HTML}{2ca02c}
\definecolor{tab-orange}{HTML}{ff7f0e}
\usepackage{tikz}
\usetikzlibrary{quantikz}
\usetikzlibrary{arrows.meta}
\usetikzlibrary{positioning}
\usetikzlibrary{calc,patterns,angles,quotes}

\usepackage{pgfplots}
\usepgfplotslibrary{groupplots,dateplot}
\usetikzlibrary{patterns,shapes.arrows}
\pgfplotsset{compat=newest}

\newcommand{\mlvec}[1]{\boldsymbol{\mathbf{#1}}}

\newcommand{\grouphomo}{{\Pi}}
\newcommand{\subspacecolumn}{{\mlvec{Q}}}

\newcommand{\deff}{{d_{\mathsf{eff}}}}
\newcommand{\kappaeff}{{\kappa_{\mathsf{eff}}}}
\newcommand{\Heff}{{\mlvec{H}_{\mathsf{eff}}}}
\newcommand{\noise}[1]{{\varepsilon_{#1}}}
\newcommand{\noisevec}{{\mlvec\varepsilon}}
\newcommand{\fronorm}[1]{\norm{#1}_{{F}}}
\newcommand{\opnorm}[1]{\norm{#1}_{\mathsf{op}}}
\newcommand{\alt}[1]{{\mathsf{#1alt}}}
\newcommand{\newketbra}[2]{{|{#1}\rangle\!\langle{#2}|}}
\newcommand{\heafullname}{{Hardware-efficient ansatz}}
\newcommand{\generatorH}[1]{{\mlvec{H}^{(#1)}}}

\def\HH{\mlvec{H}}
\def\YY{\mlvec{Y}}

\definecolor{tab-blue}{HTML}{1F77B4}
\definecolor{tab-green}{HTML}{2ca02c}
\definecolor{tab-orange}{HTML}{ff7f0e}

\numberwithin{equation}{section}

\begin{document}

\title{A Convergence Theory for Over-parameterized Variational Quantum Eigensolvers}

\author[1] {Xuchen You\thanks{Equal contribution.}}
\author[1] {Shouvanik Chakrabarti$^{*}$}
\author[1] {Xiaodi Wu} 
\affil[1]{Department of Computer Science, Joint Center for Quantum Information and Computer Science, University of Maryland, College Park, USA}

\maketitle

\begin{abstract}
The Variational Quantum Eigensolver (VQE) is a promising candidate for
quantum applications on near-term Noisy Intermediate-Scale Quantum (NISQ) computers.
Despite a lot of empirical studies and recent progress in theoretical understanding of VQE's optimization landscape,
the convergence for optimizing VQE is far less understood. 
We provide the first rigorous analysis of the convergence of VQEs in the
over-parameterization regime.
By connecting the training dynamics with the
Riemannian Gradient Flow on the unit-sphere, 
we establish a threshold on the sufficient number of parameters
for efficient convergence, which depends polynomially on the system dimension and the spectral ratio, a property of the
problem Hamiltonian, and could be resilient to gradient noise to some extent.
We further illustrate that this over-parameterization threshold could be vastly reduced for specific VQE instances by establishing an ansatz-dependent threshold paralleling our main result.
We showcase that our ansatz-dependent threshold could serve as a proxy of the trainability of different VQE ansatzes without performing empirical experiments, which hence leads to a principled way of evaluating ansatz design.  
Finally, we conclude with a comprehensive empirical study that supports our theoretical findings. 
\end{abstract}


\section{Introduction}
\label{sec:intro}
\begin{figure}
  \centering
  {
  \begin{minipage}{0.5\linewidth}
    \includegraphics[width=1.0\linewidth]{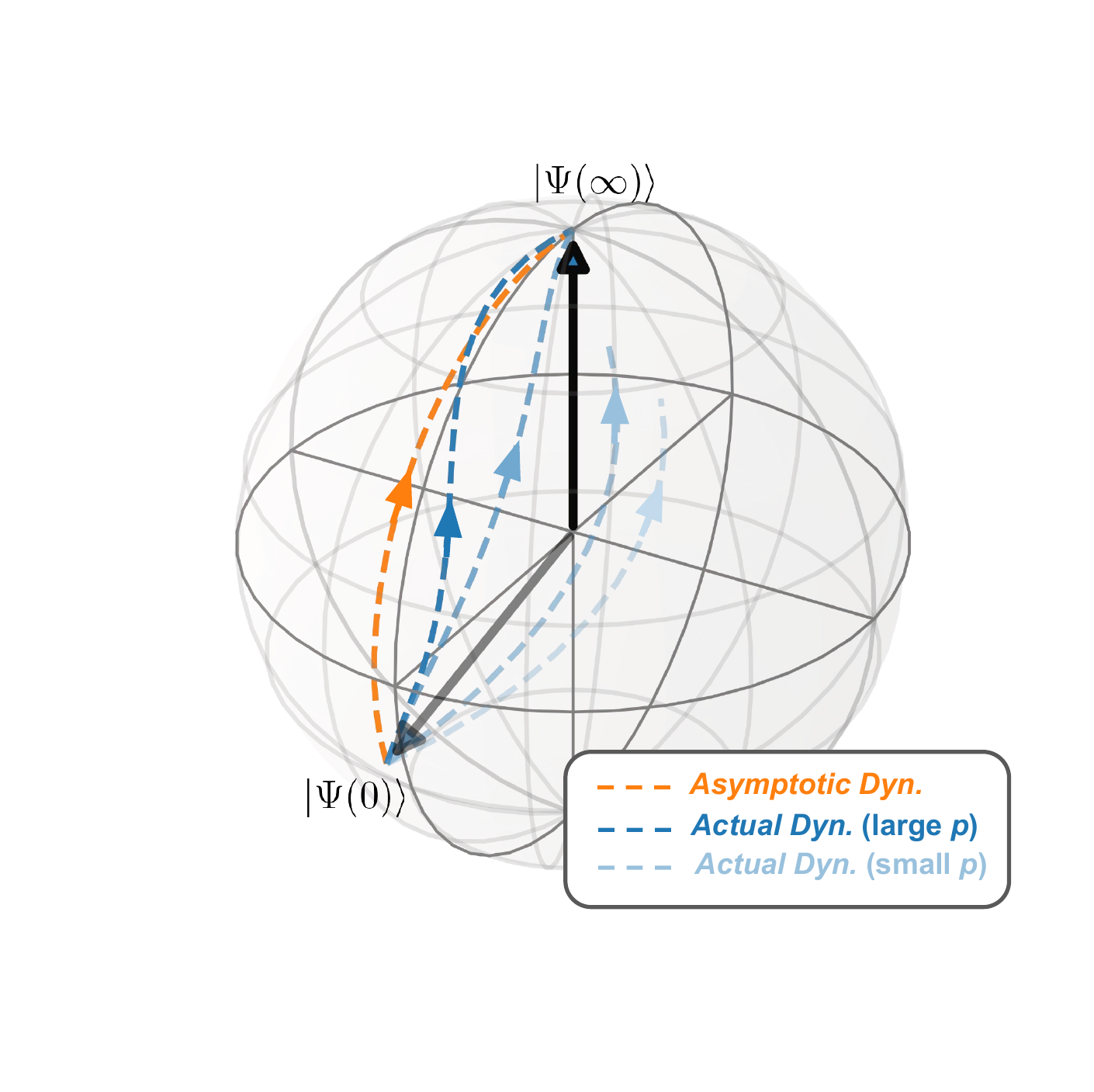}
  \end{minipage}
  \begin{minipage}{0.47\linewidth}
  \caption{Convergence of Over-parameterized VQE. As the number of parameters $p$
    increases, the {\color{tab-blue} actual dynamics} of the output state
    starting from $\ket{\Psi(0)}$ approaches the {\color{tab-orange} asymptotic
      dynamics} (Lemma~\ref{lm:vqe-concentration-init} and
    \ref{lm:vqe-concentration-time}). The convergence of over-parameterized VQE
    follows by the fact that the {\color{tab-orange} asymptotic
      dynamics} converges efficiently to the target output state
    $\ket{\Psi(\infty)}$ (Lemma~\ref{lm:vqe-perturbation}).}
  \label{fig:mainscheme}
  \end{minipage}
  }
\end{figure}

With the recent establishment of quantum supremacy ~\cite{google-supremacy,
  Zhong1460}, there has been a significant interest in demonstrating
applications on Noisy Intermediate-Scale Quantum (NISQ)
computers~\cite{Preskill2018NISQ}. Variational Quantum
Algorithms~(VQAs)~(e.g.,~\cite{NC-VQE}) have become one of the main candidates
for demonstrating such applications. These algorithms use quantum systems
controlled by classical parameters, as parameterized models that are then
classically optimized to perform tasks in machine learning or quantum chemistry.
The parameterized systems used are the so-called \emph{parameterized quantum
  circuits} (often called \emph{quantum ansatz} in Physics.) Commonly studied
VQAs include Quantum Neural Networks~(QNNs)~\cite{farhi2018classification} used
to perform classification tasks in machine learning, and Variational Quantum
Eigensolvers~(VQEs)~that are used to find the ground state of a given
Hamiltonian from physics or chemistry~\cite{NC-VQE}. Variational Quantum
Eigensolvers have gained particular interest: with applications to quantum
chemistry~\cite{google-VQE}, or to combinatorial
optimization~\cite{farhi2014quantum,liu2021variational}. The quantum systems
used in VQEs are challenging to classically simulate so there is a promising
possibility of quantum computational advantages for these important
applications. Furthermore, variational algorithms are particularly suited to the
smaller scale and noisy machines: all the control of the model and the optimization routine is deferred to a classical system and the quantum algorithm only needs to have the capacity to implement the underlying parameterized system, which is often substantially easier than implementing precise digital controls. Optimization algorithms also often have built-in robustness to errors in the objective function which gives the system a degree of natural noise resilience.

Specifically, a VQE with a problem Hermitian matrix $\mlvec{M}$ and an input
state $\ket{\Phi}$ seeks to optimize a parameterized unitary
$\mlvec{U}(\mlvec{\theta})$ to generate an output state
$|\Psi(\mlvec{\theta})\> = \mlvec{U}(\mlvec\theta)\ket{\Phi}$ that well
approximates the ground state of $\mlvec{M}$. There have been many recent works
on the optimization, design, and applications of VQEs (see surveys
in~\cite{tilly2021variational, Cerezo2020}). Despite their potential, there are some
significant challenges in the deployment of VQEs: specifically the optimization
problems involved are highly non-convex and may take exponentially long time (in
the number of optimization parameters) to find a satisfactory global minimum.
The situation in practice may be even worse: practical VQEs are optimized using
gradient-based algorithms, and in general there is no guarantee that such
algorithms converge to a global minimum \emph{at all}, the algorithm may get
trapped in suboptimal local minima or saddle points. The landscape of VQEs also
suffer from ``vanishing gradients'' or \emph{barren
  plateaus}~\cite{mcclean2018barren} which could make convergence very slow even
if the algorithm is never trapped. This makes the design and deployment of VQEs
difficult as it is very hard to know beforehand whether a given system can even
be trained to low error. The difficulty of simulating such systems also limits
the scope and scale of empirical studies, and when small-scale promising results
are demonstrated it is hard to translate the success of these models into clear
design choices. It should be noted that many systems that demonstrate successful
training on a small scale use a parameterized design that is specific to the
problem being solved~\cite{crooks2018performance,mbeng2019quantum}.

Similar difficulty in training is observed in classical deep neural networks:
the landscape of training neural networks are non-convex, can have suboptimal
local minima/stationary points, and may suffer from
the ``vanishing gradients''. Different from quantum algorithms, empirical
studies of classical neural networks can be carried out on a large scale, and these studies
revealed that, despite the possible pitfalls, the loss functions for deep neural networks can often be
efficiently minimized (in fact the training error
often falls \emph{exponentially} with the training time)~\cite{livni2014on,
  arora2018stronger}. An explanation for this phenomenon has been proposed based
on the following observation: neural networks can be efficiently trained in the
highly \emph{over-parameterized} regime, where the number of trainable
parameters is larger than the dimension of the input space as well as the number
of training examples used. 
This observation has been theoretically established starting
with the analysis of wide neural
networks~\cite{jacot2018neural,chizat2018lazy,allenzhu2019convergence} and has been later
established for a variety of architectures (e.g. \cite{arora2019exact}).

The formal similarities between variational quantum algorithms and deep neural
networks raise the important question of whether convergence results can also be
established for an algorithm such as VQE in the over-parameterized regime. VQE
do not have inputs or data samples but their complexity increases with
quantities such as the number of qubits and parameterized gates in the system,
and such quantities can be expected to control the threshold of
over-parameterization. Furthermore, if a sufficient threshold for
over-parameterization can be found to ensure convergence, this threshold may be
used to characterize and compare the trainability of VQE algorithms with
different designs.

\subsection{Contributions}
In this paper, we construct a theory, for the first time,  of the convergence of VQE in the over-parameterization regime. We use this theory to investigate and justify the properties of various VQE instances. We make the following primary claims (to be substantiated later).
\begin{enumerate}
  \item \emph{Gradient optimization algorithms on sufficiently
        over-parameterized Variational Quantum Eigensolvers converge to small
        training error, and the error decays \emph{exponentially} in terms of
        the training time.} We rigorously show that each VQE has an
        over-parameterization threshold, and if the number of trainable
        parameters exceeds this threshold, the training of the system exhibits
        exponential convergence to small training error. We establish these
        results in Section~\ref{sec:vqe-converge}, with noise-robust versions
        included in Section~\ref{sec:robustness}.
  \item \emph{The degree of over-parameterization required to ensure convergence in a Variational Quantum Eigensolver is a useful proxy for its trainability.} We empirically demonstrate that systems with a lower threshold for over-parameterization indeed require a fewer number of parameters in practice in order to ensure their convergence. Specifically, we show that 1) Sufficient over-parameterization is enough to ensure convergence in practice (Section~\ref{sec:exp_confirm}) 2) Varying quantities to adjust the theoretical over-parameterization threshold leads to a corresponding change in the threshold observed in practice (Section~\ref{sec:exp_threshold}).
  \item \emph{Variational Quantum Eigensolvers with problem-specific ansatz
        design require substantially fewer parameters to ensure convergence
        compared to problem-independent designs.} In Section~\ref{sec:subgroup},
        we establish ansatz-dependent versions of our convergence results,
        leading to \emph{ansatz-dependent} quantities that govern the threshold.
        We show that for problem-dependent ansatz choice, these quantities can
        be significantly smaller than those for general-purposed ansatz
        (Section~\ref{sec:subgroup} and Section~\ref{sec:exp_threshold}).
\end{enumerate}

\subsection{Convergence of over-parameterized VQE}
To justify our first claim, we give the first known sufficient conditions for
the convergence of training VQE. We have the following main result (formally
stated as Theorem~\ref{thm:vqe-convergence-full}).
\begin{theorem}[Convergence Theorem (Informal)]
  \label{thm:vqe-convergence-informal}
  For VQE with problem Hamiltonian $\mlvec{M} \in \complex^{d \times d}$,
  it is sufficient to have number of parameters of order
  $\mathsf{poly}(d, \kappa)$ to ensure that with high probability the
  training of VQE converges efficiently to the ground
  state of $\mlvec{M}$, where $d$ is dimension of the VQE problem and
  $\kappa = \frac{\lambda_{d} - \lambda_{1}}{\lambda_{2} -\lambda_{1}}$
  is dependent on the eigenvalues $\lambda_{1}< \lambda_{2}\leq \cdots\leq\lambda_{d}$ of $\mlvec{M}$.
\end{theorem}
To the best of our knowledge this is the first rigorous demonstration of the
trainability of VQE by gradient methods. We further show that such a convergence
still holds with noisy gradients to some extent in
Corollary~\ref{cor:noisy-convergence}.
Our results are inspired by proofs of convergence in over-parameterized deep
neural networks, but the different formalism of VQE leads to significant
conceptual and technical differences.

\paragraph{Overview of techniques} Classical frameworks for analyzing
convergence follow a general scheme: with randomly initialized
parameters, the dynamics for the training approaches a limit as the number of
parameters grows when functions of many parameters are close to their
expectation via the law of large numbers. This limit is termed as the
\emph{asymptotic} dynamics. If the \emph{asymptotic} dynamics can be shown to
converge to small training error, it is established that \emph{infinitely} large
systems also exhibit convergence~\cite{jacot2018neural}. In order to analyze
finite systems that occur in practice, the convergence of the asymptotic
dynamics must tolerate a small amount of noise. It becomes necessary to analyze
the \emph{concentration} of the random variables corresponding to central
quantities governing the dynamics, in order to establish upper bounds on the
number of parameters that suffice to make the deviations of these quantities
from their expectation smaller than the noise tolerance for convergence. This
yields a threshold of over-parameterization that suffices to ensure
convergence~\cite{arora2019exact,allenzhu2019convergence}. We describe the main
elements of our technique and contrast them with the usual treatment in
classical literature (see also Figure~\ref{fig:mainscheme} for the schematic
diagram for our proof).
\begin{itemize}
  \item \textbf{Initialization: }Neural networks are typically initialized by choosing their real parameters from a suitable normal distribution. VQEs are often initialized similarly in practice, however we note that the behavior of the system resembles that with \emph{Haar}-random unitaries. We therefore analyze the system under the following modified parameterization (which we refer to as the partially-trainable ansatz. See also Definition~\ref{def:partially-trainable-ansatz}): given a set of generating Hermitians $\{\generatorH{k}\}_{k=1}^{K}$, we choose the parameterization $\mlvec{U}(\mlvec\theta) = \big(\prod_{l=1}^{p}\mlvec{U}_{l}\exp(-i\theta_{l}\mlvec{H}) \big) \mlvec{U}_{0}$ where each $\mlvec{U}_{l}$ is initialized from the Haar measure on the subgroup of $SU(d)$ generated by the operations $\{\exp{(-i\theta\generatorH{k})} | k \in [K], \theta \in \real\}$, and $\mlvec{H}$ is one of the generating Hermitians (the particular choice can be made without loss of generality). We  will discuss in depth later in Section~\ref{sec:prelim} how the partially-trainable ansatz is connected to the VQE with the usual parameterization when the number of parameters is large. We also establish in Section~\ref{sec:exp_confirm} that the analysis on the partially-trainable ansatz faithfully captures VQEs with the usual parameterization.
  \item \textbf{Asymptotic dynamics: }The next difference between VQEs and neural networks is in the asymptotic dynamics when the number of parameters is large. For neural networks this dynamics is given by kernel training under the \emph{Neural Tangent Kernel}(~\cite{jacot2018neural}). In VQEs we instead find that the asymptotic dynamics is given by \emph{Riemannian Gradient Flow} on the sphere in $d$-dimensions (where $d$ is the Hilbert Space dimension). Specifically, the output state satisfies the differential equation
    \begin{align}
      \frac{d}{dt}\ket{\Psi} = -[\mlvec{M}, \newketbra{\Psi}{\Psi}] \ket{\Psi}. 
    \end{align}
    The convergence of this dynamics \cite{xu2018convergence} has been established previously. Our results require us to establish \emph{robust} versions of these convergence results (Lemma~\ref{lm:vqe-perturbation}) to accommodate deviations from the asymptotic case arising in finite systems. This robust analysis also allows us to show that the convergence of over-parameterized VQE is robust to a certain amount of noise in the estimation of the objective and gradient functions.
  \item \textbf{Central quantity governing concentration to asymptotic
        dynamics:} The concentration of the system dynamics to the
        over-parameterized case can be explained by the lazy evolution of certain
        primary quantities that guide the dynamics. The deviation of these
        quantities from their initialization over the timescales required for
        successful training is shown to be bounded below the robustness
        threshold of the dynamics. In neural networks, this quantity is simply
        the Neural Tangent Kernel (NTK) discussed above~\cite{jacot2018neural}.
        For VQEs we discover a new quantity
        $\mlvec{Y}$~(Equation~\ref{eq:define-Y}) which we call the
        \emph{parameterized projection operator}, that is a Hermitian matrix-valued function of all the parameters. Deviations in $\mlvec{Y}$ during the training dictate the deviations of the actual dynamics from Riemannian Gradient Flow; thus playing the same role as the NTK in classical neural networks.
  \item \textbf{Concentration: } The similarity of the actual dynamics to the asymptotic case is rigorously established by studying the concentration of the distribution of trajectories introduced by the random initialization. The random properties of the Haar distribution are different from the Gaussian processes seen in neural networks and must therefore be analyzed differently. One major difference is that in the classical analysis \cite{allenzhu2019convergence} the deviations throughout training can be bounded using quantities whose concentration need only to be studied at initialization. We found such a technique insufficient. Specifically, analyzing only the gradients and smoothness at individual points does not give a satisfactory bound. We instead model the training of the system as a \emph{random field} and establish results on the correlation between the variables at different points. This allows the maximum deviation to be satisfactorily bounded via Dudley's integral inequality. To the best of our knowledge, this technique is applied for the first time to the convergence of over-parameterized systems.
\end{itemize}
Our guarantees on the convergence of VQE may be seen to contradict the presence of vanishing gradients or barren plateaus in the optimization landscape. We clarify that the presence of vanishing gradients is not enough to prevent fast convergence in many systems, and discuss the relationship of our results to this phenomenon below.

\paragraph{Relation to the barren-plateau phenomenon}
 A phenomenon that is anticipated to present difficulties in the training of variational quantum algorithms is the so-called \emph{barren plateau} phenomenon (first observed by McClean~et.~al~\cite{mcclean2018barren}). The phenomenon shows that the gradients of sufficiently large randomly initialized parameterized quantum systems are likely to be exponentially decaying (with the number of qubits in the system). Specifically, McClean~et.~al consider a $n$-qubit parameterized quantum circuit with an ansatz $\mlvec{U}: \real^{p} \to SU(2^{n})$. When the parameters are randomly initialized to $\mlvec{\theta} \in \real^{p}$, the loss function $L(\mlvec{\theta}) = \<0|\mlvec{U}^{\dagger}(\mlvec{\theta})\mlvec{M}\mlvec{U}(\mlvec{\theta})|0\>$ and its partial derivatives are random variables. If $\mlvec{U}$ is deep enough that $\mlvec{U}(\mlvec{\theta})$ is approximately \emph{Haar} distributed on $SU(d)$, such that 
\begin{align}
  \EXP[L] = 0 \text{ and } \mathrm{Variance}[\partial L/\partial {\theta}_{k}] = O\left(\frac{1}{2^{2n}}\right), \quad \forall  k \in [p].
\end{align}
Therefore, with probability at least $1 - \delta$, $|\partial L/\partial {\theta}_{k}|^{2} \le O\left(\frac{1}{2^{2n}}\log\left(\frac{1}{\delta}\right)\right)$. For systems with a large number of qubits $n$, the gradient components can be vanishingly small, leading to the eponymous barren plateaus in the landscape. The main possible difficulties arising from this phenomenon are two-fold
\begin{itemize}
  \item Firstly, the components of the gradient of variational quantum systems are measured in practice by estimating the expectation value of some Hermitian operator through repeated measurements of \emph{shots}. If the components are exponentially decaying in $n$, the estimates of the expectations need to be correspondingly precise, leading to the number of shots necessary growing exponentially with $n$. This represents an exponential overhead in the training cost of the circuit.
  \item Secondly, the existence of vanishingly small gradients may indicate that the training landscape is infeasible to optimization by gradient based methods. Even if the landscape is free of spurious local minima, an optimization algorithm can in principle require a long time to find any minimum at all. Alongside the existence of \emph{saddle points} that trap gradient based algorithms, barren plateaus constitute one of the main difficulties in non-convex optimization.
\end{itemize}
Our results show that, for variational quantum eigensolvers with sufficient over-parameterization, the latter issue does \emph{not} arise and the deviation of the output from the target space decays exponentially over time as $\exp\left((\lambda_{2} - \lambda_{1}) t/n\right)$ where the $\lambda_{2},\lambda_{1}$ are the two lowest eigenvalues of the problem Hamiltonian (Theorem~\ref{thm:vqe-convergence-full}). This convergence can exist even with vanishing gradients because the gradients along the trajectory are spatially correlated along the training trajectory leading to significant progress towards the global minimum despite the small gradient components. Intuitively, this situation is similar to that in unstable equilibria in dynamical systems, where small forces can combine to cause significant deviations from equilibrium positions. We also show that the convergence is robust to a certain threshold of noise in the gradients (Corollary~\ref{cor:noisy-convergence}). The tolerable noise threshold is however $O(1/2^{{2n}})$ in $n$-qubit systems and therefore cannot resolve the first issue observed above.

However, as introduced in Section~\ref{sec:intro-empirical}, we identify that 
the noise tolerance as well as the parameterization threshold could  further be precisely dictated by a quantity called the \emph{effective dimension} (Corollary~\ref{cor:effective-convergence}) which is equal to $2^{n}$ in the worst case,  but could be significantly smaller for certain structured ansatz (see Section~\ref{sec:subgroup}).
This implies that noise tolerance could be significantly improved for specific VQE instances in practice. 
Finally, we mention that the vanishing gradient problem occurs also in classical neural networks where the gradients decay exponentially with the network \emph{depth}. Over-parameterization has been shown to still enable convergence in such systems (\cite{allenzhu2019convergence}), our results effectively establish the same phenomenon for VQEs.

\subsection{Ansatz-dependent bound for over-parameterization}
\label{sec:intro-empirical}

Our main theorem establishes sufficient conditions for the convergence of VQE. In
Section~\ref{sec:subgroup}, we further establish an ansatz-dependent bound for
a given VQE problem. Intuitively, when the parameterized unitaries corresponding
to the ansatz design is not universal and only explores a subgroup of the whole
unitary group, a tighter over-parameterization threshold can be achieved in
terms of the quantities associated with the subgroup (See
Corollary~\ref{cor:effective-convergence} for the formal statement):
\begin{corollary}[Ansatz-dependent Convergence Theorem (Informal)]
  Consider the smallest subspace containing the input state $\ket{\Phi}$ and is
  invariant under the parameterized unitaries corresponding to the ansatz
  design. It is sufficient to have number of parameters of
  order $\mathsf{poly}(\deff, \kappaeff)$ to ensure that with high probability
  the training of VQE converges efficiently, where
  \begin{itemize}
  \item The effective dimension $d_{\mathrm{eff}}$ is the dimension of the invariant subspace;
  \item The effective spectral ratio
          $\kappa_{\mathrm{eff}} = \frac{\lambda'_{\deff} - \lambda'_{1}}{\lambda'_{2}-\lambda'_{1}}$
          where $\{\lambda'_{i}\}_{i=1}^{\deff}$ are the eigenvalues the problem Hamiltonian
          $\mlvec{M}$ projected onto the subspace.
  \end{itemize}
\end{corollary}
\paragraph{Evaluating trainability of VQE ansatz}
As a consequence of the ansatz-dependent bound, we can now predict and compare the
performances of different ansatz designs. In Section~\ref{subsec:estimate-eff}
we describe the procedure for estimating these problem-dependent quantities $\deff$ and $\kappaeff$.
We show in Section~\ref{sec:exp_threshold} that these quantities can be used to predict the perfomance
of different ansatz on physical problems such as the Transverse Field Ising models and Heisenberg models
without carrying out the gradient descent over multiple random initializations.
These quantities can also be used to explain why some ansatz designs allow solving the VQE problem near
the critical points where the energy spectrum is almost degenerated.
We also highlight that,
since general-purpose ansatz such as HEA tend to have $\deff$ comparable to the system dimension,
to observe quantum advantage in VQE problems, it may be necessary to tailor ansatz design to specific problem instances.

\subsection{Future work}
\label{sec:future-work}
Our results establish a starting point for the rigorous analysis of training in VQEs and show how such analysis can be used to motivate heuristics for their design. In this section, we discuss possible related directions as well as extensions of the work that could lead to a deeper understanding of the VQE training, or lead to applications where the quantum advantage is better motivated theoretically.
We study VQEs using a specialized parameterization (the partially-trainable ansatz). This ansatz can be seen theoretically and empirically to  effectively mimic the convergence behavior of common practical parameterization. Exploring other equivalent parameterizations may be of interest to directly establish tighter over-parameterization bounds. Our current analysis does not yield a direct lower bound on the minimum number of parameters required to ensure convergence so we cannot be sure if our bounds are tight.
Our empirical analysis indicates that the theoretical bounds could be improved,
leading to more practically feasible thresholds for over-parameterization.
Obtaining a critical point between the over and under parameterized regimes
would be ideal, to obtain a complete theoretical characterization. In
particular, the current sufficient over-parameterization threshold is
exponential in the effective dimension and spectral ratio. This is not
surprising, since a universal polynomial bound would yield polynomial time
quantum solutions to some combinatorial optimization problems that are not
expected to be efficiently solvable in general. It is thus very important to
study \emph{structured} ansatz for particular problems arising from physics or
quantum chemistry, where it is possible that the effective dimension itself
could be polynomial in the number of qubits, leading to polynomial time quantum
algorithms in these settings.

\subsection{Related work}
The theoretical study of training variational quantum algorithms has thus far
mostly focused on the characterization of optimization landscape:
You et al~\cite{you2021exponentially} show that local minima with large
sub-optimality gap can prevail in under-parameterized quantum neural networks; on
the other end of the spectrum, Anschuetz~\cite{anschuetz2022critical} studies the role of over-parameterization
in the landscape of quantum generative models, and shows the existence of a
critical point in the number of parameters above which all local minima are
close to the global minimum in function value.

There exists another line of works focusing on the landscape in \emph{Quantum
  Control}, a topic closely related to variational quantum algorithms. In quantum control, the goal is to optimize a continuous-time
parameterized quantum evolution, and VQA can be cast as a quantum control
problem with restricted parameterization. \cite{wu2011role} and
\cite{russell2016quantum} shows that the optimization landscape of quantum
control is free of spurious local minima under a series of conditions based on
controllability of the quantum systems. We note that (full) controllability is
related to over-parameterization in that it requires sufficient number of
parameters so that the parameterized quantum evolution can traverse the
space of unitary operators.

Benign landscape in \cite{anschuetz2022critical,
  wu2011role, russell2016quantum} provides strong evidence that
the over-parameterized instance are more likely to converge efficiently under gradient
methods. 
However, to the best of
our knowledge, our work is the first to show that sufficient over-parameterization
guarantees the successful convergence of a VQE system.

Exploring the role of over-parameterization in the convergence of large classical variational systems such as deep neural networks has been a very active area of research in theoretical machine learning in recent years. Jacot et. al~\cite{jacot2018neural} introduced the notion of a Neural Tangent Kernel, identifying the dynamics of training highly over-parameterized neural networks with kernel training with a fixed kernel. Arora et. al~\cite{arora2019exact} and Allen-Zhu et. al~\cite{allenzhu2019convergence} make this notion exact by showing sufficient over-parameterization conditions for the convergence of various architectures of deep neural networks based on this observation.

There has also been some work on the study of tangent kernels in the quantum setting. Liu et. al \cite{liu2021representation} and Shirai et. al~\cite{shirai2021quantum} hypothesize that the training of Quantum Neural Networks (QNNs) can be identified with kernel training with the corresponding tangent kernel, and empirically study the training of such kernels as a stand in for directly training QNNs. Subsequently Liu et. al~\cite{liu2022analytic} explore a convergence theory for QNNs with one training example by showing that to  in the large dimension limit, the tangent kernel and its gradient are concentrated around their (constant) mean. This is proposed as evidence for the linear convergence of over-parameterized QNNs. The main differences of this work from our setting are threefold: (1) We use the VQE loss function given by the expectation of the target Hermitian, instead of the mean squared loss used in QNN. This results in a different asymptotic dynamics (Riemannian Gradient Flow vs Kernel Regression). Our convergence result is therefore based not on a slow varying kernel but on a new quantity that we introduce, the \emph{parameterized projection operator $\mlvec{Y}$}.(2) Our results do not require the Hilbert space dimension to be arbitrarily large, instead showing sufficient over-parameterization conditions for \emph{any} value of the dimension. (3) \cite{liu2022analytic} only explores the local concentration of the kernel and it's gradient at individual moments during the training which is not sufficient to establish that the quantity is slow varying over the entire training interval (since the gradients at various times are highly correlated and may all be approximately in the same direction resulting in significant contributions to the deviation.)

Larocca et. al~\cite{larocca2021theory} study over-parameterization from an
information theoretic perspective, by defining the over-parameterization threshold
as the point beyond which adding new parameters does not increase the rank of a
Quantum Fischer Information Matrix. A final paper that may be of interest to our setting is \cite{wiersema2022optimizing}, where the authors explore Riemannian Gradient Flow directly over the unitary group. We instead analyze the optimization of VQEs using Riemannian Gradient Flow over the sphere, the convergence of which has already been established classically~\cite{xu2018convergence}, allowing us to establish convergence results in the quantum setting.

\subsection*{Acknowledgement}
We thank Eric Anschuetz, Peter Bartlett, Boyang Chen, Liang Jiang, Bobak Toussi Kiani,
Iordanis Kerenidis, Junyu Liu, Pengyu Liu, and Seth Lloyd for helpful feedback and discussions.  This work received support from the U.S. Department of Energy, Office
of Science, Office of Advanced Scientific Computing Research, Accelerated Research in Quantum Computing and
Quantum Algorithms Team programs, as well as the U.S.
National Science Foundation grant CCF-1755800, CCF-1816695, and CCF-1942837 (CAREER).  S.C. was also partially supported by J.P. Morgan Chase FLARE Fellowship.


\section{Preliminaries}
\label{sec:prelim}
\newcommand{\ansatzname}{{\mathcal{A}}}
\newcommand{\paramset}{{\Theta}}
\newcommand{\randominitdist}{{\mathcal{D}}}

\subsection{Quantum information preliminaries}
\label{sec:quantum-prelim}
Quantum computations operate on systems of quantum registers or qubits.
The state of a system composed of $n$~qubits is represented by a trace-$1$
positive semidefinite (PSD) Hermitian $\mlvec{\rho} \in\complex^{d\times d}$,
with $d = 2^{n}$. We will mainly focus on quantum states with rank-$1$ density
matrices (i.e. \emph{pure states}), which can be unambiguously represented by a
normalized vector $\mlvec{v}$ in $\complex^{2^{n}}$ such that
$\mlvec{\rho} = \mlvec{v}\mlvec{v}^\dagger$. We use the Dirac notation to denote
the state vector $\mlvec{v}$ (resp. its covector $\mlvec{v}^{\dagger}$) with
$\ket{\Psi}$ (resp. $\bra{\Psi}$). Using the same notation, the inner product of
two states $|\Psi\>,|\Psi'\>$ is denoted by $\<\Psi|\Psi'\>$.

An operation over a quantum state is a linear map that is completely positive,
Hermitian-preserving and trace-preserving (\cite{watrous2018theory}). We focus
on the \emph{unitary operators}, i.e. \emph{unitary} matrices $\mlvec{U}$ that
map a state vector $\ket{\Psi}$ to $\mlvec{U}\ket{\Psi}$. We refer to the set of all \emph{special unitary}
matrices (i.e. unitary matrices with determinant $1$) on a $d$-dimensional Hilbert space as $SU(d)$.
A set of unitary gates commonly used on qubits are the \emph{Pauli
gates}:
\begin{align}\label{eq:paulis}
   \mlvec{X} =
    \begin{bmatrix}
      0 & 1 \\
      1 & 0
    \end{bmatrix},\
  \mlvec{Y} =
    \begin{bmatrix}
            0 & -i \\
            i & 0 \\
          \end{bmatrix},\
  \mlvec{Z} =
    \begin{bmatrix}
      1 & 0 \\
      0 & -1 \\
    \end{bmatrix}.
\end{align}

Information is extracted from a quantum system by \emph{measuring} a quantum
state. Each measurement is represented by a set of positive semidefinite \emph{Hermitian} operator
$\{\mlvec{M}_{m}\}$ that sums to identity and a set of real values
$\{\lambda_{m}\}$. Performing the measurement on the quantum state $\ket{\Psi}$
yields outcome $\lambda_{m}$ with probability
$\bra{\Psi}\mlvec{M}_{m}\ket{\Psi}$, and the expected value of the measurement
is given by $\bra{\Psi}\mlvec{M}\ket{\Psi}$ with
$\mlvec{M} := \sum_{m}\lambda_{m}\mlvec{M}_{m}$.

The state space of a composite system is the tensor product of the state spaces
of the subsystems: for example, for two subsystems with independent states
$\ket{\Psi_{1}}$ and $\ket{\Psi_{2}}$, the joint state is represented by
$|\Psi_{1}\>\otimes |\Psi_{2}\>$.
Composite states that cannot be represented as a tensor product of their
subsystems are said to be \emph{entangled}, and unitaries that transform product
states into entangled states are known as \emph{entangling} operations.
For composite system defined on
$S_{1}\otimes S_{2}$,  the partial trace on the first subsystem is a linear
operation defined such that
$\tr_{1}(\mlvec{A}\otimes \mlvec{B}) = \tr(\mlvec{A})\mlvec{B}$ for all linear
maps $\mlvec{A}$ on $S_{1}$ and $\mlvec{B}$ on $S_{2}$.

Throughout this paper, we deal with
different norms for vectors and matrices. We use $\opnorm{\cdot}$, $\fronorm{\cdot}$, $\trnorm{\cdot}$ to
denote the operator norm (i.e. the largest eigenvalue), Frobenius norm and the
trace norm of matrices; we use
$\norm{\cdot}_{p}$ to denote the $p$-norm of vectors, with subscript omitted for
$p=2$. We use $\tr(\cdot)$ for trace operation.

\subsection{Variational quantum algorithms}
\label{sec:vqa-prelim}
Variational Quantum Algorithms (VQA) is a paradigm of quantum algorithms that searches
over a family of parameterized quantum operations referred to as \emph{quantum
  ansatz}.
More concretely, a $p$-parameter ansatz on a $d$-dimensional Hilbert space
$\mlvec{U} \colon \real^{p}\to\complex^{d \times d}$ maps real parameters
$\mlvec{\theta}$ to a unitary operator $\mlvec{U}(\mlvec{\theta})$.

An important example of VQA is the Variational Quantum Eigensolvers (VQE)
defined as below:
\begin{definition}[Variational quantum eigensolvers]
\label{def:vqe}
A $d$-dimensional variational quantum eigensolver instance is specified by a
triplet $(\mlvec{M}, \ket{\Phi}, \mlvec{U})$ with
$d\times d$ problem Hamiltonian $\mlvec{M}$, an input state $\ket{\Phi}\in\complex^{d}$ and
an ansatz $\mlvec{U} \colon \real^{p}\to\complex^{d \times d}$. Let $\lambda_{1} < \lambda_{2} \leq \cdots \leq \lambda_{d}$ be
the eigenvalues of $\mlvec{M}$ in an ascending order. The goal is to approximate
the ground state of $\mlvec{M}$ (i.e. the eigenvector associated with $\lambda_{1}$)
with $\mlvec{U}(\mlvec\theta)\ket{\Phi}$ by solving the optimization problem:
\begin{align}
\min_{\mlvec\theta}L(\mlvec\theta):=\bra{\Phi}\mlvec{U}^{\dagger}(\mlvec\theta)\mlvec{M}\mlvec{U}(\mlvec\theta)\ket{\Phi}\label{eq:vqe}
\end{align}
\end{definition}
The search for the optimal parameters $\mlvec\theta^{\star}$ are commonly performed
by gradient descent
$\mlvec\theta \leftarrow \mlvec\theta - \eta\nabla_{\mlvec\theta}L(\mlvec\theta)$.
For sufficiently small learning rate $\eta$, the dynamics of gradient descent
reduces to that of gradient flow
\begin{align}
d\mlvec\theta/dt = - \eta\nabla_{\mlvec\theta}L(\mlvec\theta)
\end{align}

\paragraph{Fully- and partially-trainable ansatz}
The parameterization of $\mlvec{U}$ is referred to as the ansatz design in
the quantum computing literature. A popular
choice of ansatz design is the \emph{\heafullname} (HEA, e.g.
\cite{kandala2017hardware}). HEA makes use of native gates of quantum computers
and is composed of interleaving single-/two-qubit Pauli rotations and
entanglement unitaries implemented with CZ / CNOT gates. The main motivation
behind the design is to facilitate the implementation on real quantum machines.
Another choice of ansatz design is the \emph{Hamiltonian variational ansatz}
(HVA). HVA is partially inspired by the adiabatic theorem and utilizes the
structure of the problem Hamiltonians (e.g. \cite{Wiersema2020,
  mixedkitaev2021}). For HVA $\mlvec{U}$ composed of parameterized rotations
generated by a set of Hermitians that sums to the problem Hamiltonian.

In this work, we consider a general family of ansatz that includes HEA and HVA
as special cases, specified by the number of layers $L$ and a set of $K$
$d\times d$ Hermitians $\{\generatorH{1}, \generatorH{2}, \cdots, \generatorH{K}\}$:
\begin{definition}[Fully-trainable ansatz]
  \label{def:fully-trainable-ansatz}
A fully-trainable $L$-layer ansatz with a set of Hermitians
$\ansatzname=\{\generatorH{1}, \generatorH{2}, \cdots, \generatorH{K}\}$ has $K\cdot L$ trainable parameters
and is defined as
\begin{align}
  \mlvec{U}^{(L)}(\mlvec{\theta}) = \prod_{l=1}^{L}\prod_{k=1}^{K}\exp(-i\theta_{l,k} \generatorH{k}).\label{eq:fully-trainable-ansatz}
\end{align}
The superscript $L$ will be omitted when there is no ambiguity.
\end{definition}
To see that the ansatz defined in Definition~\ref{def:fully-trainable-ansatz} is
a superset of HVA and HEA, notice that the fully-trainable ansatz is an HVA if the problem
Hamiltonian $\mlvec{M}$ can be represented as a linear combination of
$\{\generatorH{k}\}_{k=1}^{K}$. As for HEA, it can in general can be expressed as
\begin{align}
  \mlvec{U}^{(L)}(\mlvec{\theta}) = \prod_{l=1}^{L}\big(\prod_{k=1}^{K'}\exp(-i\theta_{l,k'} \generatorH{k'}) \mlvec{U}_{\mathsf{ent}}\big)
\end{align}
where $\generatorH{k'}$ are single-/two-qubit Pauli rotations and
$\mlvec{U}_{\mathsf{ent}}$ corresponds to an entanglement layer composed of CZ
and CNOT gates. If the smallest integer $C$ such that
$\mlvec{U}_{\mathsf{ent}}^{C} = \mlvec{I}$ exists, the HEA can be expressed as
a fully-trainable ansatz with $K = C\cdot K'$, with each generating Hermitian
represented as
$\mlvec{U}_{\mathsf{ent}}^{c}\generatorH{k}\big(\mlvec{U}_{\mathsf{ent}}^{c}\big)^{\dagger}$
for $c \in [C]$ and $k\in[K]$.

Apparently, the parameterization is determined by a fixed set of Hermitians
 $\ansatzname = \{\generatorH{1}, \cdots, \generatorH{K}\}$ and a
domain of parameters in each layer $\paramset\subseteq\real^{K}$ up to the
choice of number of layers $L$. Hence we refer to the pair
$(\ansatzname, \paramset)$ as the design of the ansatz.
The subscript $\paramset$ will be dropped when $\paramset = \real^{K}$ for a
more concise notation.

Given an ansatz design $(\ansatzname, \paramset)$, the set of all
achievable unitary matrices forms a subgroup of $SU(d)$:
\begin{align}
  G_{\ansatzname, \paramset}  = \cup_{L=0}^{\infty}\{\mlvec{U}^{(L)}(\mlvec\theta): \mlvec\theta\in\paramset^{L}\subseteq\real^{K\cdot L}\}.
\end{align}
For many choices of ansatz with a limited set of $\ansatzname$,
$G_{\ansatzname, \paramset}$ is a strict subgroup of $SU(d)$. We omit the
subscript $\paramset$ and denote the
subgroup as $G_{\ansatzname}$ with the domain of the parameters is clear from the context.

Using the group $G_{\ansatzname}$, we define the partially-trainable ansatz associated with $\ansatzname$ as:
\begin{definition}[Partially-trainable anastz for $\ansatzname$]
  \label{def:partially-trainable-ansatz}
  Let the subgroup $G_{\ansatzname}$ be a subgroup of $SU(d)$ associated with
  fully-trainable ansatz with a set of Hermitians $\ansatzname=\{\generatorH{1}, \generatorH{2}, \cdots \generatorH{K}\}$.
  The corresponding $p$-parameter partially-trainable ansatz is defined as:
  \begin{align}
    \mlvec{U}(\mlvec\theta) = \big(\prod_{l=1}^{p}\mlvec{U}_{l}\exp(-i\theta_{l}\mlvec{H}) \big) \mlvec{U}_{0}.\label{eq:partially-trainable-ansatz}
  \end{align}
  Here $\mlvec{H}$ is an arbitrary Hermitian in $\ansatzname$ and $\mlvec{U}_{l}$ are $i.i.d.$ sampled from
  the Haar measure over ${G}_{\ansatzname}$.
\end{definition}
We highlight that the partially-trainable ansatz can be viewed as a
fully-trainable ansatz trained on a subset of the parameters, hence the name
``partially trainable'':
without loss of generality, assume we choose $\generatorH{1}$ as the generating
Hermitian $\mlvec{H}$ in Definition~\ref{def:partially-trainable-ansatz}. Performing
gradient descent on the parameters corresponding to $\generatorH{1}$ in every
$L'$-layers (i.e. gradient descent on
$\theta_{1, 1}, \theta_{L'+1, 1}, \theta_{2L'+1, 1}, \cdots$) of a
randomly-initialized fully-trainable ansatz is then equivalent to optimizing the
partially-trainable ansatz with $\mlvec{H} = \generatorH{1}$ with $\mlvec{U}_{l}$
being a $L'$-step random walk with step sample from
$S:=\{\prod_{k=1}^K \exp(-i\theta_k\generatorH{k}):\mlvec{\theta}\in \paramset\subseteq\real^{K}\}$.
Under mild regularity conditions, the random walk converges to the Haar measure over
$G_{\ansatzname, \paramset}$ (See \cite[Section 3]{varju2012random}).
In Section~\ref{subsec:exp1}, we observe that the dynamics of the
partially-trainable ansatz faithfully captures the dynamics of the
fully-trainable ansatz in the over-parameterization regime.

\subsection{Convergence in over-parameterized classical systems}
\label{sec:classical-overparam}
Over-parameterization has been proposed as an explanation for the convergence of the highly nonconvex training of parameterized classical models such as artificial neural networks~\cite{jacot2018neural, arora2019exact, chizat2018lazy}. The convergence of the models arises from two main phenomenon:
\begin{enumerate}
  \item \textbf{Convergence of expected dynamics:} When the parameters are randomly initialized, the expected dynamics of the training are shown to exhibit convergence to a global minima. The expected dynamics is therefore a smoothed version of the actual dynamics that removes some of the irregularities that can lead to a failure in convergence.
    \item \textbf{Convergence under perturbation:} Despite the convergence of the training dynamics in expectation, the actual training corresponds to a particular setting of initial parameters. This leads to the actual training being a perturbed version of the expected dynamics, it is thus necessary to show that the convergence of this dynamics is robust to small perturbations.
  \item \textbf{Concentration at initialization:} Due to the law of large numbers, with high probability, deviations from the expected dynamics decrease as the number of random parameters increases. Over-parameterization thus plays the crucial role of leading to the \emph{concentration} of the dynamics around the expected value, allowing the magnitude of random perturbations to be bounded with high probability.
    \item \textbf{Lazy training:} It must be shown that the actual training concentrates throughout the training given the convergence at initialization. This phenomenon has been characterized as \emph{lazy training}~\cite{chizat2018lazy}, where the dynamics of a system at initialization remain a good approximation throughout its training. Once again, over-parameterization plays an important role in ensuring this phenomenon; as the number of parameters increases the changes in each parameter become smaller with high probability over the course of training.
\end{enumerate}
This method can be illustrated by the example of the \emph{Neural Tangent Kernel}~\cite{jacot2018neural}, which has been used to show convergence while training several over-parameterized classical neural networks including wide feedforward networks~\cite{arora2019exact}.

Consider a classical classification problem where the input data is drawn from a distribution $p_{in}$ over $\R^{n_{0}}$ and an output in $\R^{n_{L}}$, the space of valid functions is given by $\mathcal{F} = \{\mlvec{f} \colon \R^{n_{0}} \to \R^{n_{L}}\}$. The model is specified as a \emph{realization function} mapping $p$ parameters to candidate functions $\mlvec{F}^{(L)} \colon \R^{p} \to \mathcal{F}$. Denoting the parameters at time $t$ by $\theta(t) = \(\theta_{1}(t),\dots,\theta_{p}(t)\)$, the function at time $t$ is given by $\mlvec{F}^{(L)}(\theta(t))$. The data distribtution induces an inner product over $\mathcal{F}$ given by $\langle \mlvec f, \mlvec g \rangle_{p_{in}} = \EXP_{x \sim p_{in}}[\mlvec f(x)^{T}\mlvec g(x)]$. Given a cost function $C$, the gradient flow dynamics of the system correspond to \emph{kernel training} with respect to the \emph{neural tangent kernel} (NTK) given by
  $\tilde{\mlvec{K}} = \sum_{l=1}^{p} \frac{\partial}{\partial\theta_{l}}\mlvec{F}^{(l)}(\theta) \otimes \frac{\partial}{\partial{\theta_{l}}}\mlvec{F}^{(l)}(\theta)$.

Let $\mlvec{y} \in \mathcal{F}$ be the true function mapping inputs to ouputs resulting in the residual function $\mathcal{r}(\theta(t)) = \mlvec{y} - \mlvec{F}^{(L)}(\theta(t))$. If $C$ is the squared loss function, the dynamics of the system is simply given by $\dot{\mlvec{r}} = -\eta \tilde{\mlvec{K}} \mlvec{r}$ where $\eta$ is the chosen step size. It is known that if $\tilde{\mlvec{K}}$ is a constant positive definite matrix, the system exhibits linear convergence. Following the above recipe, this leads to a framework for showing the convergence of classical neural networks, it is shown that $\mlvec{K} = \EXP(\tilde{\mlvec{K}}(\theta(0)))$ is a positive definite constant matrix. It is also shown that the dynamics $\dot{\mlvec{r}} = -\eta \tilde{\mlvec{K}} \mlvec{r}$ converges whenever $\lVert\tilde{\mlvec{K}} - \mlvec{K}\rVert \le \epsilon_{0}$. Further define an over-parameterization threshold $P^{(L)}(n_{0},n_{L})$Convergence can then be established via the following propositions:
\begin{enumerate}
  \item \textbf{Concentration at initialization}: If $p > P^{{L}}$, $\lVert\tilde{\mlvec{K}}(\theta(0)) - \mlvec{K}\rVert \le \epsilon_{0}$ with probability at least~$9/10$.
  \item \textbf{Small perturbations imply convergence}: $\lVert\tilde{\mlvec{K}}(\theta(t)) - \mlvec{K}\rVert \le \epsilon_{0} $ for all $t < t'$, we have $\lVert \mlvec{r}(t) - \mlvec{\tilde{\mlvec{r}}(t)}\rVert \le \epsilon_{1}$ for all $t \le t_{1}$, where $\tilde{\mlvec{r}}$ denotes the residuals when the kernel is frozen at initialization (in which case the system is known to converge).
  \item \textbf{Convergence implies small perturbations}: If $p > P^{{L}}$, and $\lVert \mlvec{r}(t) - \mlvec{\tilde{\mlvec{r}}(t)}\rVert \le \epsilon_{1}$ for all $t < t'$, we have $\lVert\tilde{\mlvec{K}}(\theta(t)) - \mlvec{K}\rVert \le \epsilon_{0}$ for all $t \le t'$ with probability at least $9/10$
\end{enumerate}
These propositions are sufficient to inductively prove the convergence of the training dynamics to a global minimum. Consider the earliest time $t_{0}$ where the perturbation in the kernel is too large; by the final proposition this can only occur if the convergence of the system is violated at some time $t'_{0} < t_{o}$. However, by the second proposition, this would imply that for an earlier time $t''_{0}$ the kernel perturbation must have been too large, contradicting our initial assumption that $t_{0}$ was the earliest such time. This shows that both the small perturbation condition as well as the convergence of the system are maintained throughout the training.


\section{A convergence theory for VQE}
\label{sec:vqe-converge}
In this section we use ideas from the classical theory of over-parameterized variational systems to give sufficient conditions for the convergence of a VQE to zero loss. We also establish the main factors influencing the (linear) rate of convergence. 

As discussed in Section~\ref{sec:classical-overparam}, the random initialization of the parameters plays an important role in the convergence. For the results of this section, we rely on the \emph{partially-trainable ansatz} (Definition~\ref{def:partially-trainable-ansatz}). To demonstrate the main techniques we restrict ourselves to the case where the set of Hermitian defining the ansatz $\ansatzname$ form a complete basis of the Lie Algebra of $SU(d)$. The corresponding induced subgroup $G_{\ansatzname}$ is therefore the entire unitary group $SU(d)$. 
Note that this assumption enforces an implicit restriction on the generating Hermitians $\{\generatorH{k}\}_{k=1}^{K}$, which is however satisfied by typical Hermitians used in practice (e.g., single-/two-qubit Pauli rotations).
Under this setting, we instantiate partially-trainable ansatz with $G_{\ansatzname} = SU(d)$ as follows: 
\begin{definition}[Partially-trainable ansatz when $G_{\ansatzname} = SU(d)$]
  \label{def:frozen-ansatz}
Consider a quantum system over $n$-qubits with a corresponding Hilbert space of dimension $d = 2^{n}$. We define a ($p$-parameter) random ansatz parameterized by
$\mlvec{\theta}\in\real^{p}$:
\begin{align}
  \mlvec{U}(\mlvec{\theta}) = \mlvec{U}_p\exp(-i\theta_p\mlvec{H})\cdot \cdots \cdot \mlvec{U}_1\exp(-i\theta_1\mlvec{H})\mlvec{U}_0
\end{align}
where $\mlvec{U}_l$ are sampled $i.i.d.$ with respect to the Haar measure over
$SU(d)$, and $\mlvec{H}$ is a trace-$0$ Hermitian.
\end{definition}
We comment that, due to the Haar randomness of $\mlvec{U}_l$, the choice of $\mlvec{H}$ is arbitrary in terms of its eigenvectors for fixed eigenvalues. And as we shall see later, the over-parameterization depends on the eigenvalues only through the scaling of $Z(\mlvec{H},d) = \tr(\mlvec{H}^2) / (d^2-1)$ with $d$. The partially-trainable ansatz has several appealing analytical properties for the analysis of convergence. Firstly, the random initialization of the ansatz is restricted to the unitaries $\mlvec{U}_l$ which are never updated during the training. There is thus a clear separation between the random initialization and dynamic evaluation of the system. Furthermore, since $\mlvec{U}_l$ are sampled from the right-invariant Haar measure over $SU(d)$, the distribution of $\mlvec{U}(\theta)$ is independent of the scheme used to initialize the trainable parameters $\mlvec{\theta}$. Finally, the distribution of the ansatz is also invariant under \emph{arbitrarily} changes to the parameters, therefore the distribution of the ansatz does not change as the system evolves. This intuitively indicates that important quantities connected to the ansatz may be \emph{slow-varying} as their expectation value remains the same. The remainder of the section will formalize this notion.

Recall that a VQE instance is specified by a specified by a problem Hamiltonian
$\mlvec{M} \in \complex^{d\times d}$, an input state
$\ket{\Phi} \in \complex^{d}$ and an $p$-parameter variational ansatz
$\mlvec{U}:\real^{p}\rightarrow \complex^{d\times d}$, and seeks to solve the following optimization problem
\begin{align}
\min_{\mlvec\theta}L(\mlvec\theta):=\bra{\Phi}\mlvec{U}^{\dagger}(\mlvec\theta)\mlvec{M}\mlvec{U}(\mlvec\theta)\ket{\Phi}
\end{align}
We investigate the dynamics of training the system using \emph{gradient flow}, ie. the parameters are updated according to the differential equation $d\mlvec{\theta}/dt = -\eta \nabla_{\mlvec{\theta}} L(\theta)$, where $\eta$ is some previously chosen learning rate.

\paragraph{Notation and definitions} We define here some notations, quantities, and conventions that will play an important role in the forthcoming analysis.
\begin{itemize}
  \item $\mlvec{W}_{d^{2}\times d^{2}}$ denotes the swap operator $\sum_{a,b\in[d]}\mlvec{E}_{ab}\otimes\mlvec{E}_{ba}$.
  \item $\mlvec{I}_{d^{2}\times d^{2}}$ denotes the identity
        matrix in $\complex^{d^{2}\times d^{2}}$.

  \item Without loss of generality we assume $\tr(\mlvec{H}) = 0$.
  \item We use $\mlvec{U}(\theta_{k}) = \mlvec{U}_{k}\exp(-i\theta_{k}\mlvec{H})$ to denote the composition of the randomly initialized Haar operator and the parameterized rotation, corresponding to the $k^{\mathrm{th}}$ real parameter.
  \item We use the short hand $\mlvec{U}_{l:p}(\mlvec{\theta})$ to represent
        \begin{align}
          \mlvec{U}_{l:p}(\mlvec\theta) := \mlvec{U}_p(\theta_p) \cdots \mlvec{U}_{l}(\theta_{l}),
        \end{align}
         where $\mlvec\theta$ may be omitted when not ambiguous. $\mlvec{U}_{l:p}$ represents the composition of the unitaries corresponding to the last $p - l + 1$ parameters.
  \item A common normalizing factor appears due to Haar integral: for any
        $d\times d$-Hermitian $\mlvec{A}$, define
        $Z(\mlvec{A}, d):=\frac{\tr(\mlvec{A}^{2})}{d^{2}-1}$.
  \item The matrix
    \begin{align}
      \mlvec{V}_l(\mlvec\theta):=\mlvec{U}_p(\theta_p)\cdots\mlvec{U}_l(\theta_l)=\mlvec{U}_p\exp(-i\theta_p\mlvec{H})\cdot\cdots\cdot\mlvec{U}_{l}\exp(-i\theta_{l}\mlvec{H}),
    \end{align}
    is defined as the composition of all the layers of the ansatz from $p$ to $l$.
        Using this notation, the ansatz $\mlvec{U}(\mlvec\theta)$ can be written
        as $\mlvec{V}_{1}(\mlvec\theta)\mlvec{U}_{0}$.
  \item The output state $|\Psi(t)\>$ at time $t$ is given by $\mlvec{U}(\mlvec{\theta}(t))\ket{\Phi}$.
  \item We define the matrix $\mlvec{H}_l(\mlvec\theta):=\mlvec{V}_l(\mlvec\theta)\mlvec{H}\mlvec{V}_l(\mlvec\theta)^\dagger$.
  \item The central quantity dictating the dynamics of VQE (see Lemma~\ref{lm:vqe-dynamics}) is given by the matrix
    \begin{align}
      \label{eq:define-Y}
      \mlvec{Y}(\mlvec\theta):=\frac{1}{pZ(\mlvec{H},d)}\sum_{l=1}^{p}\mlvec{H}_{l}^{\otimes 2},
    \end{align}
    which we call the \emph{parameterized projection operator}. \footnote{Note that $\mlvec{Y}$ is not a projection in the sense of linear algebra, ie., $\mlvec{Y}^{2} = \mlvec{Y}$ in general, but rather is a quantity that controls (Lemma~\ref{lm:vqe-dynamics}) the projection of the dynamics of the output state onto the Reimannian Sphere.}
  \item  We note that $\ket{\Psi(t)}, \{\mlvec{U}_{l}\}, \{\mlvec{V}_{l}\}, \{\mlvec{H}_{l}\}, \mlvec{Y}$ are also functions of
  time $t$ through the evolution of the parameters $\mlvec{\theta}(t)$. We shall sometimes abuse notation write $\mlvec{Y}$ as a function only of $t$, this is to be implicitly understood as $\mlvec{Y}(\mlvec{\theta}(t))$.
  \item All derivatives of matrix valued quantities with respect to real parameters are to be taken elementwise.
\end{itemize}

\paragraph{Main elements of theory} Following the structure in Section~\ref{sec:classical-overparam} we outline the main components of our analysis:
    \begin{enumerate}
        \item \textbf{Identifying (idealized) gradient flow dynamics of VQE with classical dynamics that is known to converge.} We show that the dynamics of VQE is equivalent to Riemannian Gradient Flow (RGF) over the unit sphere in $d$-dimensions, by tracking the evolution of the output state. Specifically, we have the following lemma
        \begin{restatable}[VQE output-state dynamics under gradient flow]{lemma}{vqedynamics}
      \label{lm:vqe-dynamics}
      For a VQE instance $(\mlvec{M}, \ket{\Phi}, \mlvec{U})$, with $\mlvec{U}$
      being the ansatz defined in Definition~\ref{def:frozen-ansatz}, when optimized
      with gradient flow with learning rate $\eta$, the output state $\ket{\Psi(t)} = \mlvec{U}(\mlvec{\theta}(t))\ket{\Phi}$
      follow the dynamics
      \begin{align}
        \label{eq:vqe-dynamics}
        \frac{d}{dt}\ket{\Psi(t)} &= -(\eta\cdot p\cdot Z(\mlvec{H},d))\tr_{1}(\mlvec{Y}([\mlvec{M}, \newketbra{\Psi(t)}{\Psi(t)}]\otimes \mlvec{I}_{d\times d})) \ket{\Psi(t)}.
      \end{align}
      Here the Hermitian $\mlvec{Y} \in \complex^{d^{2}\times d^{2}}$ is a
      time-dependent matrix defined as
      \begin{align}
        \label{eq:Y-def}
            \mlvec{Y}(\mlvec\theta):=\frac{1}{pZ(\mlvec{H},d)}\sum_{l=1}^{p}\big(\mlvec{U}_{l:p}(\mlvec\theta(t))\mlvec{H}\mlvec{U}^{\dagger}_{l:p}(\mlvec\theta(t))\big)^{\otimes 2}.
      \end{align}
    \end{restatable}
        Over the randomness of ansatz initialization, for all
    $\mlvec\theta\in\real^{p}$, the expected value of matrix $\mlvec{Y}$ is
    $\mlvec{Y}^{\star}=\mlvec{W}_{d^{2}\times d^2}-\frac{1}{d}\mlvec{I}_{d^{2}\times d^2}$. If we choose
    $\eta = \frac{1}{p Z(\mlvec{H})}$, the VQE dynamics allows the following
    decomposition:
    \begin{align}
      \label{eq:vqe-ideal-dynamics}
        \frac{d}{dt}\ket{\Psi(t)} = -[\mlvec{M}, \newketbra{\Psi(t)}{\Psi(t)}] \ket{\Psi(t)} - \tr_{1}\big((\mlvec{Y}-\mlvec{Y}^{\star}) \cdot [\mlvec{M}, \newketbra{\Psi(t)}{\Psi(t)}\otimes \mlvec{I}_{d\times d}]\big)\ket{\Psi(t)}.
    \end{align}
    The first term is exactly the Riemannian gradient descent over the unit sphere with loss
    function $\bra{\Psi(t)}\mlvec{M}\ket{\Psi(t)}$, which is known to converge linearly to
    the ground state~\cite{xu2018convergence}. Lemma~\ref{lm:vqe-dynamics} shows that the main quantity that controls the deviation of the VQE gradient flow from RGF over the sphere is $\mlvec{Y} - \mlvec{Y}^{\star}$, ie. the deviation of the parameterized projection operator $\mlvec{Y}$ from it's expectation. It also shows that the value of $\mlvec{Y}$ is a parameterized quantity that dictates the projection of the time derivative of the loss function onto the sphere of normalized quantum states. Finally, we contrast the above dynamics of the residue vector $\mlvec{r}$ for feedforward ReLU networks, given by $\dot{\mlvec{r}} = -\eta \tilde{\mlvec{K}}\mlvec{r}$, where $\tilde{\mlvec{K}}$ is the parameterized Neural Tangent Kernel. In the limit of an infinite number of parameters, $\tilde{\mlvec{K}}$ tends to a constant limit, thus the asymptotic dynamics is kernel training.
    \item \textbf{Convergence of idealized dynamics under small perturbations.} The deviation of $\mlvec{Y}$ from its expectation cannot be zero in general, so we must establish that a small perturbation to RGF on the sphere maintains the property of linear convergence. The following lemma (analogous to Lemma F.1 in \cite{arora2019exact}) states that, if $\mlvec{Y}(t)$ is close to
    $\mlvec{Y}^{\star}$ through out the optimization, then the gradient flow is guaranteed to find the ground state efficiently
    \begin{restatable}[VQE Perturbation Lemma]{lemma}{vqeperturbation}
      \label{lm:vqe-perturbation}
      Conditioned on the event that the output state at initialization
      $\ket{\Psi(0)}$ has non-negligible overlap with the target ground state
      $\ket{\Psi^{\star}}$, such that $|\bra{\Psi(0)}\Psi^{\star}\>|^2\geq \Omega(\frac{1}{d})$, if for all $t\geq 0$,
      $\opnorm{\mlvec{Y}(t) - (\mlvec{W} - \frac{1}{d}\mlvec{I}_{d^2\times d^2})} \leq O(\frac{\lambda_{2} -\lambda_{1}}{\lambda_{d}-\lambda_{1}}\cdot \frac{1}{d})$,
    then under the dynamics
      $\frac{d}{dt}\ket{\Psi(t)} = -\tr_{1}(\mlvec{Y} ([\mlvec{M}, \newketbra{\Psi(t)}{\Psi(t)}]\otimes \mlvec{I}_{d\times d})) \ket{\Psi(t)}$,
    the output state will converge to the ground state as
      $1 - |\bra{\Psi(t)}\Psi^\star\>|^2\leq \exp(-c \frac{\lambda_{2}-\lambda_{1}}{\log d} t)$
      for some constant $c$.
    \end{restatable}
    \item \textbf{Concentration to idealized dynamics throughout training.} In order to show that the perturbations from RGF over the sphere are small, we leverage the concentration properties of $\mlvec{Y}$ arising from the large number of parameters used in order to bound the deviation from expectation by a quantity decreasing in $p$. We first show that concentration holds at initialization
    \begin{restatable} [Concentration at initialization for VQE]{lemma}{vqeconcentrationinit}
      \label{lm:vqe-concentration-init}
      Over the randomness of ansatz initialization (i.e. for
      $\{\mlvec{U}_{l}\}_{l=1}^{p}$ sampled $i.i.d.$ with respect to
      the Haar measure), for any initial $\mlvec\theta(0)$, with probability $1 - \delta$:
      \begin{align}
      \opnorm{\YY(\mlvec\theta(0)) - \YY^{\star}} \leq \frac{1}{\sqrt{p}}\cdot \frac{2\opnorm{\mlvec{H}}^{2}}{Z}\sqrt{\log\frac{d^{2}}{\delta}}.
      \end{align}
    \end{restatable}
    We further show that the concentration is maintained throughout the evolution of the dynamics as long as exponential convergence holds
    \begin{restatable}[Concentration during training (time dependent)]{lemma}{vqeconcentrationtime}
    \label{lm:vqe-concentration-time}
      Suppose that under learning rate $\eta=\frac{1}{p Z(\mlvec{H},d)}$, for all $0 \le t \le T$,
      $1 - |\braket{\Psi|\Psi^{\star}}|^{2} \leq \exp(-c\frac{(\lambda_{2}-\lambda_{1})}{\log d}t)$,
      then with probability $\geq 1-\delta$, for all $0\leq t\leq T$:
      \begin{align}
     \opnorm{\mlvec{Y}(\mlvec{\theta}(t)) - \mlvec{Y}(\mlvec{\theta}(0))} \leq C_3\left(\frac{T}{\sqrt{p}}\cdot\sqrt{2}(\lambda_{d}-\lambda_{1})\cdot\sqrt{\frac{d^{2}-1}{Z(\mlvec{H},d)^3}}\left(1 + \sqrt{\log\left(\frac{2d}{\delta}\right)}\right)\right),
      \end{align}
      where $C_3$ is a constant.
    \end{restatable}
    \item \textbf{(Main Result) Sufficient conditions for convergence.} The previously established conditions on concentration and convergence under perturbations are combined to yield a sufficient condition on the degree of over-parameterization required to ensure that a VQE converges to its ground state under gradient flow.
    \begin{restatable}[Exponential convergence for VQE]{theorem}{vqeconcentrationmain}
    \label{thm:vqe-convergence-full}
    Consider a VQE system in a $d$-dimensional Hilbert space (with architecture as described in Definition~\ref{def:frozen-ansatz}) with a target Hamiltonian $\mlvec{M}$ with eigenvalues $\lambda_1 \le \lambda_2 \dots \le \lambda_d$, and an ansatz $\ansatzname$ with a generating Hamiltonian $\mlvec{H}$. Let the number of parameters $p$ be greater than or equal to a threshold $p_{\mathrm{th}} = O\left(\left(\frac{\lambda_d - \lambda_1}{\lambda_2 - \lambda_1}\right)^4,\frac{d^4}{Z(\mlvec{H},d)^{3}},\log\left(d\right)\right)$, and the ground state of the system by $|\Psi^\star\>$. Then gradient flow training of the VQE system with a learning rate of $\eta = \frac{1}{pZ(\mlvec{H},d)}$, converges to the ground state with error $\epsilon = 1 - |\<\Psi(T_{\epsilon})|\Psi^\star\>|^2$ in time $T_\epsilon = O\big(\frac{\log d}{\lambda_2 - \lambda_1}\log{\frac{1}{\epsilon}}\big)$ with failure probability at most $0.99$. The success probability can be boosted to $1 - \delta$ for any $0 \le \delta \le 1$ using $O\left(\log\left(\frac{1}{\delta}\right)\right)$ repetitions, with the parameters randomly reinitialized each time.
    \end{restatable}
\end{enumerate}

\paragraph{Technical details and proofs} In the following sections we describe the main technical ideas behind the results outlined previously. The proof of convergence under perturbation, and concentration of initialization follow relatively well known techniques and are postponed to Section~\ref{sec:vqe-proof} in the appendix. The identification of VQE gradient flow (Lemma~\ref{lm:vqe-dynamics}) is proved in Section~\ref{sec:vqe-rgf}. The proof of concentration during training (Lemma~\ref{lm:vqe-concentration-time}) is in Section~\ref*{sec:vqe-concentration}, and the main theorem is proved in Section~\ref*{sec:vqe-main-proof}.

\subsection{Identify VQE with Reimannian gradient flow (RGF) over unit sphere}
\label{sec:vqe-rgf}
\vqedynamics*
\begin{proof}
  We start by calculating the gradient of $\mlvec{U}_{r:p}(\mlvec\theta)$ with
  respect to $\theta_{l}$. (1) For $r > l$, $\mlvec{U}_{r:p}$ is independent of
  $\theta_{l}$, therefore $\partial\mlvec{U}_{r:p}/\partial\theta_{l} = 0$. (2) For $r \leq l$,
  \begin{align}
    \frac{\partial \mlvec{U}_{r:p}}{\partial\theta_{l}}
    = \mlvec{U}_{l:p}(\mlvec{\theta}) (-i\mlvec{H})  \mlvec{U}_{r:l-1}(\mlvec\theta)
    = -i\mlvec{U}_{l:p}\mlvec{H}\mlvec{U}_{l:p}^{\dagger}\mlvec{U}_{r:p}.
  \end{align}
Therefore
  \begin{align}
    \frac{\partial L(\mlvec{\theta})}{\partial\theta_l}
    &=\bra{\Phi}\mlvec{U}^{\dagger}_{0}\frac{\partial}{\partial \theta_{l}}\mlvec{U}^{\dagger}_{1:p}\mlvec{M}\mlvec{U}_{1:p}\mlvec{U}_{0}\ket{\Phi}
      + \bra{\Phi}\mlvec{U}^{\dagger}_{0}\mlvec{U}^{\dagger}_{1:p}\mlvec{M}\frac{\partial}{\partial \theta_{l}}\mlvec{U}_{1:p}\mlvec{U}_{0}\ket{\Phi}\\
    &= \bra{\Phi}\mlvec{U}^{\dagger}_{0}\mlvec{U}^{\dagger}_{1:p}i[\mlvec{U}_{l:p}\mlvec{H}\mlvec{U}^{\dagger}_{l:p},\mlvec{M}]\mlvec{U}_{1:p}\mlvec{U}_{0}\ket{\Phi}\\
    &= \bra{\Psi(t)}i[\mlvec{U}_{l:p}\mlvec{H}\mlvec{U}^{\dagger}_{l:p},\mlvec{M}]\ket{\Psi(t)}\\
    &= i\tr([\mlvec{M}, \newketbra{\Psi(t)}{\Psi(t)}] \mlvec{U}_{l:p}\mlvec{H}\mlvec{U}^{\dagger}_{l:p}).
  \end{align}
  The third equality follows from that fact that
  $\mlvec{U}_{1:p}\mlvec{U}_{0}\ket{\Phi}$ is exactly the output state
  $\ket{\Psi(t)}$. Following the dynamics of gradient flow with learning rate $\eta$:
\begin{align}
  \frac{d\theta_l}{dt} = -\eta\frac{\partial}{\partial \theta_{l}}L(\mlvec\theta)
    = -i\eta\tr([\mlvec{M}, \newketbra{\Psi(t)}{\Psi(t)}] \mlvec{U}_{l:p}\mlvec{H}\mlvec{U}^{\dagger}_{l:p}).
\end{align}

The dynamics for $\mlvec{U}_{l:p}$ and $\ket{\Psi(t)}$ as functions of $\mlvec\theta(t)$ are therefore
\begin{align}
  \frac{d}{dt}\mlvec{U}_{l:p} = \sum_{r=l}^p\frac{d\theta_r}{dt}\frac{\partial}{\partial\theta_{r}}\mlvec{U}_{l:p}
                             = -\eta\sum_{r=l}^p\tr([\mlvec{M},\newketbra{\Psi(t)}{\Psi(t)}]\mlvec{U}_{l:p}\mlvec{H}\mlvec{U}^{\dagger}_{l:p})\mlvec{U}_{l:p}\mlvec{H}\mlvec{U}^{\dagger}_{l:p}\mlvec{U}_{l:p}
\end{align}
and
\begin{align}
  \frac{d}{dt}\ket{\Psi(t)} &= \frac{d}{dt}\big(\mlvec{U}_{1:p}\mlvec{U}_0\ket{\Phi}\big)\\
  &= -(\eta\cdot p Z)\frac{1}{p Z}\big(\sum_{l=1}^p\tr([\mlvec{M},\newketbra{\Psi(t)}{\Psi(t)}]\mlvec{U}_{l:p}\mlvec{H}\mlvec{U}^{\dagger}_{l:p})\mlvec{U}_{l:p}\mlvec{H}\mlvec{U}^{\dagger}_{l:p}\big)\mlvec{U}_{1:p}\mlvec{U}_0\ket{\Phi}\\
  &= - (\eta\cdot pZ)\tr_1(\mlvec{Y}  [\mlvec{M},\newketbra{\Psi(t)}{\Psi(t)}]]\otimes \mlvec{I})\ket{\Psi(t)}.
\end{align}
\end{proof}

\subsection{Concentration of dynamics from over-parameterization}
\label{sec:vqe-concentration}
In this section we wish to prove that $\mlvec{Y}$ concentrates to its expected value throughout training upto any point in time until which the linear convergence condition holds on the gradient flow dynamics. The proof will be based on two main ideas:
\begin{enumerate}
    \item The linear convergence of the gradient flow dynamics allows the deviation of the parameters $\mlvec{\theta}$ from their initial values to be bounded in terms of the evolution time (See Lemma~\ref{lm:vqe-slow-theta-new}).
    \item The random variables $\mlvec{Y}(\mlvec{\theta}(t))$ for different times $t$ form a \emph{random field}, whose deviations $\mlvec{Y}(\mlvec{\theta}(t_1)) - \mlvec{Y}(\mlvec{\theta}(t_1))$ we show to be bounded by a quantity proportional to $|t_1 - t_2|/\sqrt{p}$. We then use Dudley's lemma~\cite[Theorem~8.1.6]{vershynin2018high} on the concentration of random fields to bound the supremum of the deviation from initialization over time.
\end{enumerate}

We first show a result connecting the evolution time to the corresponding deviation in $\mlvec{\theta}$.
\begin{restatable}[Slow-varying $\theta$]{lemma}{vqeslowtheta-new}
  \label{lm:vqe-slow-theta-new}
  Suppose that under learning rate $\eta=\frac{1}{p Z(\mlvec{H},d)}$, for all $0 \le t \le T$,
  $1 - |\braket{\Psi(t)|\Psi^{\star}}|^{2} \leq \exp(-c\frac{(\lambda_{2}-\lambda_{1})}{\log d}t)$,
  then for all $0 \le t_1,t_2 \le T$:
  \begin{align}
    \|\mlvec{\theta}(t_2) - \mlvec{\theta}(t_1)\|_{\infty}
    &\leq \frac{1}{p}\cdot\sqrt{2}(\lambda_{d}-\lambda_{1})\cdot\sqrt{\frac{d^{2}-1}{Z(\mlvec{H},d)}}\cdot |t_2 - t_1|,\\
    \|\mlvec{\theta}(t_2) - \mlvec{\theta}(t_1)\|_{2}
    &\leq \frac{1}{\sqrt{p}}\cdot\sqrt{2}(\lambda_{d}-\lambda_{1})\cdot\sqrt{\frac{d^{2}-1}{Z(\mlvec{H},d)}}\cdot |t_2 - t_1|.
  \end{align}
\end{restatable}
\begin{proof}
Let $\mlvec{J}(t) := [\mlvec{M}, \newketbra{\Psi(t)}{\Psi(t)}]$ and $\mlvec{H}_{l}:=\mlvec{U}_{l:p}\mlvec{H}\mlvec{U}_{l:p}^{\dagger}$. Recall that
\begin{align}
  \frac{d\theta_l}{dt} = -\frac{1}{pZ(\mlvec{H},d)}\tr(i\mlvec{J}(t)\mlvec{H}_l(t)),\quad
\end{align}
\begin{align}
  |{\theta}_{l}(t_{2}) - {\theta}_{l}(t_{1})| &= |\int_{t_1}^{t_2}\frac{d{\theta}_{l}(t)}{dt} dt| =\frac{1}{pZ}|\int_{t_1}^{t_2}\tr(\mlvec{H}_{l}(t) \mlvec{J}(t)dt|\\
  &\leq\frac{1}{pZ}\|\mlvec{H}(t)\|_F\int_{t_1}^{t_2} \|\mlvec{J}(t)\|_F dt\\
  &\leq\frac{1}{pZ}\sqrt{\tr(\mlvec{H}^2)}\int_{t_1}^{t_2} \sqrt{2}(\lambda_{d}-\lambda_{1})e^{-\frac{c}{2}\frac{\lambda_{2}-\lambda_{1}}{\log d} t} dt\\
  &=2\sqrt{2}\frac{\lambda_{d}-\lambda_{1}}{\lambda_{2}-\lambda_{1}}\sqrt{\frac{d^2 - 1}{p^{2}Z}}(\log d /c) \left(e^{-\frac{c}{2}\frac{\lambda_{2}-\lambda_{1}}{\log d} t_1} - e^{-\frac{c}{2}\frac{\lambda_{2}-\lambda_{1}}{4\log d} t_2}\right)\\
  &\leq \frac{1}{p}\cdot\sqrt{2}(\lambda_{d}-\lambda_{1})\cdot\sqrt{\frac{d^{2}-1}{Z(\mlvec{H},d)}}\cdot |t_2 - t_1| \cdot e^{-\frac{c}{2}\frac{\lambda_{2}-\lambda_{1}}{\log d} t_1}\\
  &\leq \frac{1}{p}\cdot\sqrt{2}(\lambda_{d}-\lambda_{1})\cdot\sqrt{\frac{d^{2}-1}{Z(\mlvec{H},d)}}\cdot |t_2 - t_1|.
\end{align}
Here we use the fact that
$\|\mlvec{J}\|_{F} \leq \sqrt{2}(\lambda_{d}-\lambda_{1})\sqrt{1 - |\braket{\Psi(t)|\Psi^{\star}}|^{2}}$,
following technical Lemma~\ref{lm:technical_commutation2}.
\end{proof}

We next consider the random variables $\mlvec{Y}(\mlvec{\theta}(t))$ for $t$ in
some interval $[0,T]$. These variables form a random field and we show a
concentration inequality on the expected deviation in
$\mlvec{Y}(\mlvec{\theta}(t))$ over different intervals. A random variable
$\mlvec{X}$ is said to be \emph{sub-gaussian} \cite[Proposition
2.5.2]{vershynin2018high} if its tails satisfy
$\Prob[\mlvec{X} \ge t] \le 2 \exp\(-t^2/K_1^2\)$ for some $K_{1}$. The largest
$K_1$ satifying this relation is defined to be the second Orlicz norm, or $\psi_2$-norm of $\mlvec{X}$.

\begin{restatable}[Concentration of deviations in $\mlvec{Y}(\mlvec{\theta}(t))$]{lemma}{vqeydiff}
  \label{lm:vqe-expected-deviation}
  \begin{align}
    \Prob[\opnorm{\mlvec{Y}(\mlvec{\theta}(t_2)) - \mlvec{Y}(\mlvec{\theta}(t_1))} > t] \le 2\exp\left(-\frac{-t^2 Z(\mlvec{H},d)^2}{2C_1\|\mlvec{\theta}(t_2) - \mlvec{\theta}(t_1)\|_{2}^2}\right)
\end{align}
for some constant $C_1$.
\end{restatable}
\begin{proof}
We first observe that due to the Haar distribution of the unitaries $U_{l}$, $\mlvec{Y}(\mlvec{\theta}(t_2)) - \mlvec{Y}(\mlvec{\theta}(t_1))$ is distributed identically to $\mlvec{Y}(\mlvec{\theta}(t_2) - \mlvec{\theta}(t_1)) - \mlvec{Y}(0)$. For convenience we define $\delta\mlvec{\theta} = \mlvec{\theta}(t_2) - \mlvec{\theta}(t_1)$ in the remainder of the proof.

Define $\mlvec{Y}_l(\mlvec{\theta}) = \mlvec{H}_l^{\otimes 2}$; then $\mlvec{Y}(\mlvec{\theta}) = \frac{1}{pZ(\mlvec{H},d)}\sum_{l=1}^{p} \mlvec{Y}_l$. 
We consider a re-parameterization of the random variables $\mlvec{H}_l(\theta)$ by constructing random variables that are identically distributed, but are functions on a different latent probability space. Defining $\mlvec{H}_{l}$ as
  $\mlvec{U}_{p}\cdots \mlvec{U}_{l}\mlvec{H}\mlvec{U}^{\dagger}_{l}\cdots\mlvec{U}^{\dagger}_{p}$,
  $\mlvec{Y}$ can be rewritten as:
  \begin{align}
    \mlvec{Y}(\mlvec\theta) = \frac{1}{pZ} \sum_{l=1}^{p}
    \big(
    e^{-i\theta_{p}\mlvec{H}_{p}}\cdots e^{-i\theta_{l+1}\mlvec{H}_{l+1}}
    \mlvec{H}_{l}
    e^{i\theta_{l+1}\mlvec{H}_{l+1}}\cdots e^{i\theta_{p}\mlvec{H}_{p}}
    \big)^{\otimes 2}.
  \end{align}
  By the Haar randomness of $\{\mlvec{U}_{l}\}_{l=1}^{p}$, we can view
  $\{\mlvec{H}_{l}\}_{l=1}^{p}$ as random Hermitians generated by
  $\{\mlvec{V}_{l}\mlvec{H}\mlvec{V}^{\dagger}_{l}\}$ for \textit{i.i.d.} Haar
  random $\{\mlvec{V}_{l}\}_{l=1}^{p}$. This variable is identically distributed to $\mlvec{Y}$ and $\mlvec{Y}_{l}$ can be defined as each term in the sum.

We will apply the well-known McDiarmid inequality~\cite[Theorem~2.9.1]{vershynin2018high} that can be stated as follows: Consider independent random variables $X_1,\dots,X_k \in \mathcal{X}$. Suppose a random variable $\phi \colon \mathcal{X}^{k} \to \real$ satisfies the condition that for all $1 \le j \le k$ and for all $x_1,\dots,x_j,\dots,x_k,x'_j \in \mathcal{X}$,
\begin{align}
    |\phi(x_1,\dots,x_j,\dots,x_k) - \phi(x_1,\dots,x'_j,\dots,x_k)| \le c_j,
\end{align}
then the tails of the distribution satisfy
\begin{align}
    \Prob[|\phi(X_1,\dots,X_k) - \EXP\phi| \ge t] \le \exp\left(\frac{-2t^2}{\sum_{i=1}^{k}c_i^2}\right).
\end{align}

With our earlier re-parameterization we can consider $\mlvec{Y}$ and consequently $\mlvec{Y}_l$ as functions of the randomly sampled Hermitian operators $\mlvec{H}_l$. Define the variable $\mlvec{Y}^{(k)}$ as that obtained by resampling $\mlvec{H}_k$ independently, and $\mlvec{Y}_l^{(k)}$ correspondingly. Finally we define
\begin{align}
\Delta^{(k)} \mlvec{Y} = \opnorm{(\mlvec{Y}(\delta\mlvec{\theta}) - \mlvec{Y}(0)) - (\mlvec{Y}^{(k)}(\delta\mlvec{\theta}) - \mlvec{Y}^{(k)}(0))} = \opnorm{\mlvec{Y}(\delta\mlvec{\theta}) - \mlvec{Y}^{(k)}(\delta\mlvec{\theta})}.
\end{align}
Via the triangle inequality,
\begin{align}
    \Delta^{(k)} \mlvec{Y} &= \lVert \mlvec{Y}(\delta\mlvec{\theta}) - \mlvec{Y}^{(k)}(\delta\mlvec{\theta}) \rVert
    = \frac{1}{pZ}\lVert \sum_{l=1}^k \mlvec{Y}_{l}(\delta\mlvec{\theta}) - \mlvec{Y}_{l}^{(k)}(\delta\mlvec{\theta}) \rVert \\
    &\le \frac{1}{pZ} \sum_{l=1}^k \lVert \mlvec{Y}_{l}(\delta\mlvec{\theta}) - \mlvec{Y}_{l}^{(k)}(\delta\mlvec{\theta}) \rVert.
\end{align}
Then by definition,
\begin{align}
    &\lVert\mlvec{Y}_{l}(\delta\mlvec{\theta}) - \mlvec{Y}_{l}^{(k)}(\delta\mlvec{\theta})\rVert \nonumber\\
    =& (e^{-i\delta\mlvec\theta_p\mlvec{H}_p}\cdots e^{-i\delta\mlvec\theta_{k+1}\mlvec{H}_{k+1}})^{\otimes 2}
    \big((e^{-i\delta\mlvec{\theta}_{k}\mlvec{H}_{k}} \mlvec{K} e^{i\delta\mlvec{\theta}_{k}\mlvec{H}_{k}})^{\otimes 2}\nonumber\\
    -&(e^{-i\delta\mlvec{\theta}_{k}\mlvec{H}^{\prime}_{k}} \mlvec{K} e^{i\delta\mlvec{\theta}_{k}\mlvec{H}^{\prime}_{k}})^{\otimes 2}
     \big)
     (e^{i\delta\mlvec{\theta}_{l+1}\mlvec{H}_{l+1}}\cdots e^{i\delta\mlvec{\theta}_{p}\mlvec{H}_{p}})^{\otimes 2}.
\end{align}
where $\mlvec{K}:=
    e^{-i\delta\mlvec{\theta}_{k-1}\mlvec{H}_{k-1}}\cdots e^{-i\delta\mlvec{\theta}_{l+1}\mlvec{H}_{l+1}}
    \mlvec{H}_{l}
    e^{i\delta\mlvec{\theta}_{l+1}\mlvec{H}_{l+1}}\cdots e^{i\delta\mlvec{\theta}_{k-1}\mlvec{H}_{k-1}}$.
By Lemma~\ref{lm:helper-taylor},
\begin{align}
    \|\big(\mlvec{Y}_{l}(\delta\mlvec{\theta}) - \mlvec{Y}_{l}(\mlvec{0})\big)
  - \big(\mlvec{Y}^{(k)}_{l}(\delta\mlvec{\theta}) - \mlvec{Y}^{(k)}_{l}(\mlvec{0})\big)\|
  \leq 8\abs{\delta\mlvec{\theta}_{k}}\opnorm{\mlvec{H}}\opnorm{\mlvec{K}}^{2}
  = 8\abs{\delta\mlvec{\theta}_{k}}\opnorm{\mlvec{H}}^{3}.
\end{align}
We finally have $\Delta^{(k)}(y) \le \frac{8|\delta\mlvec{\theta}_k|\opnorm{\mlvec{H}}^{3}}{Z}$. By the McDiarmid inequality, the result follows.
\end{proof}

To bound the supremum of the deviation over an entire time interval, we employ Dudley's integral inequality (stated below in it's matrix form).
\begin{lemma}[Dudley's integral inequality: subgaussian matrix version (Adapted from Theorem 8.1.6 in \cite{vershynin2018high}]
\label{lm:dudley-subgaussian-matrix}
Let $\mlvec{\mathcal{R}}$ be a metric space equipped with a metric $\mathbf{d}(\cdot,\cdot)$, and $\mlvec{X}: \mlvec{\mathcal{R}} \mapsto \real^{D \times D}$ with subgaussian increments ie. it satisfies
\begin{align}
    \Prob[\| \mlvec{X}(r_1) - \mlvec{X}(r_2)\|_{\mathsf{op}} > t] \le 2D\exp\left(-\frac{t^2}{C_{\sigma}^2\mathbf{d}(r_1,r_2)^2}\right),
\end{align}
Then with probability at least $1 - 2D\exp(-u^2)$ for any subset $\mathcal{S} \subseteq \mlvec{\mathcal{R}}$:
\begin{align}
    \sup_{(r_1,r_2) \in \mathcal{S}} \|\mlvec{X}(r_1) - \mlvec{X}(r_2)\|_{\mathsf{op}} \le C \cdot C_{\sigma} \left[\int_0^{\mathrm{diam}(\mathcal{S})} \sqrt{\mathcal{N}(\mathcal{S},\mathbf{d},\epsilon)}\,d\epsilon + u \cdot \mathrm{diam}(\mathcal{S})\right].
\end{align}
for some constant $C$, where $\mathcal{N}(\mathcal{S},\mathbf{d},\epsilon)$ is the metric entropy defined as the logarithm of the $\epsilon$-covering number of $\mathcal{S}$ using metric $d$.
\end{lemma}

We then have the following main result:
\vqeconcentrationtime*
\begin{proof}
Via Lemma~\ref{lm:vqe-expected-deviation},
\begin{align}
    \Prob[\lVert\mlvec{Y}(\mlvec{\theta}(t_2)) - \mlvec{Y}(\mlvec{\theta}(t_1))\rVert_{\mathsf{op}} > t] \le 2\exp\left(-\frac{-t^2 Z(\mlvec{H},d)^{2}}{2C_1\|\mlvec{\theta}(t_2) - \mlvec{\theta}(t_1)\|^2}\right),
\end{align}
By Lemma~\ref{lm:vqe-slow-theta-new} $\|\mlvec{\theta}(t_2) - \mlvec{\theta}(t_1)\|_2 \le \frac{1}{\sqrt{p}}\cdot\sqrt{2}(\lambda_{d}-\lambda_{1})\cdot\sqrt{\frac{d^{2}-1}{Z(\mlvec{H},d)}}\cdot |t_2 - t_1|$. Thus, $\mlvec{Y}$ has sub-gaussian increments if we define the metric $\mathbf{d}(t_2,t_1) = \frac{1}{\sqrt{p}}\cdot\sqrt{2}(\lambda_{d}-\lambda_{1})\cdot\sqrt{\frac{d^{2}-1}{Z(\mlvec{H},d)^3}}\cdot |t_2 - t_1|$, thereby satisfying the conditions for Lemma~\ref{lm:dudley-subgaussian-matrix}. Under this metric, the diameter of the interval $[0,T]$ is $\frac{T}{\sqrt{p}}\cdot\sqrt{2}(\lambda_{d}-\lambda_{1})\cdot\sqrt{\frac{d^{2}-1}{Z(\mlvec{H},d)^3}}$. Applying Lemma~\ref{lm:dudley-subgaussian-matrix}, with $u = \sqrt{\log(2d/\delta)}$ to ensure a failure probability at most $\delta$ we have
\begin{align}
    \sup_{t \in [0,T]} \|\mlvec{Y}(\mlvec{\theta}(t)) - \mlvec{Y}(\mlvec{\theta}(0))\|_{\mathsf{op}} \le C_2\left(\int_{0}^{\mathrm{diam}([0,T])}\epsilon^{-1}\,d\epsilon + \mathrm{diam}([0,T])\sqrt{\log\left(\frac{2d}{\delta}\right)}\right).
\end{align}
We assume without loss of generality that $p$ is large enough such that $\mathrm{diam}([0,T]) < 1$. Then,
\begin{align}
    \sup_{t \in [0,T]} \|\mlvec{Y}(\theta(t)) - \mlvec{Y}(\theta(0))\|_{\mathsf{op}} \le C_2\left(\mathrm{diam}([0,T])\left(1 + \sqrt{\log\left(\frac{2d}{\delta}\right)}\right)\right).
\end{align}
By the previous consideration,
\begin{align}
    \sup_{t \in [0,T]} \|\mlvec{Y}(\theta(t)) - \mlvec{Y}(\theta(0))\|_{\mathsf{op}} \le C_3\left(\frac{T}{\sqrt{p}}\cdot\sqrt{2}(\lambda_{d}-\lambda_{1})\cdot\sqrt{\frac{d^{2}-1}{Z(\mlvec{H},d)^3}}\left(1 + \sqrt{\log\left(\frac{2d}{\delta}\right)}\right)\right),
\end{align}
where $C_2,C_3$ are constants.
\end{proof}

\subsection{Linear convergence to the ground state}
\label{sec:vqe-main-proof}
Finally, we can combine our previous results to show that with sufficient over-parameterization, the VQE dynamics can be made to exponentially converge to the ground state

\vqeconcentrationmain*
\begin{proof}
  We first show that the initial output state $\ket{\Psi(0)}$ satisfies the condition $|\<\Psi(0)|\Psi^{\star}\>|^{2} \ge \Omega(1/d)$, required by Lemma~\ref{lm:vqe-perturbation}, with high probability. To see this, observe that $\ket{\Psi(0)}$ is obtained by applying a Haar uniform unitary operator to an input vector, therefore, $\ket{\Psi(0)}$ obeys the uniform Haar distribution over the space of quantum states $S(\complex^{d})$. Due to this uniformity, $|\<\Psi(0)|\Psi^{\star}\>|^{2}$ is equidistributed to $|\<\Psi(0)|1\>|^{2} $. Furthermore, $|\Psi(0)\>$ follows the same distribution as a state vector $|w\> = \frac{1}{\sum_{j=1}^{d} w_{j,\mathrm{re}}^{2} + w_{j,\mathrm{im}}^{2} }\left(\sum_{k=1}^{d}w_{k,\mathrm{re}} + iw_{k,\mathrm{im}}\right)$, where each $w_{j,\mathrm{im}},w_{j,\mathrm{re}}$ are drawn from independent standard normal distributions. The distribution of $|\<\Psi(0)|\Psi^{\star}\>|^{2}$ is therefore identical to that of the quantity $\frac{w_{1,\mathrm{re}}^{2}}{\sum_{j=1}^{d} w_{j,\mathrm{re}}^{2} + w_{j,\mathrm{im}}^{2} }$. By standard concentration of normal variables, the numerator is $\Theta(1)$ and the denominator is $\Theta(d)$, with any constant failure probability. Choosing both the failure probabilities to be $0.0025$, we have that the condition $|\<\Psi(0)|\Psi^{\star}\>|^{2} \ge \Omega(1/d)$ is satisfied with probability at least $0.995$.

  Once the above condition is satisfied, Lemma~\ref{lm:vqe-perturbation} states
  that if the closeness condition on $\mlvec{Y}$ is maintained for time
  $T_\epsilon = \frac{1}{c}\frac{\log d}{\lambda_2 - \lambda_1}\log\left(\frac{1}{\epsilon}\right)$
  the obtained error is less than or equal to $\epsilon$. Therefore, by
  Lemma~\ref{lm:vqe-concentration-time} and
  Lemma~\ref{lm:vqe-concentration-init}, in order to ensure  with failure
  probability at most $0.005$, that
  $\|\mlvec{Y}(t) - \mlvec{Y}(0)\| \leq \frac{C_{0}}{d} \cdot\frac{\lambda_2 - \lambda_1}{\lambda_d - \lambda_1}$
  for constant $C_{0}$
  up to any time $t$ such that $0 < t \le T_\epsilon$ and $1 - |\braket{\Psi|\Psi^{\star}}|^{2} \leq \exp(-c\frac{(\lambda_{2}-\lambda_{1})}{\log d}t')$ for all $t' \le t$, it suffices to choose $p$ such that
  \begin{align}
    C_3\left(\frac{T_{\epsilon}}{\sqrt{p}}\cdot\sqrt{2}(\lambda_{d}-\lambda_{1})\cdot\sqrt{\frac{d^{2}-1}{Z(\mlvec{H},d)^3}}\left(1 + \sqrt{\log\left(\frac{2d}{0.005}\right)}\right)\right) \le \frac{\lambda_{2} -\lambda_{1}}{\lambda_{d}-\lambda_{1}}\cdot \frac{C_{0}}{d}.
  \end{align}
  By simple algebra, it can be verified from the above that it suffices to choose any $p$ greater than or equal to some over-parameterization threshold $p_{\mathrm{th}}$, where $p_{\mathrm{th}} = O\left(\left(\frac{\lambda_d - \lambda_1}{\lambda_2 - \lambda_1}\right)^4,\frac{d^4}{Z(\mlvec{H},d)^3},\log\left(d\right)\right)$.

The above argument shows that if the dynamics exhibits linear convergence upto some time $t$, the closeness condition on $\mlvec{Y}$ will also be satisified with failure probability at most $\delta$, if the number of parameters is chosen appropriately. We now argue that \emph{both} these conditions must hold until time $T_{\epsilon}$.
Let $t_0$ be the minimum time such that either $1 - |\braket{\Psi_{t_0}|\Psi^{\star}}|^{2} > \exp(-c\frac{(\lambda_{2}-\lambda_{1})}{\log d}t_0)$ or $\|\mlvec{Y}(t_0) - \mlvec{Y}(0)\| > \frac{\lambda_2 - \lambda_1}{\lambda_d - \lambda_1}\cdot\frac{C_{0}}{d}$. If $1 - |\braket{\Psi|\Psi^{\star}}|^{2} > \exp(-c\frac{(\lambda_{2}-\lambda_{1})}{\log d}t_0)$, we must have $\|\mlvec{Y}(t'_0) - \mlvec{Y}(0)\| > \frac{\lambda_2 - \lambda_1}{\lambda_d - \lambda_1}\cdot\frac{C_{0}}{d}$ at some earlier time $t'_0$ (Lemma~\ref{lm:vqe-perturbation}). Similarly, if $\|\mlvec{Y}(t_0) - \mlvec{Y}(0)\| > \frac{\lambda_2 - \lambda_1}{\lambda_d - \lambda_1}\cdot\frac{C_{0}}{d}$, we must have $1 - |\braket{\Psi(t'_0)|\Psi^{\star}}|^{2} > \exp(-c\frac{(\lambda_{2}-\lambda_{1})}{\log d}t'_0)$ at some earlier time $t'_0$ (Lemma~\ref{lm:vqe-concentration-time}). Therefore by contradiction, both conditions must be realized for all times $t \le T_\epsilon$, yielding the result.
\end{proof}


\section{Convergence under noisy gradient}
\label{sec:robustness}
So far we have assumed perfect access to the exact gradient
$\nabla L = (\frac{\partial L}{\partial\theta_{1}}, \frac{\partial L}{\partial\theta_{2}}, \cdots, \frac{\partial L}{\partial\theta_{p}})^{T}$.
In practical NISQ settings, the estimation of gradients is noisy either due to
the finite number of measurements, or to the noisy implementation of circuits.
In this section, we extend the convergence theorem and show
that for sufficiently small amount of noise, the efficient convergence remains.
We comment that, while the noise level required in the following theorem depends
polynomially on $1/d$ and is still not practical for NISQ settings, our result is
the first rigorous result to establish the convergence of VQE in the noisy
setting. In addition, the result in this section suggest that the convergence
theorem is robust, and reveal the dependency of the noise level on the approximation
error $1 - |\braket{\Psi(t)|\Psi^{\star}}|^{2}$.

In the continuous-time setting we consider the following definition for noisy gradient flow:
\begin{definition}[Noisy gradient flow] For loss function
  $L: \real^{p}\rightarrow \real$, the noisy gradient flow on the parameters
  $\mlvec\theta\in\real^{p}$ with learning rate $\eta$ is defined as
  \begin{align}
  \frac{d\mlvec\theta}{dt} = -\eta (\nabla L + \noisevec) \quad \text{or}\quad\frac{d\theta_{l}}{dt} = -\eta (\frac{\partial L}{\partial \theta_{l}} + \noise{l}(t))\quad \forall l\in[p]
  \end{align}
  where $\noisevec(t):=(\noise{1}(t), \cdots, \noise{p}(t))^{T}$ is the noise to the
gradient estimation.
\end{definition}
The following noisy version of the convergence theorem states that when
the $\ell_{\infty}$-norm of $\noisevec(t)$ is sufficiently small, the convergence result still holds:
\begin{restatable}[Convergence theorem with noisy gradient]{corollary}{noisyconvergence}
\label{cor:noisy-convergence}
    Consider training a $p$-parameter $d$-dimensional VQE instance
    $(\mlvec{M}, \ket{\Phi}, \mlvec{U})$ with learning rate
    $\eta = \frac{1}{pZ(\mlvec{H}, d)}$, where the ansatz $\mlvec{U}$ is generated by $\mlvec{H}$ as described in
    Definition~\ref{def:frozen-ansatz}. Let
    $\lambda_{1} < \lambda_{2} \leq \cdots \leq \lambda_{d}$ denote the eigenvalues of
    $\mlvec{M}$, and $\ket{\Psi^{\star}}$ denote the ground state.
    If
    \begin{itemize}
      \item the number of parameters $p$ greater than a threshold of order
    $O\big((\frac{\lambda_d - \lambda_1}{\lambda_2 - \lambda_1})^4,\frac{d^4}{Z(\mlvec{H},d)^{3}},\log(d)\big)$,
      \item the gradient estimation error
      $\|\noisevec(t)\|_{\infty} \leq c'\cdot \frac{Z}{\opnorm{\mlvec{H}}}(\lambda_{2} -\lambda_{1})\sqrt{1 - |\bra{\Psi(t)}\Psi^\star\>|^2} |\bra{\Psi(t)}\Psi^\star\>|$,
            for some constant $c'$,
    \end{itemize}
    then with probability $\geq 0.99$, the output state $\ket{\Psi(t)}$ converges under noisy gradient to the
      ground state with error $\epsilon := 1 - |\bra{\Psi(T_{\epsilon})}\Psi^\star\>|^2$ in time
      $T_\epsilon = O\big(\frac{\log d}{\lambda_2 - \lambda_1}\log\frac{1}{\epsilon}\big)$. The success probability can be boosted to $1 - \delta$ for any $0 < \delta < 1$ using $O\left(\log\left(\frac{1}{\delta}\right)\right)$ random restarts.
\end{restatable}
\begin{remark}
To interpret the upper bound on $\norm{\noisevec}_{\infty}$, notice that
\begin{align}
\sqrt{1 - |\bra{\Psi(t)}\Psi^\star\>|^2} |\bra{\Psi(t)}\Psi^\star\>| \leq \max\{ |\bra{\Psi(t)}\Psi^\star\>|^{2}, 1-|\bra{\Psi(t)}\Psi^\star\>|^{2}\}.
\end{align}
At the initial stage of training, $\norm{\noisevec}_{\infty}$ need to be
$O(|\bra{\Psi(t)}\Psi^\star\>|^{2})$ so that the worst-case perturbation in the
gradient does not eliminate the overlap between $\ket{\Psi(t)}$ and
$\ket{\Psi^{\star}}$; at the final stage of training $\norm{\noisevec}_{\infty}$ need to be
$O(1-|\bra{\Psi(t)}\Psi^\star\>|^{2})$ to obtain solutions with high quality.
\end{remark}
\begin{remark}
  The premise of Corollary~\ref{cor:noisy-convergence} requires
  $\|\noisevec(t)\|_{\infty} / \opnorm{\mlvec{H}}$ to be of order
  $Z / \opnorm{\mlvec{H}}^{2}$, which depends polynomially on $1/d$. We
 highlight that our analysis here considers the worst-case (or adversarial)
perturbation on the gradient. It is possible that the requirement on
$\norm{\noisevec}_{\infty}$ can be further relaxed in the practical scenerio.
For example, when the noise is purely due to the finite measurements, we can
further assume $\noisevec$ to be stochastic and unbiased.
\end{remark}
The proof of Corollary~\ref{cor:noisy-convergence} follows directly from the following
Lemma~\ref{lm:noisy-vqe-dynamics}, which calculates the dynamics at the presence of
gradient noise, and Lemma~\ref{lm:noisy-vqe-perturbation}, which states the
convergence of the noisy dynamics.
The proofs for the lemmas are based on the proofs for Lemma~\ref{lm:vqe-dynamics} and \ref{lm:vqe-perturbation} and are postponed
to Section~\ref{sec:app_noise}.

\begin{restatable}[Output-state dynamics with noisy gradient
  estimation]{lemma}{noisyvqedynamics}
  \label{lm:noisy-vqe-dynamics}
  Consider VQE instance $(\mlvec{M}, \ket{\Phi}, \mlvec{U})$, with $\mlvec{U}$
  being the ansatz defined in Definition~\ref{def:frozen-ansatz}.
  Under gradient flow with learning rate $\eta$ and noisy gradient estimation
  $\nabla L + \noisevec(t) = \big(\frac{\partial L}{\partial \theta_{l}} + \noise{l}(t)\big)_{l\in[p]}$,
  the output state $\ket{\Psi(t)}$ follow the dynamics
  \begin{align}
    \label{eq:noisy-vqe-dynamics}
    \frac{d}{dt}\ket{\Psi(t)} &= -(\eta\cdot p\cdot Z(\mlvec{H},d))\tr_{1}(\mlvec{Y}([\mlvec{M}, \newketbra{\Psi(t)}{\Psi(t)}]\otimes \mlvec{I}_{d\times d})) \ket{\Psi(t)} + \eta \sum_{l=1}^{p}i\noise{l}\mlvec{H}_{l}\ket{\Psi(t)}.
  \end{align}
  Here
  $\mlvec{H}_{l}$ are function of $\mlvec\theta(t)$, defined as
  $\mlvec{U}_{l:p}(\mlvec\theta(t))\mlvec{H}\mlvec{U}^{\dagger}_{l:p}(\mlvec\theta(t))$
  for all $l\in[p]$, and
  $\mlvec{Y}$ is
  defined as $\frac{1}{pZ(\mlvec{H},d)}\sum_{l=1}^{p}\mlvec{H}_{l}^{\otimes 2}$.
\end{restatable}

The following modified version of Lemma~\ref{lm:vqe-perturbation} implies that the
main theorem holds with noisy gradient estimation:
\begin{restatable}[VQE perturbation lemma under noisy gradients]{lemma}{noisyvqeperturbation}
    \label{lm:noisy-vqe-perturbation}
    If
    \begin{itemize}
    \item the output state at initialization $\ket{\Psi(0)}$ has non-negligible overlap with the ground state
            $\ket{\Psi^{\star}}$: $|\bra{\Psi(0)}\Psi^{\star}\>|^2\geq \Omega(\frac{1}{d})$,
    \item for all $0\leq t\leq T$,
            $\opnorm{\mlvec{Y}(t) - \mlvec{Y}^{\star}(t)} \leq O\big(\frac{\lambda_{2} -\lambda_{1}}{\lambda_{d}-\lambda_{1}}\cdot \frac{1}{d}\big)$,
    \item for all $0\leq t\leq T$,
    $\|\noisevec(t)\|_{\infty} \leq c'\frac{Z}{\opnorm{\mlvec{H}}}(\lambda_{2} -\lambda_{1})\sqrt{1 - |\bra{\Psi(t)}\Psi^\star\>|^2} |\bra{\Psi(t)}\Psi^\star\>|$
            for some positive constant $c'$,
    \end{itemize}
    then under the dynamics
    \begin{align}
      \frac{d}{dt}\ket{\Psi(t)} = -\tr_{1}(\mlvec{Y} ([\mlvec{M}, \newketbra{\Psi(t)}{\Psi(t)}]\otimes \mlvec{I}_{d\times d})) \ket{\Psi(t)} + \frac{1}{pZ}\sum_{l=1}^{p}i\noise{l}\mlvec{H}_{l}\ket{\Psi(t)},
    \end{align}
the output states converges to the ground state such that for all $0\leq t\leq T$:
\begin{align}
    1 - |\bra{\Psi(t)}\Psi^\star\>|^2\leq \exp(-c\frac{\lambda_{2}-\lambda_{1}}{\log d} t),\ \text{for
  some constant $c$.}
\end{align}
\end{restatable}


\section{Ansatz-dependent convergence theorem}
\label{sec:subgroup}
Theorem~\ref{thm:vqe-convergence-full} provides a sufficient condition on the
number of classical parameters to ensure a VQE instance to converge
with a linear rate. The bound depends on the system dimension $d$ as well as the
spectral ratio $\kappa =\frac{\lambda_d - \lambda_1}{\lambda_2 - \lambda_1}$. In
this section, we develop a tighter ansatz-dependent bound.

In Section~\ref{subsec:ansatz-dependent-bound}, we provide an ansatz-dependent
bound on the over-parameterization  threshold. We connect the
over-parameterization threshold to
the effective dimension $\deff$ and effective spectral ratio $\kappaeff$, two
ansatz-dependent quantities that will
defined later in Definition~\ref{def:effective-quantities}. As shown later in
Section~\ref{sec:exp_threshold}, $\deff$ and $\kappaeff$ can be significantly
smaller than $d$ and $\kappa$ for physical problems with good ansatz designs.

In Section~\ref{subsec:estimate-eff}, we describe a procedure for estimating
$\deff$ and $\kappaeff$ given an ansatz design. This leads to a principled
method for evaluating the performance of a given ansatz without repeatedly
performing gradient descent. In contrast, as the convergence of VQE training
depends highly on initializations, most existing empirical studies evaluate and
compare the ansatz performance by repeatedly optimizing the VQE instance over random
initializations, even for toy-sized VQE problems with small $d$.

\subsection{Ansatz-dependent upper bound}
\label{subsec:ansatz-dependent-bound}
Given an ansatz design $\ansatzname$, recall that $G_{\ansatzname}$ is a subgroup of
$SU(d)$ associated with $\ansatzname$ defined in Section~\ref{sec:prelim},
containing all the realizable unitary matrices by $\ansatzname$ with varying number of
layers $L=0, 1, 2, \cdots$. Fixing the input state $\ket{\Phi}$, if $G_{\ansatzname}$
is a proper subgroup of $SU(d)$, output state $\ket{\Psi} = \mlvec{U}\ket{\Phi}$
is restricted to a subspace of $\complex^{d}$, leading to a tighter bound on the
number of parameters for convergence. We now formalize the intuition
using the group theory language.

A finite-dimensional representation $(W, \grouphomo)$ of a group $G$ is
specified by a vector space $W$ and a group homomorphism
$\grouphomo \colon G \to GL(W)$, such that
$\grouphomo(g_1)\grouphomo(g_2) = \grouphomo(g_1g_2)$ for all $g_1,g_2 \in G$.
And the representation is said to be \emph{unitary} if $\grouphomo(g)$ is unitary for all
$g \in G$. Apparently, as $G_{\ansatzname}$ are composed of unitary matrices, the identity map
furnishes a unitary representation of $G_{\ansatzname}$ (which we will refer to
as the \emph{natural representation}).

An important concept in the group representation theory is
\emph{irreducibility}. Given a representation $(W, \grouphomo)$ of
$G_{\ansatzname}$, a subspace $V \subseteq W$ is
said to be \emph{invariant} if $\grouphomo(g)v \in V$ for all $v \in V$ and
$g \in G_{\ansatzname}$. A representation is further said to be
\emph{irreducible} if it has no invariant subspaces other than the trivial
subspaces consisting of the empty set $\emptyset$ and the whole space $W$. We
are especially interested in the setting where $G_{\ansatzname}$ is reducible,
as the reducibility induces a decomposition of the ambient space $\mathcal{H} = \complex^{d}$:
\begin{proposition}[Adapted from {\cite[Proposition 4.27]{hall2003lie}}]
\label{prop:unitary-reducibility}
Let $G$ be a group with \emph{unitary} representation $\Pi$ acting on a vector space $W$. Then this representation is \emph{completely reducible} i.e. $W$ is isomorphic to a direct sum $V_1 \oplus \dots \oplus V_m$ where each $V_{j}$ is an invariant subspace which itself has no non-trivial invariant subspaces.
\end{proposition}
By Proposition~\ref{prop:unitary-reducibility}, the natural representation of $G_{\ansatzname}$
induces a decomposition of the state space
$\mathcal{H} = V_1\oplus\dots\oplus V_{m}$.
We now define the ansatz compatibility and the key quantities $\deff$ and $\kappaeff$ for a VQE instance
$(\mlvec{M}, \ket{\Phi}, \mlvec{U})$ using this decomposition.
\begin{definition}[Compatibility of ansatz]
  \label{def:ansatz-compatibility}
  Consider a VQE instance $(\mlvec{M}, \ket{\Phi}, \mlvec{U})$ with ansatz
  design $\ansatzname$.
  Let $\mathcal{H} = V_1\oplus\dots\oplus V_{m}$ be the completely-reduced
  decomposition induced by the ansatz design $\ansatzname$ through the natural
  representation of $G_{\ansatzname}$ and let $\ket{\Psi^{\star}}$ denote the
  ground state of $\mlvec{M}$. The ansatz design $\ansatzname$ is
  said to be compatible with the VQE problem if there
  exists $j\in[m]$ such that both the input state $\ket{\Phi}$ and the target
  ground state $\ket{\Psi^{\star}}$ lie within the invariant subspace $V_{j}$.
  We will drop the subscript $j$ and refer to this subspace as $V$ when there
  is no ambiguity.
\end{definition}
As we see in Section~\ref{subsec:estimate-eff} and later in
Section~\ref{sec:exp_threshold}, all the combinations and physical problems
(including Kitaev model, Ising model and Heisenberg model) we
examine are compatible.

The effective quantities for compatible ansatz can be defined using the
invariant subspace:
\begin{definition}[Effective dimension $\deff$ and effective ratio $\kappaeff$]
  \label{def:effective-quantities}
  Consider a VQE instance $(\mlvec{M}, \ket{\Phi}, \mlvec{U})$ with compatible ansatz
  design $\ansatzname$. And let $V$ denote the invariant subspace where the input and
  the ground state lies with projection
  $\mlvec{\Pi}=\subspacecolumn\subspacecolumn^{\dagger}$ (here
  $\subspacecolumn\in\complex^{d\times \deff}$ is an arbitrary set of
  orthonormal basis).
  The effective dimension $\deff$ is defined as the dimension of $V$.
  The effective spectrum is defined as the ordered eigenvalues $(\lambda_{1}',\cdots,\lambda_{\deff}')$ of the Hermitian
  $\subspacecolumn^{\dagger}\mlvec{M}\subspacecolumn$.
  The effective spectral ratio $\kappaeff$ is defined as $\frac{\lambda'_{\deff} - \lambda'_{1}}{\lambda'_{2}-\lambda'_{1}}$.
  The effective generating Hamiltonian $\Heff$ is defined as $\subspacecolumn^{\dagger}\mlvec{H}\subspacecolumn$.
\end{definition}
Given the projection $\mlvec{\Pi}$ onto $V$, the basis $\subspacecolumn$
is not unique, but allow a $\deff\times \deff$ unitary transformation. This does
not introduce any ambiguity in the definition of the effective spectrum as
unitary transformations does not change the eigenvalues of $\subspacecolumn^{\dagger}\mlvec{M}\subspacecolumn$.

The Killing-Cartan classification indicates that the subgroup $G_{\ansatzname}$
restricted on the invariant subspace $V$ must be one of the simple lie groups. Here we
focus on the case where the subgroup $G_{\ansatzname}$ restricted
on the invariant subspace $V$ is a special unitary group $SU(\deff)$. Similar
results can be proved for special orthogonal, symplectic group by replacing the
integral forumla, which can be found for example in \cite{collins2006integration}).
By definition $V$ is invariant under the action of
any operator represented by ansatz $\ansatzname$, indicating the dynamics of the
output state is restricted to the subspace $V$ spanned by $\subspacecolumn$. By
transforming all the Hamiltonians and the input state
by $\subspacecolumn$ in the proof of Theorem~\ref{thm:vqe-convergence-full}, we
have the following corollary:
\begin{corollary}
\label{cor:effective-convergence}
Let $(\mlvec{M}, \ket{\Phi}, \mlvec{U})$ be a VQE instance using compatible
ansatz design $\ansatzname$ with $\deff$, $\kappaeff$, $\Heff$ and $(\lambda_{1}',\cdots, \lambda_{\deff}')$ as defined in
Definition~\ref{def:ansatz-compatibility} and that the distribution of the subgroup $G_{\ansatzname}$
restricted to the subspace is Haar measure over special unitary matrices. Let $\ket{\Psi^{\star}}$ denote the
ground state of $\mlvec{M}$ and $\ket{\Psi(t)}$ the output state at time $t$.
If the number of parameters $p$ is greater than or equal to an over-parameterization threshold
$p_{\mathrm{th}}$ of order
      $O\big(
          \kappa_{\mathrm{eff}}^4,
          \frac{d_{\mathrm{eff}}^4}{Z({\Heff},d_{\mathrm{eff}})^{3}},
          \log\left({d_{\mathrm{eff}}}\right)
        \big)$,
then with probability $\geq 0.99$, under gradient flow with learning rate
of $\eta = \frac{1}{pZ(\Heff,\deff)}$,
the output state converges to the ground state with error
$\epsilon = 1 - |\<\Psi(t)|\Psi^\star\>|^2$ in time
$T_\epsilon = O\big(\frac{\log d_\mathrm{eff}}{\lambda'_2 - \lambda'_1}\log\frac{1}{\epsilon}\big)$. The success probability may be boosted to $1-\delta$ for any $0 \le \delta \le 1$ using $O\left(\log\frac{1}{\delta}\right)$ random restarts.
\end{corollary}
The proof for Corollary~\ref{cor:effective-convergence} is postponed to
Section~\ref{sec:app_subgroup}. For general ansatz design $\ansatzname$ including HEA with
$G_{\ansatzname} = SU(d)$, the effective dimension $\deff$ (resp. effective
ratio $\kappaeff$) is the same as the system dimension $d$ (resp. the ratio
$\kappa=\frac{\lambda_{d} - \lambda_{1}}{\lambda_{2}-\lambda_{1}}$). In fact
this is the case for fully-trainable ansatz that contain universal gate sets and
satisfy the premise of \cite{brandao2016local}.
On the other hand, a problem-specific compatible ansatz design can have much
smaller $\deff$ and $\kappaeff$ and achieve reasonable performance with much
fewer number of parameters. As we see in Section~\ref{sec:exp_threshold}, for
physical problems like transverse field Ising models and Heinsenberg models,
certain HVA designs can have $\deff$ and $\kappaeff$ orders of magnitudes
smaller than $d$ and $\kappa$.

\subsection{Estimating \texorpdfstring{$\deff$}{deff} and \texorpdfstring{$\kappaeff$}{kappa\_eff}}
\label{subsec:estimate-eff}
Given a VQE problem $(\mlvec{M}, \ket{\Phi}, \mlvec{U})$ with a compatible ansatz design
$\ansatzname$, we can estimate the column space of $\subspacecolumn$ of the
invariant subspace by estimating the support of the matrix
\begin{align}
\hat{\mlvec{\Pi}} = \frac{1}{R}\sum_{r=1}^R\mlvec{U}_{r}\newketbra{\Phi}{\Phi}\mlvec{U}^{\dagger}_{r}\label{eq:estimate_pi}
\end{align}
with $\mlvec{U}_{r}$ sampled $i.i.d.$ from the Haar measure over
$G_{\ansatzname}$. Empirically we approximate the Haar measure over $G_{\ansatzname}$ by calculating
\begin{align}
  \mlvec{U}(\mlvec\phi) = \prod_{l'=1}^{L_\mathsf{sample}}\prod_{k=1}^{K}\exp(-i\phi_{l',k} \generatorH{k}) \label{eq:approximate_haar}
\end{align}
for large $L_{\mathsf{sample}}$ and randomly initialized
$\{\phi_{l',k}\}_{k\in[K], l'\in[L_\mathsf{sample}]}$
(throughout this work $\phi_{l,k}$ are sampled uniformly and $i.i.d.$ from the
whole real space).
Any orthonormal
basis of the support of $\mlvec{\Pi}$ can be used as $\subspacecolumn$ to
estimate $\deff$ and $\kappaeff$ using
Definition~\ref{def:effective-quantities}. The computational cost for the
procedure depends on the quantities $R$ and $L_{\mathsf{sample}}$, and can be
$\mathsf{poly}(\deff)$ in the worst case, therefore we do not claim a fundamental
superiority in terms of computational complexity for large $d$ and $\deff$ when compared
with the standard practice of directly training VQE over multiple random
initializations and sweeping different number of parameters. However, we do
observe in our experiments that the estimation of $\deff$ and $\kappaeff$ for a
family of problem Hamiltonians is tremendously faster than training over multiple
random initializations and varying number of parameters for a single problem Hamiltonian. For example, it takes
$<0.2$ hours to evaluate transverse field ising model with up to $10$-qubit for
transverse field ranging from 0.1 to 1.5 on an Amazon C5 EC2 instance, while it
takes $~5$ hours to characterize a $4$-qubit instance with transverse field $g=0.3$
by performing training using the same machine.

\paragraph{Example: Kitaev model}
For a concrete example, consider the HVA for the Kitaev model on
square-octagon lattice with external field introduced in \cite{mixedkitaev2021}.
We will see that the proper ansatz design leads to an effective dimension much
smaller than the system dimension ($\deff=76$ v.s. $d=256$) and that the
effective ratio $\kappaeff$ can be orders of magnitudes smaller than
$\kappa=\frac{\lambda_{d} - \lambda_{1}}{\lambda_{2}-\lambda_{1}}$
(Figure~\ref{fig:cond-num-kitaev}).

The problem Hamiltonian for Kitaev models with external field is defined as
\begin{align}
  \mlvec{M}_{\mathsf{Kitaev}}(J_{xy}, h) = \sum_{(u,v)\in S_{Z}} Z_{u}Z_{v} + \frac{J_{xy}}{\sqrt{2}}\big(\sum_{(u,v)\in S_{X}}X_{u}X_{v} + \sum_{(u,v)\in S_{Y}}Y_{u}Y_{v}\big)  + h   \sum_{i=0}^{7}\big(X_{i} + Y_{i} + Z_{i}\big)
\end{align}
with $X_{i}$ denoting the Pauli-$X$ matrix acting on the $i$-th qubit. This
system has coupling in the X, Y, Z directions on edge sets $S_{X}$, $S_{Y}$ and
$S_{Z}$ respectively. The parameter $J_{xy}$ controls the coupling in the $X/Y$-direction and
$h$ controls the strength of the external field.  For $8$-qubit Kitaev models on
square-octagon lattice, by labeling each qubit with indexes $0$ through $7$, the
edge sets are defined as
$S_{X}= \{(0,1), (2,3)\}$, $S_{Y}= \{(1,2), (0,3)\}$, and
$S_{Z}= \{(4,0), (1,5), (3,7), (2,6)\}$ (See Figure~\ref{fig:kitaevconfig} or
Figure 1(c) in \cite{mixedkitaev2021}).

We use the ansatz proposed in \cite{mixedkitaev2021}
$\ansatzname = \{\generatorH{1}, \cdots, \generatorH{6}\}$ with
\begin{align}
  & \generatorH{1} \propto\sum_{(u,v)\in S_{X}}X_{u}X_{v}, \generatorH{4} \propto\sum_{i=0}^{7}X_{i},\nonumber\\
  & \generatorH{2} \propto\sum_{(u,v)\in S_{Y}}Y_{u}Y_{v}, \generatorH{5} \propto\sum_{i=0}^{7}Y_{i},\nonumber\\
  & \generatorH{3} \propto\sum_{(u,v)\in S_{Z}}Z_{u}Z_{v}, \generatorH{6} \propto\sum_{i=0}^{7}Z_{i}.\label{eq:ansatz-kitaev}
\end{align}
In Figure~\ref{fig:invar_sub}, we plot the eigenvalues of $\hat{\mlvec\Pi}$ for
the Kitaev models for input state $\ket{\Phi} = \ket{0}^{\otimes 8}$ and the ansatz specified in
Equation~(\ref{eq:ansatz-kitaev}) with $L_{\mathsf{sample}}$ chosen to be $20$.
The x-axis correponds to the indices of eigenvalues for the $256\times 256$
problem Hamiltonian, and the y-axis corresponds to the sorted eigenvalues.
The spectrums are color-coded for different $R$ ranging from $1$ to $100$, with
blue corresponding to small $R$ and red corresponding to large $R$. For small
$R$, $\hat{\mlvec{\Pi}}$ is restricted to a small subspace. As the number of
samples $R$ increases,  the rank of $\hat{\mlvec\Pi}$ increases, and converges
to a matrix with uniform eigenvalues. Figure~\ref{fig:invar_sub} indicates that the $\ket{\Phi}$ lies within the
$76$-dimensional invariant subspace $V$ embedded in a $256$-dimensional state
space $\mathcal{H}$. It is also verified that the ground state of $\mlvec{M}_{\mathsf{Kitaev}}$ lies
within the subspace $V$ as well. We also compare the effective ratio $\kappaeff$
with $\kappa=\frac{\lambda_{d}-\lambda_{1}}{\lambda_{2}-\lambda_{1}}$ (i.e. the
effective ratio for generic ansatz designs) for a wide range of parameters
$(J_{xy}, h)$ in Figure~\ref{fig:cond-num-kitaev}. We observe that the HVA proposed in
\cite{mixedkitaev2021} reduces $\kappaeff$ by orders of magnitudes.
\begin{figure}[!htbp]
  \begin{minipage}{0.45\textwidth}
    \centering
    \includegraphics[width=0.4\linewidth]{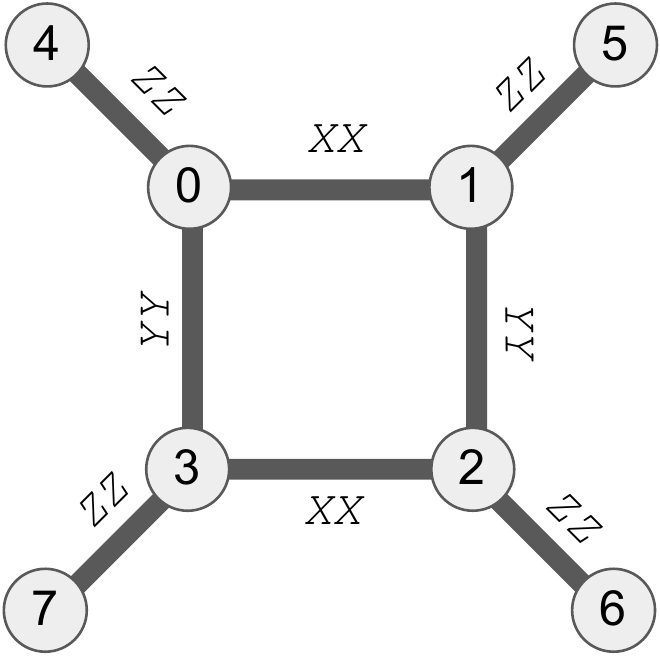}
    \caption{
    Configuration of the $8$-qubit Kitaev model on square-octagon lattice
    defined in \cite{mixedkitaev2021}. Qubits are labeled by $0, 1, \cdots, 7$,
    and each edge corresponds to an interation term. The types of interactions
    $XX, YY$ and $ZZ$
    are as specified in texts.
    }
    \label{fig:kitaevconfig}
  \end{minipage}
  \hfill
  \begin{minipage}{0.45\textwidth}
    \centering
      \includegraphics[width=0.7\linewidth]{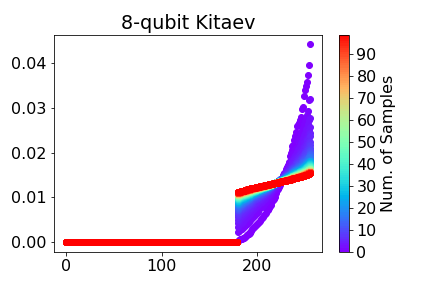}
    \caption{Spectrum of $\hat{\mlvec{\Pi}}$ for $8$-qubit Kitaev
      model with 8 qubits for number of samples $R=1, 2, \cdots, 100$. As
    the number of samples increases (the color changing from blue to red),
    $\hat{\mlvec{\Pi}}$ converges to a Hermitian with
    uniform spectrum, and can thus be good approximation of the normalized
    projection to the invariant subspace $V$.
    }
    \label{fig:invar_sub}
  \end{minipage}
\end{figure}

\begin{figure}[!htbp]
  \centering
  \subfigure[Kitaev Model: N=8, varying $J_{xy}$]{
    \includegraphics[width=0.33\linewidth]{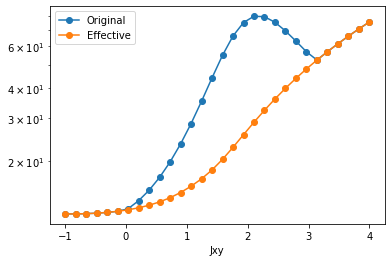}
  }
  \subfigure[Kitaev Model: N=8, varying $h$]{
    \includegraphics[width=0.33\linewidth]{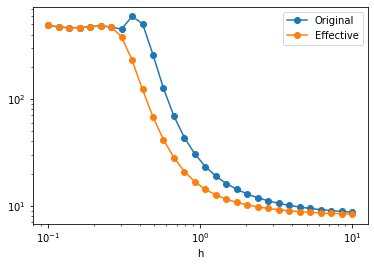}
  }
  \caption{The spectral ratio $\kappaeff$ for 8-qubit Kitaev models by varying
     $J_{xy}$ while fixing
     the external field $h=1$ and varying $h$ while fixing $J_{xy}=1$. The
     effective ratio is significantly smaller than the actual
     ratio for a wide range of $(J_{xy}, h)$.
  }
  \label{fig:cond-num-kitaev}
\end{figure}


\section{Empirical study I: Soundness of theory}
\label{sec:exp_confirm}
In this section we present two sets of numerical simulations to corroborate our
theoretical results.
\begin{enumerate}
  \item In Section~\ref{subsec:exp1}, we calculate the deviation of
        $\mlvec{Y}$ and $\mlvec{\theta}$ for HVAs and HEAs.
        We show that Lemma~\ref{lm:vqe-concentration-time} and
        \ref{lm:vqe-slow-theta-new} correctly predict the maximal deviation of
        $\mlvec{Y}$ and $\mlvec\theta$ for both the partially-trainable and the
        fully-trainable settings.

  \item In Section~\ref{subsec:exp2}, we confirm that the over-parameterization
        threshold is positively correlated to the proposed quantities
        $\kappa_{\mathsf{eff}}$ and $d_{\mathsf{eff}}$ as predicted in
        Theorem~\ref{thm:vqe-convergence-full} and
        Corollary~\ref{cor:effective-convergence} using synthetic VQE examples.
\end{enumerate}

\subsection{Experiment 1: Deviation of key quantities during training}
\label{subsec:exp1}
In this section we optimize Hamiltonian variational ansatz (HVA) and
\heafullname (HEA) in both the partially- and fully-trainable settings
and evaluate $\mlvec{Y}$ and $\mlvec\theta$ during training.
Recall that $\mlvec{Y}$ is a function of timestep $t$ through its dependency on
the parameters $\mlvec\theta(t)$.
Lemma~\ref{lm:vqe-slow-theta-new} and \ref{lm:vqe-concentration-time} predict that $\mlvec\theta$ remains in a $\ell_{\infty}$-ball
centered at $\mlvec\theta(0)$ with radius $O(1/p)$ throughout training, and that
$\|\mlvec{Y}(t) -\mlvec{Y}(0)\|_{\mathsf{op}} = O(\frac{1}{\sqrt{p}})$. Our
experiments show that it is true for both partially- and fully-trainable HVAs and HEAs.

\paragraph{Defining \mlvec{Y} for fully-trainable ansatz}
For
partially-trainable ansatz defined in Definition~\ref{def:partially-trainable-ansatz}, $\mlvec{Y}$ can be equivalently expressed as
\begin{align}
   \mlvec{Y}(\mlvec\theta) := \frac{1}{p} \sum_{l=1}^{p}\big(\mlvec{U}_{l,+}(\mlvec\theta) \mlvec{H} \mlvec{U}^{\dagger}_{l,+}(\mlvec\theta)\big)^{\otimes 2}
\end{align}
using $\mlvec{U}_{l,+}(\mlvec\theta) = \big(\prod_{l'=l+1}^{p}\mlvec{U}_{l'}\exp(-i\theta_{l'}\mlvec{H}) \big) \mlvec{U}_{l}$,
 the matrix applied to the input state after the rotation $\exp(-i\theta_{l}\mlvec{H})$.
 Similarly for fully-trainable ansatz
 $\mlvec{Y}$ can be defined as:
 \begin{align}
   \mlvec{Y}(\mlvec
   \theta) := \frac{1}{p} \sum_{l=1}^{L}\sum_{k=1}^{K}\big(\mlvec{U}_{(l,k),+}(\mlvec\theta) \generatorH{k} \mlvec{U}^{\dagger}_{(l,k),+}(\mlvec\theta)\big)^{\otimes 2}
 \end{align}
 by defining
 $\mlvec{U}_{(l,k),+}(\mlvec\theta)$ as $\prod_{l'=l+1}^{L}\prod_{k=1}^{K}\exp(-i\theta_{l',k} \generatorH{k})\cdot \prod_{k'=k+1}^{K}\exp(-i\theta_{l,k'}\generatorH{k'})$
 as the matrix applied to the input state after the rotation
 $\exp(-i\theta_{l,k}\generatorH{k})$. Recall that in the fully-trainable setting the total number of trainable
 parameters $p$ is  $K\cdot L$.

\paragraph{HVA for transverse field Ising models} For HVAs, we consider the one-dimensional transverse field Ising
models (TFI1d). The $N$-qubit problem Hamiltonian is defined as
\begin{align}
  \mlvec{M}_{\mathsf{TFI1d}}(g) = \sum_{i=0}^{N-1}X_{i}X_{i+1} + g \sum_{i=0}^{N-1}Z_{i}
  \label{eq:tfi-problemHamiltonian}
\end{align}
with periodic boundary conditions (i.e the $N$-th qubit is identified with the $0$-th qubit).
The parameter $g$ is the strength of the transverse field. In this experiment we
choose the input state $\frac{1}{\sqrt{2^{N}}}(1,1,\cdots,1)^{T}$ and the compact HVA for TFI1d model proposed in \cite{Wiersema2020} with $K=2$ and
\begin{align}
\generatorH{1} \propto \sum_{i=0}^{N-1}X_{i}X_{i+1}, \quad \generatorH{2} \propto \sum_{i=0}^{N-1}Z_{i}.\label{eq:tfi-TFI2}
\end{align}
For all the experiments, $\{\generatorH{k}\}_{k=1}^{K}$ are normalized such
that ${Z(\generatorH{k}, d)} = \tr(\generatorH{k}^{2}) / (d^{2}-1) = 1$.

For both the partially- and fully-trainable settings, we solve $4$-qubit TFI1d model
with external field $g=0.3$ using gradient descent with learning rate
$1\times 10^{-4} / p$, where the numbers of trainable parameters varying from
$30$ to $150$. For each $p$ we repeat the training over $20$ random initializations.
It is observed that for both settings the deviations of $\mlvec{Y}$ in operator
norm ($\opnorm{\mlvec{Y}(t) - \mlvec{Y}(0)}$) saturate after a few
iterations (see Figure~\ref{fig:deviation_conv_Y} in the appendix), and
$\max_{t\geq 0}\opnorm{\mlvec{Y}(t) - \mlvec{Y}(0)}$ displays an $O(1/\sqrt{p})$
dependency on $p$ (Figure~\ref{fig:deviation_Y}). Moreover, note that the same
reference line $50/\sqrt{p}$ (plotted in {\color{tab-green} green}) is used in
both Figure~\ref{fig:deviation_Y}(a) and (b). This indicates that
the maximal deviation of $\mlvec{Y}$ in the two settings not only match in the
dependency on $p$ but also on constants.
Similarly, the $O(1/p)$-dependencies of
$\max_{t\geq 0}\|\mlvec\theta(t) -\mlvec\theta(0)\|_{\infty}$ are demonstrated
in Figure~\ref{fig:deviation_theta}.

\begin{figure}[!htbp]
  \centering
  \subfigure[Partially-trainable HVA]{
    \includegraphics[width=0.33\linewidth]{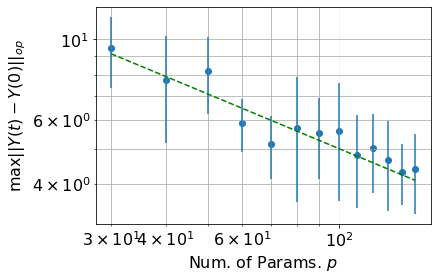}
  }
  \subfigure[Fully-trainable HVA]{
    \includegraphics[width=0.33\linewidth]{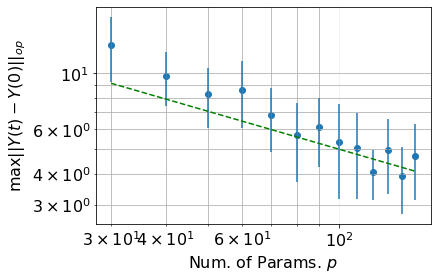}
  }
  \caption{Maximal deviation of $\mlvec{Y}$ from initial value as a function of
    number of trainable parameters during the training of HVA for $4$-qubit TFI1d model with
    transverse field $g=0.3$. The mean values and the standard deviations are
    calculated over $20$ random initializations for each number of trainable parameters
    $p=30, 40, \cdots, 150$. In both figures, the reference lines $50/\sqrt{p}$
    are plotted in {\color{tab-green} green}, showing that our theory correctly
    predicts the $O(1/\sqrt{p})$-dependency of
    $\max_{t\geq 0}\opnorm{\mlvec{Y}(t) - \mlvec{Y}(0)}$ for both settings.
  }
  \label{fig:deviation_Y}
\end{figure}

\begin{figure}[!htbp]
  \centering
  \subfigure[Partially-trainable HVA]{
    \includegraphics[width=0.33\linewidth]{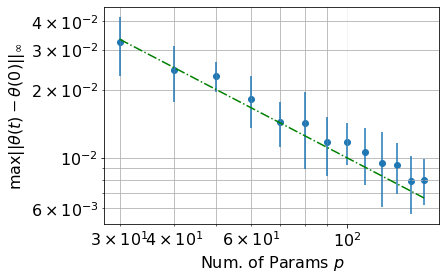}
  }
  \subfigure[Fully-trainable HVA]{
    \includegraphics[width=0.33\linewidth]{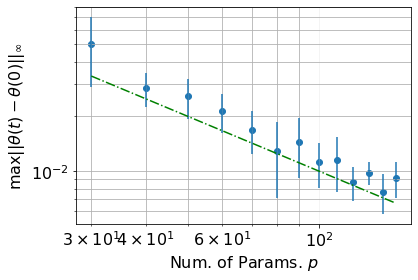}
  }
   \caption{Maximal deviation of $\mlvec{\theta}$ from initial value as a function of
    number of trainable parameters during the training of HVA for $4$-qubit TFI1d model with
    transverse field $g=0.3$. The mean values and the standard deviations are
    calculated over $20$ random initializations for each number of trainable parameters
    $p=30, 40, \cdots, 150$. In both figures, the reference lines $1/p$
    are plotted in {\color{tab-green} green}, showing that our theory correctly
    predicts the $O(1/{p})$-dependency of
    $\max_{t\geq 0}\|\mlvec{\theta}(t) - \mlvec{\theta}(0)\|_{\infty}$ for both settings.
  }
  \label{fig:deviation_theta}
\end{figure}

\paragraph{HEA with CZ entanglement} Similar observations occur in
\heafullname\ (HEA) with layers of single-qubit $X/Y$-rotations and $CZ$
entanglements as illustrated in Figure~\ref{fig:heaconfig}. For an $N$-qubit
instance, let $CZ_{ij}$ denote the CZ gate acting on the $i$-th and $j$-th
qubits, we define the CZ entanglement layer $\mlvec{U}_{\mathsf{CZ}}$ as:
\begin{align}
  \mlvec{U}_{\mathsf{CZ}} = \prod_{\mathsf{even }i\in [N]} CZ_{i,i+1} \prod_{\mathsf{odd }i\in [N]} CZ_{i,i+1}.
\end{align}
Using that fact that $CZ^{2}_{ij}$ is identity for any pair of $(i,j)$, the HEA
can be fit into the ansatz defined in Definition~\ref{def:fully-trainable-ansatz} with $K = 4N$ and
\begin{align}
  \generatorH{4i+1} \propto X_{i}, \quad \generatorH{4i+2} \propto Y_{i}, \quad
  \generatorH{4i+3} \propto \mlvec{U}_{\mathsf{CZ}}X_{i}\mlvec{U}_{\mathsf{CZ}}, \quad \generatorH{4i+4}\propto \mlvec{U}_{\mathsf{CZ}}Y_{i}\mlvec{U}_{\mathsf{CZ}},
  \quad\forall i \in [N]. \label{eq:HEA-CZ}
\end{align}
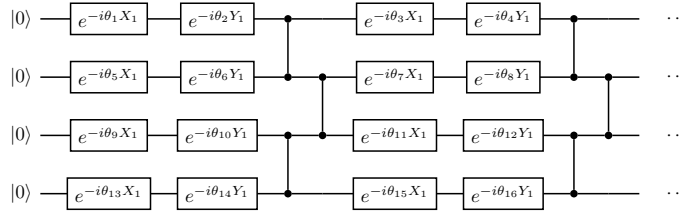
\begin{figure}[!htbp]
  \centering
  \resizebox{0.6\linewidth}{!}{
  \begin{quantikz}
    \lstick{$\ket{0}$} & \gate{e^{-i\theta_{1} {X}_1}} & \gate{e^{-i\theta_{2}
        {Y}_1}} & \ctrl{1} & \qw & \gate{e^{-i\theta_{3} {X}_1}}&
    \gate{e^{-i\theta_{4} {Y}_1}} & \ctrl{1} & \qw & \qw & \cdots\\
    \lstick{$\ket{0}$} & \gate{e^{-i\theta_{5} {X}_1}} & \gate{e^{-i\theta_{6}
        {Y}_1}} & \control{} & \ctrl{1} & \gate{e^{-i\theta_{7} {X}_1}}&
    \gate{e^{-i\theta_{8} {Y}_1}} & \control{} & \ctrl{1} & \qw & \cdots\\
    \lstick{$\ket{0}$} & \gate{e^{-i\theta_{9} {X}_1}} & \gate{e^{-i\theta_{10}
        {Y}_1}} & \ctrl{1} & \control{} & \gate{e^{-i\theta_{11} {X}_1}}&
    \gate{e^{-i\theta_{12} {Y}_1}} &  \ctrl{1} & \control{} &  \qw &\cdots\\
    \lstick{$\ket{0}$} & \gate{e^{-i\theta_{13} {X}_1}} & \gate{e^{-i\theta_{14}
        {Y}_1}} & \control{} & \qw & \gate{e^{-i\theta_{15} {X}_1}}&
    \gate{e^{-i\theta_{16} {Y}_1}} & \control{} & \qw &  \qw &\cdots
  \end{quantikz}
  }
  \caption{
  Configuration of the $4$-qubit HEA with CZ entanglements.
  }
  \label{fig:heaconfig}
\end{figure}
We use the ansatz defined in Equation~(\ref{eq:HEA-CZ}) to solve the problem Hamiltonian
\begin{align}
  \mlvec{M}_{\mathsf{HEA}} = \mathsf{diag}(0, 0.5, 1, \cdots, 1)
\end{align}
with input state
$\ket{\Phi} = \ket{0}^{\otimes N} = (1.0, 0, \cdots, 0)^{\dagger}$ and
learning rate $1\times 10^{-2} / p$. The
empirical results are summarized in Figure~\ref{fig:HEA_deviation_Y} and
\ref{fig:HEA_deviation_theta}.

\begin{figure}[!htbp]
  \centering
  \subfigure[Partially-trainable HEA]{
    \includegraphics[width=0.33\linewidth]{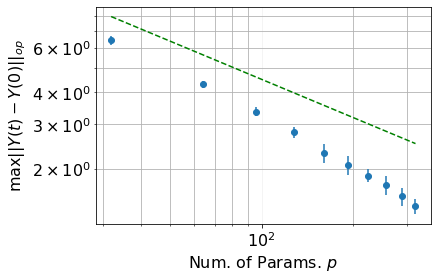}
  }
  \subfigure[Fully-trainable HEA]{
    \includegraphics[width=0.33\linewidth]{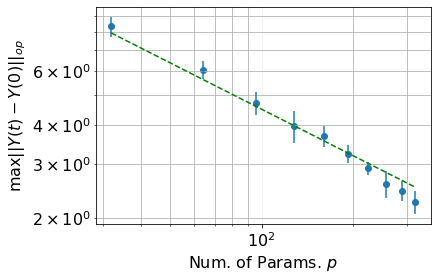}
  }
  \caption{Maximal deviation of $\mlvec{Y}$ in $4$-qubit \heafullname
    (HEA) with CZ entanglement. The mean values and the standard deviations are
    calculated over $10$ random initializations for each number of trainable parameters
    $p=32, 64, \cdots, 320$. In both figures, the reference lines are
    $45/\sqrt{p}$.
  }
  \label{fig:HEA_deviation_Y}
\end{figure}

\begin{figure}[!htbp]
  \centering
  \subfigure[Partially-trainable HEA]{
    \includegraphics[width=0.33\linewidth]{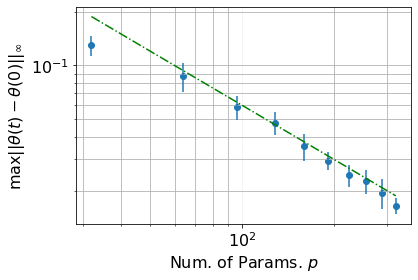}
  }
  \subfigure[Fully-trainable HEA]{
    \includegraphics[width=0.33\linewidth]{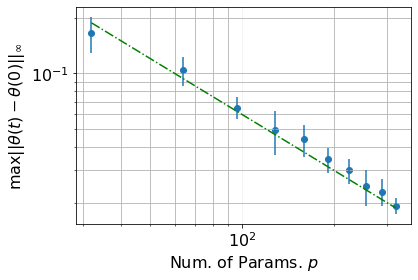}
  }
  \caption{Maximal deviation of $\mlvec{\theta}$ in $4$-qubit HEA with CZ
    entanglement. The mean values and standard deviations are calculated over
    $10$ random initialization for each $p$ varying from $32$ to $320$. The
    references lines in both figures are $6/p$.
  }
  \label{fig:HEA_deviation_theta}
\end{figure}

We also extend our experiments to the setting when the gradient estimation is
noisy. In Figure~\ref{fig:noisy_deviation}, we consider
$\noise{l}(t)$ sampled $i.i.d.$ from $\mathcal{N}(0, 1\times 10^{-5})$ for all $t$ and
$l\in[p]$, and have similar observation on the dependency of the maximal
deviation of $\mlvec{Y}$ and $\mlvec{\theta}$ on $p$.
\begin{figure}[!htbp]
  \centering
  \subfigure[Deviation of $\mlvec{Y}$ with noisy gradient]{
    \includegraphics[width=0.33\linewidth]{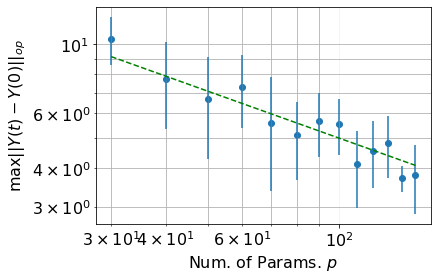}
  }
  \subfigure[Deviation of $\mlvec{\theta}$ with noisy gradient]{
    \includegraphics[width=0.33\linewidth]{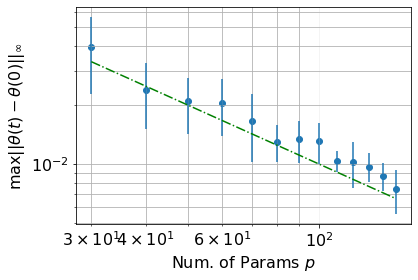}
  }
  \caption{Maximal deviation of $\mlvec{Y}$ and $\mlvec{\theta}$ in HVA for
    $4$-qubit TFI1d model with transverse field $g=0.3$ and random gaussian
    noise with $\sigma=1e-5$.
    The references lines for $\mlvec{Y}$ in both figures are $50/\sqrt{p}$.
    The references lines for $\mlvec\theta$ in both figures are $1/p$.
  }
  \label{fig:noisy_deviation}
\end{figure}

\subsection{Experiment 2: Over-parameterization in synthetic models}
\label{subsec:exp2}
In this section, we simulate gradient descent in synthetic VQE problems with varying
$d$, $\deff$ and $\kappaeff$ using ansatz with different number of parameters.
We show that
the over-parameterization thresholds are positively correlated to the effective dimensions
$\deff$ and spectral ratios $\kappaeff$ as predicted in Corollary~\ref{cor:effective-convergence}.

\paragraph{Estimating over-parameterization threshold} For a concrete criterion
of over-parameterization, we estimate the success rate for the training to
converge to an output state $\ket{\Psi(t)}$ such that the error
$1 - |\braket{\Psi(t)|\Psi^{\star}}|^{2}$ is less than $0.01$, where
$\ket{\Psi^{\star}}$ is the ground state. We define the over-parameterization
threshold as the smallest $p$ such that
$1 - |\braket{\Psi(t)|\Psi^{\star}}|^{2} > 0.01$ with probability $\geq 98\%$ over random initialization.

For physical problems like TFI1d, the system dimension $d$, effective
dimension $d_{\mathsf{eff}}$ and  $\kappa_{\mathsf{eff}}$ are jointly defined by
the number of qubits and the system parameters, be it external fields or the
strengths of coupling. We
decouple these parameters by starting with synthetic problems. For a synthetic
problem with $(d, \deff, \kappaeff)$, we embed a $\deff\times \deff$ Hermitian with eigenvalues
$(0,\frac{1}{\kappaeff}, 1, \cdots, 1)$ into a $d$-dimensional space and consider
ansatz with rotations restricted to the $\deff$-dimensional space (see
Section~\ref{sec:appendix_exp} for the concrete definition for the synthetic problems). For each set of
$(d, \deff, \kappaeff)$ and each number of trainable parameters $p$, the training is
repeated over $100$ random initializations with
learning rate $1\times 10^{-2}/ p$.

In Figure~\ref{fig:toy} we examine how the convergence depends on
the number of parameters $p$ for synthetic instances with varying $(d, \deff, \kappaeff)$:
In Figure~\ref{fig:toy}(a) we change
the system dimension $d$ with $\deff$ and $\kappaeff$ fixed. For all $d=8, 16, 32$, the
over-parameterization threshold is around $8$, showing that the convergence is almost
independent of the system dimension for fixed $\deff$ and $\kappaeff$.
In Figure~\ref{fig:toy}(b), we fix the system dimension $d=16$, $\kappaeff=2.0$ and
vary the effective dimension $\deff$: the over-parameterization threshold increases as
the effective dimension increases.

For a more quantitative evaluation, we define
the over-parameterization threshold as the smallest $p$ to achieve a success rate
of at least $98\%$, and plot the threshold for different $\deff$ in
Figure~\ref{fig:toy-thresh}(a). The dependency of the
over-parameterization thresholds on $\kappaeff$ are displayed in a similar way
in Figure~\ref{fig:toy}(c) and Figure~\ref{fig:toy-thresh}(b). It is clearly
reflected in Figure~\ref{fig:toy-thresh} (a) and (b) that the
over-parameterization threshold is positively correlated to $\deff$ and $\kappaeff$.
\begin{remark}
Readers may notice that the dependency on $\deff$ is almost
linear, seemingly contradicting previous empirical observation in
\cite{kiani2020learning} on using VQA to learn unitaries. There are several factors that may have contributed to
this discrepancy.
(1) The very first reason is that the over-parameterization threshold is
defined differently in \cite{kiani2020learning} as the smallest number of
parameters to achieve a certain shape of training curves.
(2) Another plausible reason is that the ratio $\kappaeff$ is a concise but inexact descriptor of all the eigenvalues of the
problem Hamiltonian. As the effective dimension varies, the eigenvalues also
vary in spite of the controlled $\kappaeff$.
(3) A third factor that might have contributed is the statistical error due to the
finite number of random initializations.
\end{remark}

\begin{figure}[!htbp]
  \centering
  \subfigure[Varying system dimension $d$]{
    \includegraphics[width=0.3\linewidth]{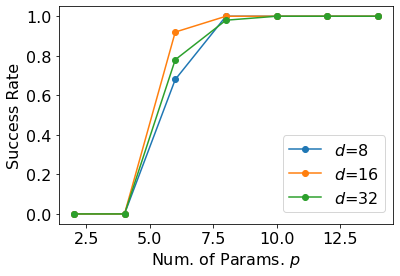}
  }
  \subfigure[Varying effective dimension $\deff$]{
    \includegraphics[width=0.3\linewidth]{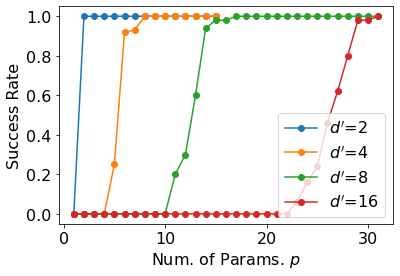}
  }
  \subfigure[Varying effective spectral ratio $\kappaeff$]{
    \includegraphics[width=0.3\linewidth]{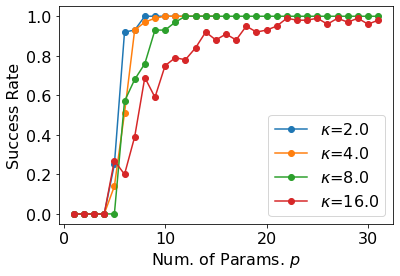}
  }
  \caption{
    Dependency of the over-parameterization threshold on system dimension $d$,
    effective dimension $\deff$ and the spectral ratio $\kappaeff$ in synthetic problems:
    the x-axes are the numbers of trainable parameters $p$, and the y-axes are the
    success rates for finding solutions with error less than $0.01$.
    For each data point, the success rate is estimated over $50$ random initializations.
    (a) Fixing $\deff=4, \kappaeff=2.0$, the over-parameterization threshold does
    not depend on the system dimension for $d=8, 16, 32$.
    (b) Fixing $d=16, \kappaeff=2.0$ for $\deff=2, 4, 6, 8$. The
    threshold increases as the effective dimension increases.
    (c) Fixing $d=16, \deff=4$ for $\kappaeff=2.0, 4.0, 8.0, 16.0$. The threshold
    is positively correlated to the effective ratio of the system.
  }
  \label{fig:toy}
\end{figure}

\begin{figure}[!htbp]
  \centering
  \subfigure[Varying effective dimension $\deff$]{
    \includegraphics[width=0.3\linewidth]{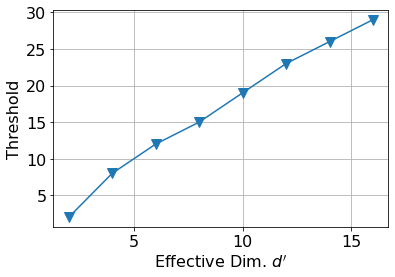}
  }
  \subfigure[Varying effective spectral ratio $\kappaeff$]{
    \includegraphics[width=0.3\linewidth]{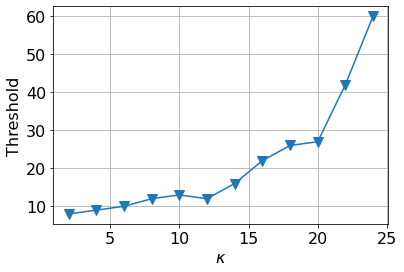}
  }
  \caption{
    Dependency of the over-parameterization threshold on
    effective dimension $\deff$ and effective spectral ratio $\kappaeff$ in synthetic problems:
    the over-parameterization thresholds in the plots are defined as the
    smallest $p$ achieving a success rate $\geq 98\%$ to find a solution with
    error less than $0.01$. Success rates are estimated over $50$ random initializations.
    (a) Fixing $d=16, \kappaeff=2.0$ for $\deff=2, 4, 6, \cdots, 16$.
    (b) Fixing $d=16, \deff=4$ for $\kappaeff=2.0, 4.0, 6.0, \cdots, 26.0$.
  }
  \label{fig:toy-thresh}
\end{figure}


\section{Empirical study II: Ansatz evaluation}
\label{sec:exp_threshold}
In this section, we use Corollary~\ref{cor:effective-convergence} to explain the
performances of different ansatz for Ising models and Heisenberg models.
Specifically, (1) we calculate $\deff$ and $\kappaeff$ using the procedure
described in Subsection~\ref{subsec:estimate-eff} and (2) directly estimate the
over-parameterization thresholds by repetitive training over random
initializations with different number of parameters. The results are summarized as follows:
\begin{itemize}
  \item For transverse field Ising (TFI) model, we compare ansatz $TFI_{\alt{2}}$ and
        $TFI_{\alt{3}}$ (defined below). $TFI_{\alt{2}}$ and $TFI_{\alt{3}}$ have identical $\kappaeff$, but
        $TFI_{\alt{2}}$ has smaller $\deff$. Empirically, we observe $TFI_{\alt{2}}$ reaches
        over-parameterization with fewer number of parameters.
  \item For the Heisenberg XXZ model, we compare ansatz $XXZ_{\alt{4}}$ and $XXZ_{\alt{6}}$
        (defined below). $XXZ_{\alt{4}}$ and $XXZ_{\alt{6}}$ have $\deff$ of same order of
        magnitude, but the $\kappaeff$ of $XXZ_{\alt{6}}$ diverges at the critical
        point while $\kappaeff$ of $XXZ_{\alt{4}}$ remain bounded. Empirically, we
        observe that as the system approaches the level-crossing point, $XXZ_{\alt{6}}$
        requires significantly more number of parameters to obtain a good
        approximation to the ground state.
  \item For both TFI and XXZ models and all HVA considered, $\deff$ is much smaller
        than the system dimension $d$. Also for $TFI_{\alt{2}}$, $TFI_{\alt{3}}$, $XXZ_{\alt{4}}$, the effective ratio
        $\kappaeff$ remain bounded near level-crossings where
        $\kappa=\frac{\lambda_{d}-\lambda_{1}}{\lambda_{2}-\lambda_{1}}$
        approaches infinity. This explains why problem-specifc HVA can be used
        to solve VQEs that can not be efficiently solved by general-purposed ansatz like
        HEA (\cite{Wiersema2020}) (Recall that for typical HEA design, $\deff$
        is the system dimension $d$ and $\kappaeff$ is simply $\kappa$).
\end{itemize}
These observations demonstrate the predicting power of the quantities $\deff$
and $\kappaeff$ and highlight that problem-specific ansatz designs are crucial to the
efficient training of VQE in practice.

\paragraph{Transverse field Ising (TFI) models} For one-dimensional TFI (TFI1d) models, in addition to HVA with
2-alternating Hermitian mentioned in Equation~\ref{eq:tfi-TFI2} in Section~\ref{subsec:exp1}  (which we will now refer to as
$TFI_{\alt{2}}$),  we consider the ansatz design $TFI_{\alt{3}}$ that contains 3 Hermitians
$\ansatzname = \{\generatorH{1}, \generatorH{2}, \generatorH{3}\}$ with
\begin{align}
\generatorH{1} \propto \sum_{\mathsf{even} i}X_{i}X_{i+1},\quad \generatorH{2} \propto \sum_{\mathsf{odd} i}X_{i}X_{i+1}, \quad \generatorH{3} \propto \sum_{i=0}^{N-1}Z_{i}\label{eq:tfi-TFI3}.
\end{align}
Compared with $TFI_{\alt{2}}$, $TFI_{\alt{3}}$ decouples the odd and even coupling in the
$X$ direcion. The effective dimension $\deff$ of
$TFI_{\alt{2}}$ and $TFI_{\alt{3}}$ for $N=4,6,8,10$ are summarized in Table~\ref{tbl:tfi}:
both ansatz designs achieve small effective dimension compared with the system dimension $d$, and the effective dimension $\deff$ for
$TFI_{\alt{2}}$ is consistently smaller than that of $TFI_{\alt{3}}$ for different $N$'s.
\begin{table}[h!]
\centering
\begin{tabular}{|c | c c c c|}
 \hline
  $N$ & 4 & 6 & 8 & 10\\
 \hline
 $d$ & 16 & 64 & 256 & 1024 \\
 $TFI_{\alt{2}}$ & 4 & 8 & 16 & 32 \\
 $TFI_{\alt{3}}$ & 5 & 10 & 25 & 50 \\
 \hline
\end{tabular}
\caption{System dimensions $d$ for $N$-qubit TFI1d models with  $N=4,6,8,10$
  and corresponding effective dimensions $\deff$ for ansatz $TFI_{\alt{2}}$ and $TFI_{\alt{3}}$.}
\label{tbl:tfi}
\end{table}

Despite the difference in $\deff$, $TFI_{\alt{2}}$ and $TFI_{\alt{3}}$ has
similar $\kappaeff$: in Figure~\ref{fig:tfi23-spectrum}, we
visualize the eigenvalues and $\kappaeff$ of $TFI_{\alt{2}}$, $TFI_{\alt{3}}$
and of the original
problem Hamiltonian $\mlvec{M}_{\mathsf{TFI1d}}(g)$ with varying transverse
field $g$ for $6$-qubit TFI1d models.
In Figure~\ref{fig:tfi23-spectrum} (a) and (b), we plot
the 4 smallest eigenvalues of the effective Hamiltonian $\mlvec{M}'$ associated with
$TFI_{\alt{2}}$ and $TFI_{\alt{3}}$: while $TFI_{\alt{2}}$ and $TFI_{\alt{3}}$ have different effective dimensions, they
have similar eigenvalues.

This allows us to demonstrate the dependency of the
threshold on the $\deff$ with controlled $\kappaeff$. In Figure~\ref{fig:tfi-varye} we plot the success rate against the
number of parameters $p$ for both ansatz with number of qubits $N=4,6,8,10$:
it is observed that $TFI_{\alt{2}}$ ($\blacktriangledown$) consistently achieve lower
over-parameterization threshold $p$ than $TFI_{\alt{3}}$ ($\blacksquare$) due to
smaller $\deff$.

Ground states of TFI1d models are degenerated for $|g| \leq 1$ in the
thermodynamic limit $N\rightarrow \infty$. Although there are no degeneracy for
finite $N$, the first excitation energy (i.e. the smallest eigen-gaps) decrease
quickly as $g$ drops below $1.0$. In Figure~\ref{fig:tfi23-spectrum}(c), we
visualize the smallest 4 eigenvalues for $N=6$. The vanishing eigen-gap for
small $g$ leads to drastic increase of $\kappaeff$ as plotted in {\color{tab-blue}
  blue} in Figure~\ref{fig:tfi23-spectrum}. On the contrary, the effective ratio $\kappaeff$ for both
$TFI_{\alt{2}}$ and $TFI_{\alt{3}}$ remain small as $g$ approaches $0$. As a result, the
over-parameterization threshold remains almost the same for $TFI_{\alt{2}}$ as
the transverse field $g$ decreases from $0.5$ to $0.1$ (as shown in
Figure~\ref{fig:tfi-varyg}). This shows that the usage of HVA instead of general
purpose ansatz design allows solving VQE problems efficiently near critical points.

\begin{figure}
  \centering
  \subfigure[$TFI_{\alt{2}}$]{
    \includegraphics[width=0.33\linewidth]{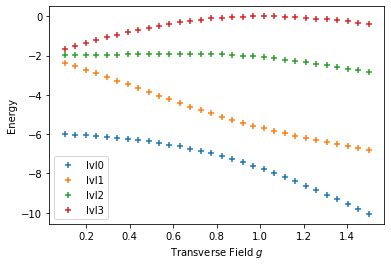}
  }
  \subfigure[$TFI_{\alt{3}}$]{
    \includegraphics[width=0.33\linewidth]{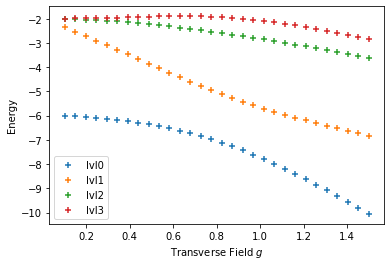}
  }\\
  \subfigure[Original]{
    \includegraphics[width=0.33\linewidth]{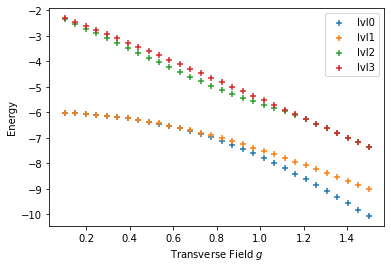}
  }
  \subfigure[Comparison of $\kappaeff$]{
    \includegraphics[width=0.33\linewidth]{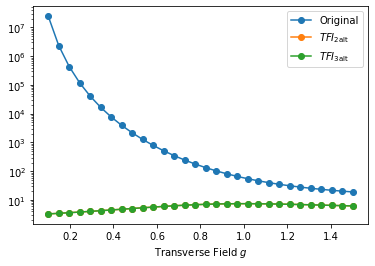}
  }
  \caption{
    Energy of the ground state and the first 3 excitation states.
    The smallest 4 eigenvalues
    for the effective Hamiltonian with $TFI_{\alt{2}}$ (a), $TFI_{\alt{3}}$ (b) and for the original
    Hamiltonian $H_{\mathsf{TFI1d}}(g)$ (c) for $N=6$ with transverse field $g$
    varying from $0.1$ to $1.5$.
    As plotted in (d) $\kappaeff$ for the original Hamiltonian increases quickly
    for $g$ close to $0$ while $\kappaeff$ for both $TFI_{\alt{2}}$ and
    $TFI_{\alt{3}}$ remain small. Note that $\kappaeff$ for $TFI_{\alt{2}}$ and
    $TFI_{\alt{3}}$ are overlapping.
  }
  \label{fig:tfi23-spectrum}
\end{figure}

\begin{figure}
  \centering
  \subfigure[]{
    \includegraphics[width=0.60\linewidth]{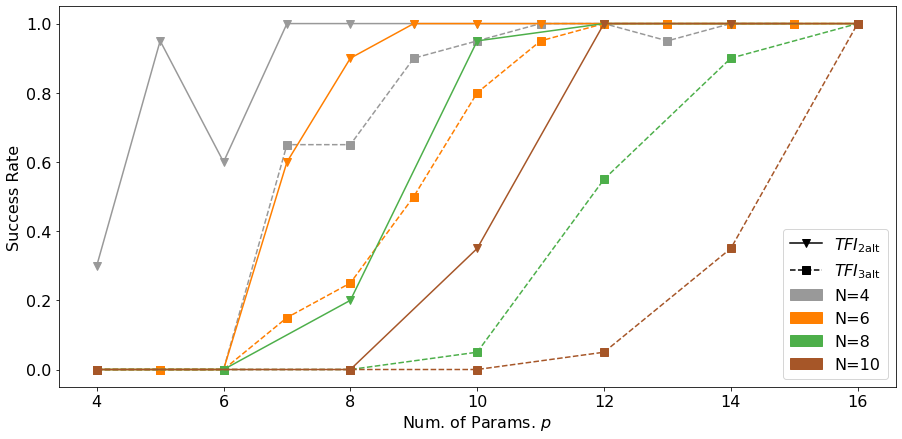}
  }
  \subfigure[]{
    \includegraphics[width=0.22\linewidth]{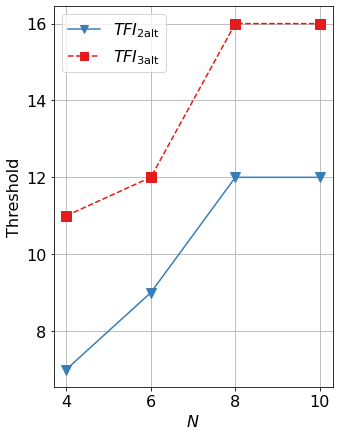}
  }
  \caption{
    Comparison of the over-parameterization threshold for $TFI_{\alt{2}}$ and
    $TFI_{\alt{3}}$ ansatz for $N=4, 6, 8, 10$. (a) The success rates for finding a
    solution with error less than $0.01$ versus the number of parameters for
    instances with different ansatz and different sizes. The number of qubits is encoded by
    different colors and the ansatz design is encoded by $\blacktriangledown$
    for $TFI_{\alt{2}}$ and
    $\blacksquare$ for $TFI_{\alt{3}}$.
    For each data point, the success rate is estimated over 20 random initializations.
    (b) Plot of the over-parameterization
    threshold versus number of qubits for different ansatz. The threshold is
    defined as the smallest number of parameters to achieve success rate
    over $98\%$.
    For each $N$, the threshold for $TFI_{\alt{2}}$ is lower than that of $TFI_{\alt{3}}$.
  }
  \label{fig:tfi-varye}
\end{figure}

\begin{figure}
  \centering
  \subfigure[N=6]{
    \includegraphics[width=0.3\linewidth]{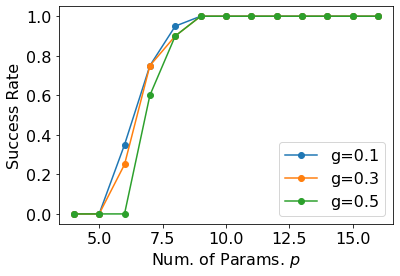}
  }
  \subfigure[N=8]{
    \includegraphics[width=0.3\linewidth]{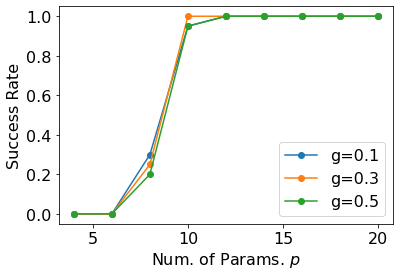}
  }
  \subfigure[N=10]{
    \includegraphics[width=0.3\linewidth]{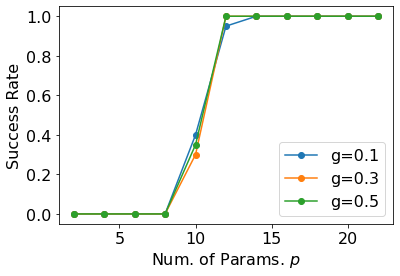}
  }
  \caption{
    Comparison of the over-parameterization threshold for $TFI_{\alt{2}}$ with
    transverse field $g=0.1, 0.3, 0.5$  for (a) $N=6$ (b)$N=8$ (c) $N=10$. The
    x-axis is the number of trainable parameters $p$, and the y-axis is the
    success rate for finding a solution with error less than $0.01$. For
    $N=6,8,10$, despite the vanishing eigen-gap of $H_{TFI1d}(g)$ for small
    $g$, the ground state can be found with reasonable $p$ with ansatz $TFI_{\alt{2}}$.
    For each data point, the success rate is estimated over 20 random initializations.
  }
  \label{fig:tfi-varyg}
\end{figure}

\paragraph{Heisenberg XXZ model} The one-dimensional XXZ (XXZ1d) model is a special case of Heisenberg model with
problem Hamiltonian defined as
\begin{align}
  \mlvec{M}_{\mathsf{XXZ1d}}(J_{zz}) = \sum_{i=0}^{N-1}X_{i}X_{i+1} + Y_{i}Y_{i+1} + J_{zz} \sum_{i=0}^{N-1}Z_{i}Z_{i+1}.
\end{align}
The parameter $J_{zz}$ controls the coupling in the $Z$-direction. XXZ1d model
is essentially different from the TFI1d model in that an actual level-crossing happens for
finite $N$ at $J_{zz}=-1$.

We examine the ansatz design proposed in \cite{Wiersema2020} (denoted as $XXZ_{\alt{4}}$):
\begin{align}
  & \generatorH{1} \propto \sum_{\mathsf{even} i}X_{i}X_{i+1} + \sum_{\mathsf{even} i}Y_{i}Y_{i+1},\\
  & \generatorH{2} \propto \sum_{\mathsf{odd} i}X_{i}X_{i+1} + \sum_{\mathsf{odd} i}Y_{i}Y_{i+1},\\
  & \generatorH{3} \propto \sum_{\mathsf{even} i}Z_{i}Z_{i+1},\quad \generatorH{4} \propto \sum_{\mathsf{odd} i}Z_{i}Z_{i+1}
\end{align}
as well as a similar design (denoted as $XXZ_{\alt{6}}$)
\begin{align}
  & \generatorH{1} \propto \sum_{\mathsf{even} i}X_{i}X_{i+1},\quad \generatorH{2} \propto \sum_{\mathsf{odd} i}X_{i}X_{i+1},\quad \generatorH{3} \propto \sum_{\mathsf{even} i}Y_{i}Y_{i+1},\\
  & \generatorH{4} \propto \sum_{\mathsf{odd} i}Y_{i}Y_{i+1}, \quad \generatorH{5} \propto \sum_{\mathsf{even} i}Z_{i}Z_{i+1},\quad \generatorH{6} \propto \sum_{\mathsf{odd} i}Z_{i}Z_{i+1}.
\end{align}
The effective dimensions for $XXZ_{\alt{4}}$ and $XXZ_{\alt{6}}$ are summarized in
Table~\ref{tbl:xxz}. While both $XXZ_{\alt{4}}$ and $XXZ_{\alt{6}}$ significantly reduce the
effective dimension $\deff$, $XXZ_{\alt{4}}$ further removes the level-crossing: in
Figure~\ref{fig:cond-num-xxz1d}, we see that both $XXZ_{\alt{4}}$ and $XXZ_{\alt{6}}$
reduces the ratio $\kappaeff$ by orders of magnitude, and the ratio $\kappaeff$ for
$XXZ_{\alt{4}}$ (in {\color{tab-orange} orange}) remains small as $J_{zz}\rightarrow -1$
while the ratio for both $XXZ_{\alt{6}}$ (in {\color{tab-green} green}) and the original
Hamiltonian (in {\color{tab-blue} blue}) increases to infinity. In
Figure~\ref{fig:xxz-varyg}, we present side-by-side the success rates of $XXZ_{\alt{4}}$
and $XXZ_{\alt{6}}$ for $N=4$ with $J_{zz} = -0.9, -0.5, -0.3, 0.1$. It is observed
that the over-parameterization threshold $XXZ_{\alt{4}}$ remain similar across
different values of $J_{zz}$ and the over-parameterization thresholds for
$XXZ_{\alt{6}}$ increase significantly as $J_{zz}$ decreases to $-0.9$ due to the
vanishing eigen-gaps.

\begin{table}[h!]
\centering
\begin{tabular}{|c | c c c c|}
 \hline
  $N$ & 4 & 6 & 8 & 10\\
 \hline
 $d$ & 16 & 64 & 256 & 1024 \\
 $XXZ_{\alt{4}}$ & 3 & 4 & 12 & 21 \\
 $XXZ_{\alt{6}}$ & 4 & 5 & 19 & 34 \\
 \hline
\end{tabular}
\caption{System dimensions $d$ and effective dimensions $\deff$ for XXZ1d model with
  $N=4,6,8,10$ for $XXZ_{\alt{4}}$ and $XXZ_{\alt{6}}$.}
\label{tbl:xxz}
\end{table}

\begin{figure}[!htbp]
  \centering
  \subfigure[XXZ1d Model: N=4]{
    \includegraphics[width=0.33\linewidth]{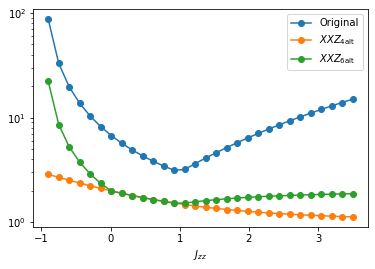}
  }
  \subfigure[XXZ1d Model: N=6]{
    \includegraphics[width=0.33\linewidth]{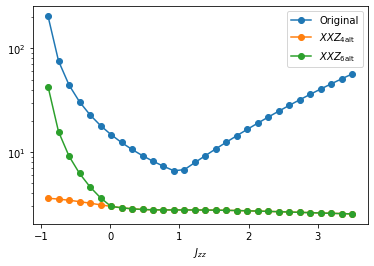}
  }\\
  \subfigure[XXZ1d Model: N=8]{
    \includegraphics[width=0.33\linewidth]{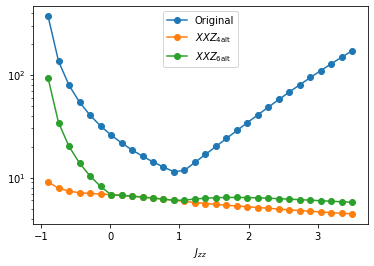}
  }
  \subfigure[XXZ1d Model: N=10]{
    \includegraphics[width=0.33\linewidth]{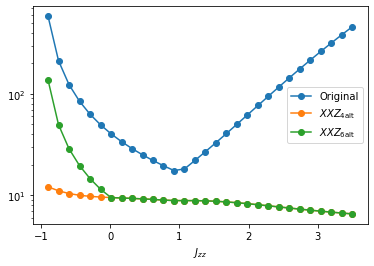}
  }
  \caption{
    Spectral ratios
    $\kappa$ and $\kappaeff$ for $XXZ_{\alt{4}}$ and $XXZ_{\alt{6}}$ for $N=4, 6, 8, 10$.
     We plot $\kappa$ for XXZ1d model and $\kappaeff$ for $XXZ_{\alt{4}}$ and
     $XXZ_{\alt{6}}$ for different values of $J_{zz}$.
     For both $XXZ_{\alt{6}}$ and the original problem Hamiltonian, level crossing
     happens at $J_{zz}=-1$, making it impossible to solve for the ground state
     when $J_{zz}$ is close to $-1$. Note that the level crossing breaks down
     under $XXZ_{\alt{4}}$.
  }
  \label{fig:cond-num-xxz1d}
\end{figure}

\begin{figure}[!htbp]
  \centering
  \subfigure[$XXZ_{\alt{4}}$]{
    \includegraphics[width=0.33\linewidth]{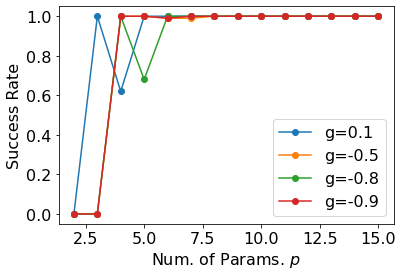}
  }
  \subfigure[$XXZ_{\alt{6}}$]{
    \includegraphics[width=0.33\linewidth]{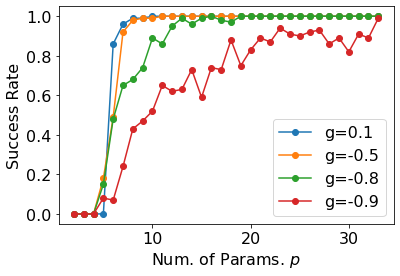}
  }
  \caption{
    Comparison of the over-parameterization threshold for (a) $XXZ_{\alt{4}}$ and (b) $XXZ_{\alt{6}}$ with
    $Z$-coupling  $J_{zz}=0.1, -0.3, -0.5, -0.9$. The x-axis is the number of trainable parameters $p$, and the y-axis is the
    success rate for finding a solution with error less than $0.01$. For
    $XXZ_{\alt{4}}$ the over-parameterization threshold remain similar for various
    $J_{zz}$, while for $XXZ_{\alt{6}}$ the threshold drastically increases as $J_{zz}$
    approaches $-1$ as a result of the level-crossing.
    For each data point, the success rate is estimated over 100 random initializations.
  }
  \label{fig:xxz-varyg}
\end{figure}


\newpage

\bibliographystyle{plain}
\bibliography{references}

\appendix

\section{Proof of lemmas for VQE}
\label{sec:vqe-proof}
\subsection{Proof for Lemma~\ref{lm:vqe-concentration-init}}
\vqeconcentrationinit*
\begin{proof}
Define
\begin{align}
\mlvec{X}_{l}:=\frac{1}{Z(\mlvec{H},d)}\big(\mlvec{U}_{l:p}(\mlvec\theta(0))\mlvec{H}\mlvec{U}^{\dagger}_{l:p}(\mlvec\theta(0))\big)^{\otimes 2} - \mlvec{Y}^{\star}.
\end{align}
By straight-forward calculation, we know that $X_{l}$ is centered (i.e
$\EXP[X_{l}] = 0$). The set $\{\mlvec{X}_{l}\}$ can be viewed as independent random matrices as the Haar
random unitary removes all the correlation. The matrix on the left-hand side can
therefore be expressed as the arithmetic average of $p$ independent random matrices.
The square of $\mlvec{X}_{l}$ is bounded in operator norm:
\begin{align}
  \opnorm{\mlvec{X}_{l}^{2}} = \opnorm{\mlvec{X}_{l}}^{2}\leq (\frac{\opnorm{\mlvec{H}}^{2}}{Z} + \frac{d+1}{d})^{2} \leq (\frac{2\opnorm{\mlvec{H}}^{2}}{Z(\mlvec{H},d)})^{2}
 \end{align}
where the second inequality follows from the fact that the ratio $g_{1} = \opnorm{\HH}^{2} / \tr(\HH^{2})$ satisfies that $1 \geq g_{1} \geq 1/d$.
By Hoeffding's inequality(\cite{tropp2012user}, Thm 1.3), with probability
$\geq 1-\delta$,
\begin{align}
  \opnorm{\YY(\mlvec\theta(0)) - \YY^{\star}} \leq \frac{1}{\sqrt{p}}\cdot \frac{2\opnorm{\mlvec{H}}^{2}}{Z(\mlvec{H},d)}\sqrt{\log\frac{d^{2}}{\delta}}.
\end{align}
\end{proof}

\subsection{Proof for Lemma~\ref{lm:vqe-perturbation}}
\label{subsec:app_proof_vqe_perturbation}
\vqeperturbation*
\begin{proof}[Proof for Lemma~\ref{lm:vqe-perturbation}]
  Let $\E(t)$ denote the deviation of $\YY(t)$ from its expected value $\mlvec{Y}^{\star}$:
  \begin{align}
    \E(t) := \YY(t) - \mlvec{Y}^{\star} = \YY(t) - (\mlvec{W} - \frac{1}{d}\mlvec{I}_{d^{2}\times d^{2}}).
  \end{align}
  The matrix that governs the dynamics can be expressed as
  \begin{align}
    \tr_{1}(\mlvec{Y}(t) ([\mlvec{M}, \newketbra{\Psi(t)}{\Psi(t)}] \otimes \mlvec{I}_{d\times d}))
    = [\mlvec{M}, \newketbra{\Psi(t)}{\Psi(t)}] + E(t)
  \end{align}
  where
  \begin{align}
    E(t) :=  \tr_{1}\big(\E(t)([\mlvec{M}, \newketbra{\Psi{(t)}}{\Psi{(t)}}] \otimes \mlvec{I}_{d\times d})\big).
  \end{align}

  Define $h$ as $|\braket{\Psi^{\star}|\Psi(t)}|^{2}$, and the time derivate of $h$ is
  \begin{align}
    \frac{d}{dt}h &= (\frac{d}{dt}\ket{\Psi(t)})^{\dagger}\newketbra{\Psi^{\star}}{\Psi^{\star}}\Psi(t)\rangle +
                    \bra{\Psi(t)}\newketbra{\Psi^{\star}}{\Psi^{\star}}\frac{d}{dt}\ket{\Psi(t)} \\
    &= 2(\bra{\Psi(t)} \mlvec{M} \ket{\Psi(t)} - \lambda_{1})|\braket{\Psi^{\star}|\Psi(t)}|^{2}  + \tr(E(t) [\newketbra{\Psi^{\star}}{\Psi^{\star}}, \newketbra{\Psi(t)}{\Psi(t)}]).\label{ln:perturbed_GF}
  \end{align}
  The first term in Line~(\ref{ln:perturbed_GF}) corresponds to the actual
  Riemannian gradient flow on the sphere:
  \begin{align}
    2(\bra{\Psi(t)} \mlvec{M} \ket{\Psi(t)} - \lambda_{1})|\braket{\Psi^{\star}|\Psi(t)}|^{2}
    &=2(\bra{\Psi(t)} \mlvec{M} \ket{\Psi(t)} - \lambda_{1}) h\\
    &\geq 2 ((1-h)\lambda_{2} + h\lambda_{1} - \lambda_{1}) h\\
    &= 2 (\lambda_{2} - \lambda_{1}) (1-h)h\label{ln:perturbed_GF-term1}
  \end{align}
  The second term in Line~(\ref{ln:perturbed_GF}) stems from the deviation of
  $\mlvec{Y}$ from its expectation:
  \begin{align}
    &\tr(E(t) [\newketbra{\Psi^{\star}}{\Psi^{\star}}, \newketbra{\Psi(t)}{\Psi(t)}])\\
    =&\tr \big(\tr_{1}\big(\E(t)([\mlvec{M}, \newketbra{\Psi{(t)}}{\Psi{(t)}}] \otimes \mlvec{I}_{d\times d})\big) [\newketbra{\Psi^{\star}}{\Psi^{\star}}, \newketbra{\Psi(t)}{\Psi(t)}]\big)\\
    =&\tr\big(
       \E(t)([\mlvec{M}, \newketbra{\Psi{(t)}}{\Psi{(t)}}] \otimes [\newketbra{\Psi^{\star}}{\Psi^{\star}}, \newketbra{\Psi(t)}{\Psi(t)}])
       \big)\\
    \geq& - \opnorm{\E(t)} \|[\mlvec{M}, \newketbra{\Psi{(t)}}{\Psi{(t)}}] \otimes [\newketbra{\Psi^{\star}}{\Psi^{\star}}, \newketbra{\Psi(t)}{\Psi(t)}]\|_{\mathsf{tr}}\\
    =& - 2 \sqrt{h(1-h)} \opnorm{\E(t)} \|[\mlvec{M}, \newketbra{\Psi{(t)}}{\Psi{(t)}}] \|_{\mathsf{tr}} \label{ln:perturbed_GF_term2_eq1}\\
    \geq& - 2 \sqrt{d} \sqrt{h(1-h)} \opnorm{\E(t)} \|[\mlvec{M}, \newketbra{\Psi{(t)}}{\Psi{(t)}}] \|_{F} \\
    \geq& - 2\sqrt{2} \sqrt{d} \sqrt{h}(1-h) (\lambda_{d} - \lambda_{1}) \opnorm{\E(t)}. \label{ln:perturbed_GF_term2_eq2}
  \end{align}
  The equality on Line~(\ref{ln:perturbed_GF_term2_eq1}) follows from
  Lemma~\ref{lm:technical_commutation1}, and the inequality on
  Line~(\ref{ln:perturbed_GF_term2_eq2}) follows from Lemma~\ref{lm:technical_commutation2}.

  Combining the two terms, we can lower bound the time derivative of $h$ as
  \begin{align}
    \frac{d}{dt}h \geq 2(\lambda_{2}-\lambda_{1})(1-h)h(1 - \sqrt{2d} \frac{\lambda_{d} - \lambda_{1}}{\lambda_{2}-\lambda_{1}}\opnorm{\E(t)}\frac{1}{\sqrt{h}})),
  \end{align}
  or by dividing both sides by a negative number $-h$:
  \begin{align}
    \frac{d}{dt}(-\ln{h}) \leq -2(\lambda_{2}-\lambda_{1})(1-h)(1 - \sqrt{2d} \frac{\lambda_{d} - \lambda_{1}}{\lambda_{2}-\lambda_{1}}\opnorm{\E(t)}\frac{1}{\sqrt{h}})).
  \end{align}
  Dividing both sides by a positive number $-\ln{h}$:
  \begin{align}
    \frac{d}{dt}\ln{(-\ln{h})}
    &\leq -2(\lambda_{2}-\lambda_{1})\frac{1-h}{-\ln{h}}\cdot(1 - \sqrt{2d} \frac{\lambda_{d} - \lambda_{1}}{\lambda_{2}-\lambda_{1}}\opnorm{\E(t)}\frac{1}{\sqrt{h}}))\\
    &\leq -2(\lambda_{2}-\lambda_{1})\frac{1}{1-\ln{h}}\cdot(1 - \sqrt{2d} \frac{\lambda_{d} - \lambda_{1}}{\lambda_{2}-\lambda_{1}}\opnorm{\E(t)}\frac{1}{\sqrt{h}}))
  \end{align}
  where the second inequality follows from the fact that
  $\frac{1-h}{-\ln h} \geq \frac{1}{1-\ln{h}}$ for $h\in(0,1)$ (adaptive from
  the technical Lemma 4 in \cite{xu2018convergence}). For $\E(t)$ such that
  $1 - \sqrt{2d}\frac{\lambda_{d}-\lambda_{1}}{\lambda_{2}-\lambda_{1}}\opnorm{\E(t)} \frac{1}{\sqrt{h}}$
  is positive, $h$ is non-decreasing, meaning that $h(t) \geq h(0)$ for all $t$.
  Conditioned on $h(0)\geq \Omega(\frac{1}{d}$) at initialization, there exists
  a pair of constants $C_{0}, c$ such that if
  $\opnorm{\E(t)} \leq \frac{C_{0}}{d}\frac{\lambda_{2}-\lambda_{1}}{\lambda_{d}-\lambda_{1}}$
  for all $t$,
  $1 - h(t) \leq - \ln h(t) \leq \exp(-c\frac{\lambda_{2}-\lambda_{1}}{\log d} t)$.
 \end{proof}


\section{Technical lemmas}
\label{sec:app_technical}
\subsection{Technical lemma for VQE convergence}
\begin{lemma}\label{lm:technical_commutation1}
  Let $\mlvec{x}, \mlvec{v}$ be two vectors in $\complex^d$, the commutator
  $i[\mlvec{x}\mlvec{x}^\dagger, \mlvec{v}\mlvec{v}^\dagger]$ has a pair of non-zero eigenvalues
  $\pm |\<\mlvec{x},\mlvec{v}\>|
  \sqrt{1 - |\<\mlvec{x},\mlvec{v}\>|^2}$.
\end{lemma}
\begin{proof}
  Express $\mlvec x$ as $\alpha \mlvec{v} + \beta \mlvec{w}$ with $\mlvec{w}$
  orthogonal to $\mlvec{v}$.
  \begin{align}
    i[\mlvec{xx}^\dagger, \mlvec{vv}^\dagger] &= i\alpha^*\beta \mlvec{wv}^\dagger - i \alpha\beta^*\mlvec{vw}^\dagger.
  \end{align}
  This rank-2 Hermitian has two real eigenvalues $\lambda_+$ and $\lambda_-$ such
  that $\lambda_+ + \lambda_- = 0$ and $\lambda_+ \lambda_- = - |\alpha|^2|\beta|^2$.
\end{proof}
\begin{lemma}[Bounding commutator norms]\label{lm:technical_commutation2}
  Let $\mlvec{M}:=\sum_{j=1}^d\lambda_j\mlvec{v}_j\mlvec{v}_j^\dagger$ be a
  $d\times d$-Hermitian matrix with eigenvalues $\lambda_1\leq \cdots \leq \lambda_d$. The frobenius norm of the commutator $[\mlvec{M},
  \mlvec{xx}^\dagger]$ can be bounded in terms of $|\<\mlvec{x},\mlvec{v}\>|$ as:
  \begin{align}
    \|[\mlvec{M}, \mlvec{xx}^\dagger]\|_F \leq& \sqrt{2}(\lambda_d - \lambda_1) \sqrt{1-|\<\mlvec{x}, \mlvec{v}\>|^2}.
  \end{align}
\end{lemma}
\begin{proof}
  We first notice that, for any real value $\lambda$,
  $[\mlvec{M} - \lambda \mlvec{I}, \mlvec{x}\mlvec{x}^{\dagger}] = [\mlvec{M}, \mlvec{x}{x}^{\dagger}] - \lambda [\mlvec{I}, \mlvec{x}\mlvec{x}^{\dagger}] = [\mlvec{M}, \mlvec{x}\mlvec{x}^{\dagger}]$.
  Therefore to bound $\|[\mlvec{M}, \mlvec{x}\mlvec{x}^{\dagger}]\|_{F}$, it
  suffices to bound $\|[\tilde{\mlvec{M}}, \mlvec{x}\mlvec{x}^{\dagger}]\|_{F}$
  for $\tilde{\mlvec{M}} = \mlvec{M} - \lambda_{1}\mlvec{I}$, with $\lambda_{1}$
  being the smallest eigenvalue of $\mlvec{M}$.

  Expand $\mlvec{x}$ as $\alpha\mlvec{v_1} + \beta\mlvec{w}$, where
  $\mlvec{v}_{1}$ is the ground state of $\mlvec{M}$ and unit vector $\mlvec{w}$ is
  orthogonal to $\mlvec{v}_{1}$:
  \begin{align}
    \|[\mlvec{M}, \mlvec{xx}^\dagger]\|^2_F &= \|[\tilde{\mlvec{M}}, \mlvec{xx}^\dagger]\|^2_F = 2 \big(\mlvec{x}^\dagger\tilde{\mlvec{M}}^2\mlvec{x} - (\mlvec{x}^\dagger\tilde{\mlvec{M}}\mlvec{x})^2\big) \\
    & = 2 \big(|\beta|^2\mlvec{w}^\dagger\tilde{\mlvec{M}}^2\mlvec{w} - (|\beta|^{2}\mlvec{x}^\dagger\tilde{\mlvec{M}}\mlvec{x})^2\big) \leq 2 |\beta|^2\mlvec{w}^\dagger\tilde{\mlvec{M}}^2\mlvec{w}\\
    &\leq 2 |\beta|^2(\lambda_d - \lambda_1)^2.
  \end{align}
\end{proof}

\subsection{Techinical lemma for concentration of VQE dynamics}
In this subsection, we state and prove an estimation lemma used for proving the
concentration properties of $\mlvec{Y}(\mlvec{\theta})$.
\begin{restatable}[Estimation with Taylor expansion]{lemma}{helpertaylor}
  \label{lm:helper-taylor}
  Let $\mlvec{V}$ be a unitary matrix generated by Hermitian $\mlvec{H}$ as
  $\mlvec{V}=\exp(-i\theta\mlvec{H})$, we have that for any Hermitian
  $\mlvec{K}$
  \begin{align}
    \opnorm{(\mlvec{VKV}^{\dagger})^{\otimes 2} - \mlvec{K}^{\otimes 2}} &\leq 4 |\theta|\opnorm{\mlvec{H}}\opnorm{\mlvec{K}}^{2}.
  \end{align}
\end{restatable}
\begin{proof}
  The first- and second-order derivatives of $(\mlvec{V}\mlvec{K}\mlvec{V}^{\dagger})^{\otimes 2}$ are:
  \begin{align}
    \frac{d}{d\theta} (\mlvec{V}\mlvec{K}\mlvec{V}^{\dagger})^{\otimes 2} &= \mlvec{V}^{\otimes 2} ([-i\mlvec{H}, \mlvec{K}]\otimes \mlvec{K} + \mlvec{K} \otimes [-i\mlvec{H}, \mlvec{K}])(\mlvec{V}^{\dagger})^{\otimes 2},\\
    \frac{d^{2}}{d\theta^{2}} (\mlvec{V}\mlvec{K}\mlvec{V}^{\dagger})^{\otimes 2} &= - \mlvec{V}^{\otimes 2}
                                                            (2 [\mlvec{H}, \mlvec{K}]\otimes [\mlvec{H},\mlvec{K}] + [\mlvec{H}, [\mlvec{H}, \mlvec{K}]]\otimes \mlvec{K} + \mlvec{K} \otimes [\mlvec{H}, [\mlvec{H},\mlvec{K}]])
                                                            (\mlvec{V}^{\dagger})^{\otimes 2}.
  \end{align}
  Hence
  \begin{align}
    &\opnorm{(\mlvec{V}\mlvec{K}\mlvec{V}^{\dagger})^{\otimes 2} - \mlvec{K}^{\otimes 2}}\\
    =& \opnorm{\int_{0}^{\theta}d\theta^{\prime}(e^{-i(\theta-\theta^{\prime})\mlvec{H}})^{\otimes 2} ([-i\mlvec{H}, \mlvec{K}]\otimes \mlvec{K} + \mlvec{K} \otimes [-i\mlvec{H}, \mlvec{K}])(e^{i(\theta-\theta^{\prime})\mlvec{H}})^{\otimes 2}} \\
    \leq & 4|\theta|\opnorm{\mlvec{H}}\opnorm{\mlvec{K}}^{2}.
  \end{align}
\end{proof}

\section{Proof of Corollary~\ref{cor:noisy-convergence}}
\label{sec:app_noise}
\noisyvqedynamics*
\begin{proof}
  We start by calculating the gradient of $\mlvec{U}_{r:p}(\mlvec\theta)$ with
  respect to $\theta_{l}$.
  For $r > l$, $\mlvec{U}_{r:p}(\mlvec\theta)$ is independent of $\theta_{l}$; for $r \leq l$,
  \begin{align}
    \frac{\partial \mlvec{U}_{r:p}}{\partial\theta_{l}}
    = \mlvec{U}_{l:p}(\mlvec{\theta}) (-i\mlvec{H})  \mlvec{U}_{r:l-1}(\mlvec\theta)
    = -i\mlvec{U}_{l:p}\mlvec{H}\mlvec{U}_{l:p}^{\dagger}\mlvec{U}_{r:p}.
  \end{align}
Therefore
  \begin{align}
    \frac{\partial L(\mlvec{\theta})}{\partial\theta_l}
    &=\bra{\Phi}\mlvec{U}^{\dagger}_{0}\frac{\partial}{\partial\theta_{l}}\mlvec{U}^{\dagger}_{1:p}\mlvec{M}\mlvec{U}_{1:p}\mlvec{U}_{0}\ket{\Phi}
      + \bra{\Phi}\mlvec{U}^{\dagger}_{0}\mlvec{U}^{\dagger}_{1:p}\mlvec{M}\frac{\partial}{\partial\theta_{l}}\mlvec{U}_{1:p}\mlvec{U}_{0}\ket{\Phi}\\
    &= \bra{\Phi}\mlvec{U}^{\dagger}_{0}\mlvec{U}^{\dagger}_{1:p}i[\mlvec{U}_{l:p}\mlvec{H}\mlvec{U}^{\dagger}_{l:p},\mlvec{M}]\mlvec{U}_{1:p}\mlvec{U}_{0}\ket{\Phi}\\
    &= i\tr([\mlvec{M}, \newketbra{\Psi}{\Psi}] \mlvec{U}_{l:p}\mlvec{H}\mlvec{U}^{\dagger}_{l:p}).
  \end{align}
  Following gradient flow with learning rate $\eta$:
\begin{align}
  \frac{d\theta_l}{dt} = -\eta\big(\frac{\partial}{\partial\theta_{l}}L(\mlvec\theta) + \noise{l}\big)
    = -i\eta\tr([\mlvec{M}, \newketbra{\Psi(t)}{\Psi(t)}] \mlvec{U}_{l:p}\mlvec{H}\mlvec{U}^{\dagger}_{l:p}) - \eta\noise{l}.
\end{align}

The dynamics for $\mlvec{U}_{l:p}(\mlvec\theta(t))$ and $\ket{\Psi(t)}$ are therefore:
\begin{align}
  &\frac{d}{dt}\mlvec{U}_{l:p}(t)\\
  =& \sum_{r=l}^p\frac{d\theta_r}{dt}\frac{\partial}{\partial\theta_{r}}\mlvec{U}_{l:p}\\
  =& -\eta\sum_{r=l}^p\tr([\mlvec{M},\newketbra{\Psi(t)}{\Psi(t)}]\mlvec{U}_{r:p}\mlvec{H}\mlvec{U}^{\dagger}_{r:p})\mlvec{U}_{r:p}\mlvec{H}\mlvec{U}^{\dagger}_{r:p}\mlvec{U}_{l:p} +i \eta \sum_{r=1}^{p}\noise{r} \mlvec{U}_{r:p}\mlvec{H}\mlvec{U}^{\dagger}_{r:p}\mlvec{U}_{l:p},
\end{align}
and
\begin{align}
  \frac{d}{dt}\ket{\Psi(t)}
  &= \frac{d}{dt}\mlvec{U}_{1:p}\mlvec{U}_0\ket{\Phi}\\
  &= -(\eta\cdot p Z)\frac{1}{p Z}
    \big(\sum_{l=1}^p\tr([\mlvec{M},\newketbra{\Psi(t)}{\Psi(t)}]\mlvec{U}_{l:p}\mlvec{H}\mlvec{U}^{\dagger}_{l:p})\mlvec{U}_{l:p}\mlvec{H}\mlvec{U}^{\dagger}_{l:p}\big)
    \mlvec{U}_{1:p}\mlvec{U}_0\ket{\Phi}\\
  &+i\eta \sum_{l=1}^{p}\noise{l} \mlvec{U}_{l:p}\mlvec{H}\mlvec{U}^{\dagger}_{l:p}\mlvec{U}_{1:p}\mlvec{U}_{0}\ket{\Phi}\\
  &= - (\eta\cdot pZ)\tr_1(\mlvec{Y}  [\mlvec{M},\newketbra{\Psi(t)}{\Psi(t)}]\otimes \mlvec{I})\ket{\Psi(t)}+ \eta \sum_{l=1}^{p}i\noise{l}\mlvec{H}_{l}\ket{\Psi(t)}.
\end{align}
\end{proof}

\noisyvqeperturbation*
\begin{proof}[Proof for Lemma~\ref{lm:noisy-vqe-perturbation}]
  Let $\E(t):=\YY(t) - (\mlvec{W} - \frac{1}{d}\mlvec{I}_{d^{2}\times d^{2}})$ denote the deviation of $\YY(t)$ from its expected value.
  The matrix that governs the dynamics can be expressed as
  \begin{align}
    \tr_{1}(\mlvec{Y}(t) ([\mlvec{M}, \newketbra{\Psi{(t)}}{\Psi{(t)}}] \otimes \mlvec{I}_{d\times d}))
    = [\mlvec{M}, \newketbra{\Psi{(t)}}{\Psi{(t)}}] + E(t)
  \end{align}
  where
  \begin{align}
    E(t) :=  \tr_{1}\big(\E(t)([\mlvec{M}, \newketbra{\Psi{(t)}}{\Psi{(t)}}] \otimes \mlvec{I}_{d\times d})\big).
  \end{align}

  Define $h$ as $|\braket{\Psi^{\star}|\Psi(t)}|^{2}$, the time derivative of $h$
  \begin{align}
    \frac{d}{dt}h &= (\frac{d}{dt}\ket{\Psi(t)})^{\dagger}\newketbra{\Psi^{\star}}{\Psi^{\star}}\Psi(t)\rangle +
                    \braket{\Psi(t)|\Psi^{\star}}\cdot\bra{\Psi^{\star}}\frac{d}{dt}\ket{\Psi(t)} \\
                  &= 2(\bra{\Psi(t)} \mlvec{M} \ket{\Psi(t)} - \lambda_{1})|\braket{\Psi^{\star}|\Psi(t)}|^{2}\\
                  &+ \tr(E(t) [\newketbra{\Psi^{\star}}{\Psi^{\star}}, \newketbra{\Psi(t)}{\Psi(t)}])\\
                  &+ \tr(N(t) [\newketbra{\Psi^{\star}}{\Psi^{\star}}, \newketbra{\Psi(t)}{\Psi(t)}])
  \end{align}
  with $N(t)$ defined as $-\frac{1}{pZ} \sum i\noise{l}_{t}\mlvec{H}_{l}$.
  The first term corresponds to the actual
  Riemannian gradient flow on the sphere:
  \begin{align}
    2(\bra{\Psi(t)} \mlvec{M} \ket{\Psi(t)} - \lambda_{1})|\braket{\Psi^{\star}|\Psi(t)}|^{2}
    &=2(\bra{\Psi(t)} \mlvec{M} \ket{\Psi(t)} - \lambda_{1}) h\\
    &\geq 2 ((1-h)\lambda_{2} + h\lambda_{1} - \lambda_{1}) h\\
    &= 2 (\lambda_{2} - \lambda_{1}) (1-h)h.
  \end{align}
  The second term stems from the deviation of
  $\mlvec{Y}$ from its expectation:
  \begin{align}
    &\tr(E(t) [\newketbra{\Psi^{\star}}{\Psi^{\star}}, \newketbra{\Psi(t)}{\Psi(t)}])\\
    =&\tr \big(\tr_{1}\big(\E(t)([\mlvec{M}, \newketbra{\Psi{(t)}}{\Psi{(t)}}] \otimes \mlvec{I}_{d\times d})\big) [\newketbra{\Psi^{\star}}{\Psi^{\star}}, \newketbra{\Psi(t)}{\Psi(t)}]\big)\\
    =&\tr\big(
       \E(t)([\mlvec{M}, \newketbra{\Psi{(t)}}{\Psi{(t)}}] \otimes [\newketbra{\Psi^{\star}}{\Psi^{\star}}, \newketbra{\Psi(t)}{\Psi(t)}])
       \big)\\
    \geq& - \opnorm{\E(t)} \|[\mlvec{M}, \newketbra{\Psi{(t)}}{\Psi{(t)}}] \otimes [\newketbra{\Psi^{\star}}{\Psi^{\star}}, \newketbra{\Psi(t)}{\Psi(t)}]\|_{\mathsf{tr}}\\
    =& - 2 \sqrt{h(1-h)} \opnorm{\E(t)} \|[\mlvec{M}, \newketbra{\Psi{(t)}}{\Psi{(t)}}] \|_{\mathsf{tr}} \\
    \geq& - 2 \sqrt{d} \sqrt{h(1-h)} \opnorm{\E(t)} \|[\mlvec{M}, \newketbra{\Psi{(t)}}{\Psi{(t)}}] \|_{F} \\
    \geq& - 2\sqrt{2} \sqrt{d} \sqrt{h}(1-h) (\lambda_{d} - \lambda_{1}) \opnorm{\E(t)}\\
    \geq& - C_{4}h(1-h) (\lambda_{2} - \lambda_{1}) \frac{1}{\sqrt{hd}}.
  \end{align}
  Here we use technical Lemma~\ref{lm:technical_commutation1} and
  \ref{lm:technical_commutation2} and the fact that $\opnorm{\E(t)}$ is $O(\frac{\lambda_{2}-\lambda_{1}}{\lambda_{d}-\lambda_{1}}\cdot\frac{1}{d})$.

  The third term is a result of inaccurate estimation of gradients:
  \begin{align}
    &\tr(N(t) [\newketbra{\Psi^{\star}}{\Psi^{\star}}, \newketbra{\Psi(t)}{\Psi(t)}]) \geq -2 \opnorm{N(t)} \sqrt{h(1-h)} \geq - C_{5}(\lambda_{2}-\lambda_{1}) h(1-h)
  \end{align}
  Here we use the fact that
  $\opnorm{N(t)} \leq O\big((\lambda_{2} - \lambda_{1})\sqrt{h(1-h)}\big)$ if
  $\|\noisevec\|_{\infty}$ is $O\big(\frac{Z}{\opnorm{\mlvec{H}}}(\lambda_{2}-\lambda_{1})\sqrt{h(1-h)}\big)$.

  Combining all three terms, we have
  \begin{align}
    \frac{d}{dt}h \geq C_{6}(\lambda_{2}-\lambda_{1})(1-h)h(1-\frac{C_{7}}{\sqrt{hd}}).
  \end{align}
  Following the same calculation in
  Section~\ref{subsec:app_proof_vqe_perturbation}, we have
  $1 - h(t) \leq - \ln h(t) \leq \exp(-c\frac{\lambda_{2}-\lambda_{1}}{\log d} t)$
  for some constant $c$ if $h(0)$ is $\Omega(1/d)$.
\end{proof}


\section{Proof of Corollary~\ref{cor:effective-convergence}}
\label{sec:app_subgroup}
The proof of Corollary~\ref{cor:effective-convergence} involves replacing the
integration formula in the proof to the main theorem with integration over
subgroups. We start by presenting a basic fact about block-diagonal matrices (Lemma~\ref{lm:subgroup_fact}) and
the integration formula for subgroups of $SU(d)$
(Lemma~\ref{lm:subgroup_integration}).

\begin{lemma}[Basic fact]
  \label{lm:subgroup_fact}
  Let $G$ be a matrix subgroup of $SU(d)$ inducing a decomposition of invariant
  subspace $V=\oplus_{j=1}^{m}{V}_{j}$ with projections
  $\{\mlvec{\Pi}_{j}\}_{j=1}^{m}$. Without loss of generality, assume $V_{1}$ to be
  the subspace of interest. Then for any Hermitian $\mlvec{A}$ and unitary matrix $\mlvec{U}$ in
  group $G$:
  \begin{align}
    \mlvec{\Pi}_{1}\mlvec{U}\mlvec{A}\mlvec{U}^{\dagger}\mlvec{\Pi}_{1} = \mlvec{\Pi}_{1}\mlvec{U}\mlvec{\Pi}_{1} \  \mlvec{\Pi}_{1}\mlvec{A}\mlvec{\Pi}_{1} \  \mlvec{\Pi}_{1}\mlvec{U}^{\dagger}\mlvec{\Pi}_{1}
  \end{align}
\end{lemma}
\begin{proof}
  The decomposition of invariant subspaces dictates that any
  $\mlvec{U}\in G$ is block-diagonal under $\{\mlvec{\Pi}_{j}\}_{j=1}^{m}$,
  namely $\forall \mlvec{U}\in G, \forall j\neq j'$,
  $\mlvec{\Pi}_{j'}\mlvec{U}\mlvec{\Pi}_{j} = 0$.
   \begin{align}
    &\mlvec{\Pi}_{1}\mlvec{U}\mlvec{A}\mlvec{U}^{\dagger}\mlvec{\Pi}_{1}\\
    =& \mlvec{\Pi}_{1}\mlvec{U}\sum_{j=1}^{m}\mlvec{\Pi}_{j}\mlvec{A}\sum_{j'=1}^{m}\mlvec{\Pi}_{j'}\mlvec{U}^{\dagger}\mlvec{\Pi}_{1}\\
    =& \sum_{j,j'\in[m]}(\mlvec{\Pi}_{1}\mlvec{U}\mlvec{\Pi}_{j})\mlvec{A}(\mlvec{\Pi}_{j'}\mlvec{U}^{\dagger}\mlvec{\Pi}_{1})\\
     =& \mlvec{\Pi}_{1}\mlvec{U}\mlvec{\Pi}_{1}\mlvec{A}\mlvec{\Pi}_{1}\mlvec{U}^{\dagger}\mlvec{\Pi}_{1}\\
     =& \mlvec{\Pi}_{1}\mlvec{U}\mlvec{\Pi}_{1} \mlvec{\Pi}_{1} \mlvec{A} \mlvec{\Pi}_{1} \mlvec{\Pi}_{1}\mlvec{U}^{\dagger}\mlvec{\Pi}_{1}.
  \end{align}
  The last equation uses the property of projections $\mlvec{\Pi}_{j}^{2} = \mlvec{\Pi}_{j}$.
\end{proof}
As a direct result, we have the following generic integral formula for $\mlvec{U}$
sampled from any $\mathcal{D}$ supported on the subgroup $G$:
\begin{lemma}[Integration formula on subgroup restricted to an invariant subspace]
  \label{lm:subgroup_integration}
  Let $G$ be a matrix subgroup of $SU(d)$ inducing a decomposition of invariant
  subspace $V=\oplus_{j=1}^{m}{V}_{j}$ with projections
  $\{\mlvec{\Pi}_{j}\}_{j=1}^{m}$. Without loss of generality, assume $V_{1}$ to be
  the subspace of interest and let $\subspacecolumn \in \complex^{d\times \deff}$ be an arbitrary orthonormal
  basis for $V_{1}$. For any Hermitians $\{\mlvec{A}_{r}\}_{r=1}^{R}$ and
  measure $\mathcal{D}$ over $G$:
  \begin{align}
    (\subspacecolumn^{\dagger})^{\otimes R} \EXP_{\mlvec{U}\sim\mathcal{D}}[\otimes_{r=1}^{R} \mlvec{U}\mlvec{A}_{r}\mlvec{U}^{\dagger}]  \subspacecolumn^{\otimes R}
    = \EXP_{\mlvec{U}^{(1)}\sim\mathcal{D}^{(1)}}[\otimes_{r=1}^{R} \mlvec{U}^{(1)}\mlvec{A}_{r}^{(1)}(\mlvec{U}^{(1)})^{\dagger}]
  \end{align}
  where $\mathcal{D}^{(1)}$ is the distribution of
  $\subspacecolumn^{\dagger}\mlvec{U}\subspacecolumn$ for $\mlvec{U}$ sampled with respect
  to $\mathcal{D}$, and
  $\mlvec{A}_{r}^{(1)}:=\subspacecolumn^{\dagger}\mlvec{A}_{r}\subspacecolumn$ is the
  Hermitian $\mlvec{A}_{r}$ restricted to the subspace $V_{1}$.
\end{lemma}
Lemma~\ref{lm:subgroup_integration} allows using the integration formula in
\cite{collins2006integration} when $\mathcal{D}^{(1)}$ is the Haar measure over a
special unitary, special orthogonal or symplectic group. We are now ready to
present the proof of Corollary~\ref{cor:effective-convergence}.

\paragraph{Proof of Corollary~\ref{cor:effective-convergence}}
  Without loss of generality, we assume $V_{1}$ to be the relevant subspace with
  projection $\mlvec{\Pi}_{1} = \subspacecolumn\subspacecolumn^{\dagger}$. For
  concise notations, define
  $\mlvec{U}^{(1)} = \subspacecolumn^{\dagger}\mlvec{U}\subspacecolumn$,
  $\mlvec{A}^{(1)} = \subspacecolumn^{\dagger}\mlvec{A}\subspacecolumn$ and
  $\ket{\Psi^{(1)}} = \subspacecolumn^{\dagger}\ket{\Psi}$ for any unitary
  $\mlvec{U}$, Hermitian $\mlvec{A}$ and vector $\ket{\Psi}$.

  Note that the potential function we track in the proof of
  Theorem~\ref{thm:vqe-convergence-full}
  $|\<\Psi^{\star}|\Psi(t)\>|^{2}$ is equal to $|\<\Psi^{(1), \star}|\Psi^{(1)}(t)\>|^{2}$
  if both $\ket{\Psi}^{\star}$ and $\ket{\Psi(t)} \in V_{1}$. Therefore for the
  purpose of the proof it suffices to track the dynamics of
  $\ket{\Psi^{(1)}(t)}$. Below we (1) first establish that $\ket{\Psi(t)} \in V_{1}$
  through out the training and (2) then show that the dynamics of
  $\ket{\Psi^{(1)}(t)}$ takes the same form as stated in
  Lemma~\ref{lm:vqe-dynamics} by replacing $\mlvec{M}$ and $\mlvec{H}$ with
  $\mlvec{M}^{(1)} = \subspacecolumn^{\dagger}\mlvec{M}\subspacecolumn$ and
  $\mlvec{H}^{(1)}=\subspacecolumn^{\dagger}\mlvec{H}\subspacecolumn$.

  By Lemma~\ref{lm:vqe-dynamics}, the dynamics of $\ket{\Psi}$ takes the form
  \begin{align}
    \frac{d}{dt}\ket{\Psi} \propto -\frac{1}{p}\sum_{l=1}^{p}\tr([\mlvec{M}, \newketbra{\Psi(t)}{\Psi(t)}]\mlvec{U}_{l:p}\mlvec{H}\mlvec{U}_{l:p}^{\dagger})\mlvec{U}_{l:p}\mlvec{H}\mlvec{U}_{l:p}^{\dagger}\ket{\Psi(t)}.\label{eq:subgroup_dynamics}
  \end{align}

  We first show that $\ket{\Psi(t)}$ remains in $V_{1}$ for all $t$ (i.e. $\ket{\Psi(t)} = \mlvec{\Pi}_{1}\ket{\Psi(t)}$).
  It suffices to show the time derivate $\frac{d\ket{\Psi}}{dt}$ stays in $V_{1}$
  for $\ket{\Psi} \in V_{1}$ by noticing that for all $l\in[p]$,
   \begin{align}
    & \mlvec{\Pi}_{1}\mlvec{U}_{l:p}\mlvec{H}\mlvec{U}_{l:p}^{\dagger}\ket{\Psi(t)}\\
   =& \mlvec{U}_{l:p}\mlvec{\Pi}_{1}\mlvec{H}\mlvec{U}_{l:p}^{\dagger}\ket{\Psi(t)}\\
   =& \mlvec{U}_{l:p}\mlvec{H} \mlvec{\Pi}_{1}\mlvec{U}_{l:p}^{\dagger}\ket{\Psi(t)}\\
   =& \mlvec{U}_{l:p}\mlvec{H} \mlvec{U}_{l:p}^{\dagger} \mlvec{\Pi}_{1}\ket{\Psi(t)}\\
   =& \mlvec{U}_{l:p}\mlvec{H} \mlvec{U}_{l:p}^{\dagger} \ket{\Psi(t)}.
  \end{align}
 The first and the third equality is because
 $\mlvec{U}_{l:p}\in G_{\ansatzname}$  for all $l\in[p]$ and therefore block-diagoanl under $\{\mlvec{\Pi}_{j}\}_{j=1}^{m}$;
 The second equality is because and $\mlvec{H}$ is block-diagoanl under
 $\{\mlvec{\Pi}_{j}\}_{j=1}^{m}$;
 The last equality follows from $\ket{\Psi} \in V_{j}$.

 We now calculate the dynamics of $\ket{\Psi^{(1)}(t)}$. For the trace operation in each term,
 \begin{align}
   & \tr([\mlvec{M}, \newketbra{\Psi(t)}{\Psi(t)}] \mlvec{U}_{l:p}\mlvec{H}\mlvec{U}_{l:p}^{\dagger})\\
   =& \tr([\newketbra{\Psi(t)}{\Psi(t)}, \mlvec{U}_{l:p}\mlvec{H}\mlvec{U}_{l:p}^{\dagger}] \mlvec{M})\\
   =& \tr([\mlvec{\Pi}_{1}\newketbra{\Psi(t)}{\Psi(t)}\mlvec{\Pi}_{1}, \mlvec{U}_{l:p}\mlvec{H}\mlvec{U}_{l:p}^{\dagger}] \mlvec{M})\\
   =& \tr([\mlvec{\Pi}_{1}\newketbra{\Psi(t)}{\Psi(t)}\mlvec{\Pi}_{1}, \mlvec{\Pi}_{1}\mlvec{U}_{l:p}\mlvec{\Pi}_{1}\mlvec{H}\mlvec{\Pi}_{1}\mlvec{U}_{l:p}^{\dagger}\mlvec{\Pi}_{1}] \mlvec{M})\\
   =& \tr([\mlvec{\Pi}_{1}\newketbra{\Psi(t)}{\Psi(t)}\mlvec{\Pi}_{1}, \mlvec{\Pi}_{1}\mlvec{U}_{l:p}\mlvec{\Pi}_{1}\mlvec{H}\mlvec{\Pi}_{1}\mlvec{U}_{l:p}^{\dagger}\mlvec{\Pi}_{1}] \mlvec{\Pi}_{1} \mlvec{M}\mlvec{\Pi}_{1})\\
   =& \tr(
      [
      \subspacecolumn^{\dagger}\newketbra{\Psi(t)}{\Psi(t)}\subspacecolumn,
      \subspacecolumn^{\dagger}\mlvec{M}\subspacecolumn
      ]
      \subspacecolumn^{\dagger}\mlvec{U}_{l:p}\subspacecolumn
      \subspacecolumn^{\dagger}\mlvec{H}\subspacecolumn
      \subspacecolumn^{\dagger}\mlvec{U}_{l:p}^{\dagger}\subspacecolumn
      ).
 \end{align}
 The first, fourth and the fifth equation follow from basic properties of
 trace operators; the second equality uses the fact that $\ket{\Psi(t)}$ stays
 in $V_{j}$; the third equality uses the fact that $\mlvec{U}_{l:p}$ and
 $\mlvec{H}$ are block-diagonal.
 Therefore we can rewrite Equation~(\ref{eq:subgroup_dynamics}) as
 \begin{align}
   \frac{d}{dt}\ket{\Psi^{(1)}(t)}
   \propto
   -\frac{1}{p}\sum_{l=1}^{p}
   \tr(
   [\mlvec{M}^{(1)}, \newketbra{\Psi^{(1)}(t)}{\Psi^{(1)}(t)}]
   \mlvec{U}^{(1)}_{l:p}\mlvec{H}^{(1)}(\mlvec{U}_{l:p}^{(1)})^{\dagger}
   )
   \mlvec{U}^{(1)}_{l:p} \mlvec{H}^{(1)} (\mlvec{U}_{l:p}^{(1)})^{\dagger}\ket{\Psi^{(1)}(t)}.
 \end{align}
 The dynamics of $\ket{\Psi^{(1)}(t)}$ depends on
 $\subspacecolumn^{\dagger}\mlvec{M}\subspacecolumn$,
 $\subspacecolumn^{\dagger}\mlvec{H}\subspacecolumn$ and
 $\subspacecolumn^{\dagger}\mlvec{U}\subspacecolumn$.  Corollary~\ref{cor:effective-convergence} follows
 trivially by using the integration formula specified in Lemma~\ref{lm:subgroup_integration}.


\section{More on the empirical studies}
\label{sec:appendix_exp}
\paragraph{Implementation of partially-trainable ansatz} We implement the
partially-trainable ansatz (Definition~\ref{def:partially-trainable-ansatz}) by
 approximating the Haar measure over $G_{\ansatzname}$ by calculating
\begin{align}
  \mlvec{U}(\mlvec\phi) = \prod_{l'=1}^{L_\mathsf{sample}}\prod_{k=1}^{K}\exp(-i\phi_{l',k} \mlvec{H}_{k})
\end{align}
for $L_{\mathsf{sample}} = 20$ and randomly initialized
$\{\phi_{l',k}\}_{k\in[K], l'\in[L_\mathsf{sample}]}$.

\paragraph{Deviation of $\mlvec{Y}$ and $\mlvec\theta$ as functions of time
  $t$} In Figure~\ref{fig:deviation_conv_Y} and Figure~\ref{fig:deviation_conv_theta}, we plot the deviation of $\mlvec{Y}$
and $\mlvec\theta$ as functions of time steps $t$ for both the partially- and
fully-trainable settings. The mean values are plotted in solid lines and the
shaded areas represent the standard deviation over random initializations.
The maximum time steps is set to be $10,000$.
As observed in Figure~\ref{fig:deviation_conv_Y} and
\ref{fig:deviation_conv_theta}, the deviation of$\mlvec{Y}$ and
$\mlvec\theta$ saturates quickly after a few time steps.

\begin{figure}[!htbp]
  \centering
  \subfigure[Partially-trainable HVA]{
    \includegraphics[width=0.43\linewidth]{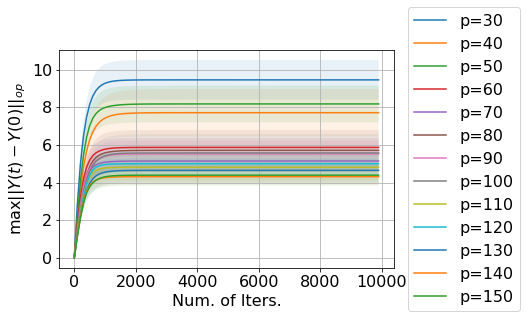}
  }
  \subfigure[Fully-trainable HVA]{
    \includegraphics[width=0.43\linewidth]{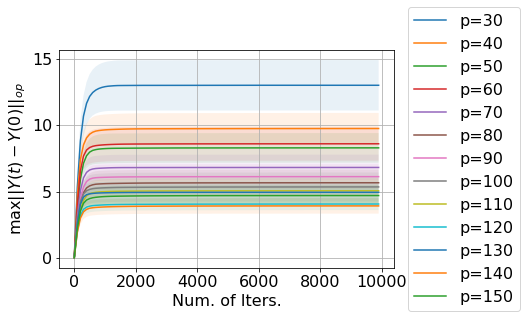}
  }\\
  \subfigure[Partially-trainable HEA]{
    \includegraphics[width=0.43\linewidth]{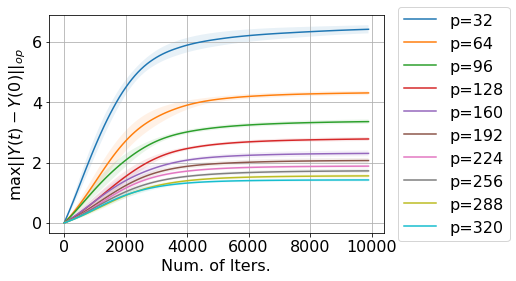}
  }
  \subfigure[Fully-trainable HEA]{
    \includegraphics[width=0.43\linewidth]{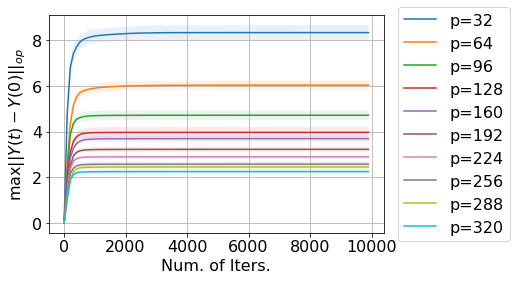}
  }
  \caption{Deviation of $\mlvec{Y}$ during training for HVA and HEA}
  \label{fig:deviation_conv_Y}
\end{figure}

\begin{figure}[!htbp]
  \centering
  \subfigure[Partially-trainable HVA]{
    \includegraphics[width=0.43\linewidth]{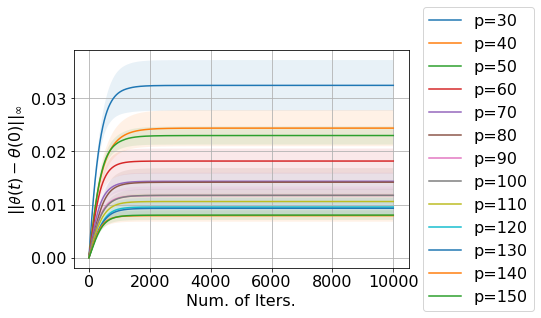}
  }
  \subfigure[Fully-trainable HVA]{
    \includegraphics[width=0.43\linewidth]{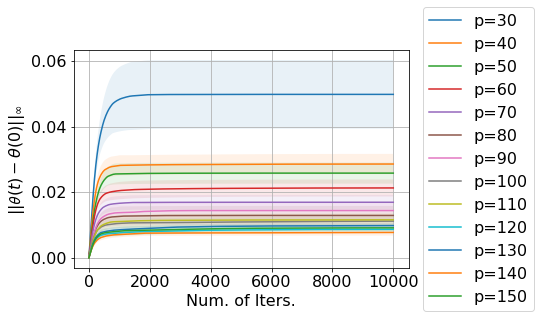}
  }\\
  \subfigure[Partially-trainable HEA]{
    \includegraphics[width=0.43\linewidth]{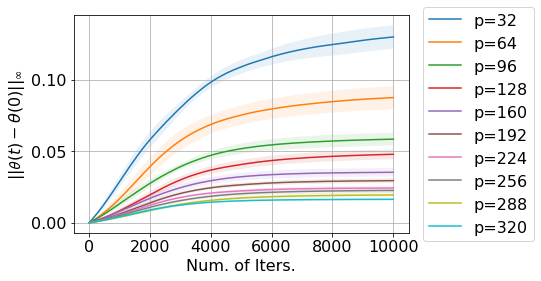}
  }
  \subfigure[Fully-trainable HEA]{
    \includegraphics[width=0.43\linewidth]{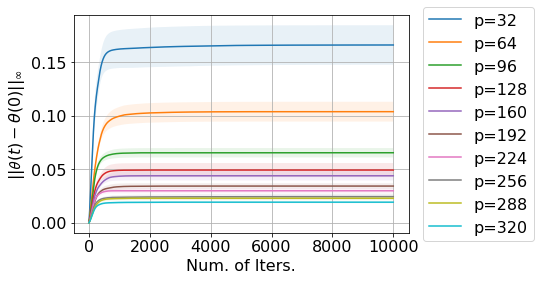}
  }
  \caption{Deviation of $\mlvec{\theta}$ during training for HVA and HEA}
  \label{fig:deviation_conv_theta}
\end{figure}

\paragraph{Definition of the synthetic problems} For the synthetic problem with system dimension $d$, effective
dimension $\deff$ and the effective spectral ratio
$\kappaeff$, we embed a $\deff\times \deff$ problem Hamiltonian $\mlvec{M}^{(1)}=\subspacecolumn^{\dagger}\mlvec{M}\subspacecolumn$ with
eigenvalues $(0, \frac{1}{\kappaeff}, 1, \cdots, 1)$, generators
$\mlvec{H}^{(1)} = \subspacecolumn^{\dagger}\mlvec{H}\subspacecolumn$ and
unitaries $\{\mlvec{U}_{l}^{(1)}=\subspacecolumn^{\dagger}\mlvec{U}_{l}\subspacecolumn \}_{l=1}^{p}$ into a $d$-dimensional space using arbitrary
$d\times d$ unitary
$\mlvec{U}_{\mathsf{embed}} = \begin{bmatrix}\subspacecolumn & \subspacecolumn^{\perp}\end{bmatrix}$
with $\subspacecolumn^{\perp}$ being arbitrary complementary columns of $\subspacecolumn$:
\begin{align}
\mlvec{M} &= \mlvec{U}_{\mathsf{embed}} \begin{bmatrix}\mlvec{M}^{(1)} & 0 \\ 0 & \mlvec{I}_{d-\deff\times d-\deff}\end{bmatrix} \mlvec{U}_{\mathsf{embed}}^{\dagger}\\
\mlvec{H} &= \mlvec{U}_{\mathsf{embed}} \begin{bmatrix}\mlvec{H}^{(1)} & 0 \\ 0 & 0 \end{bmatrix} \mlvec{U}_{\mathsf{embed}}^{\dagger}\\
\mlvec{U}_{l} &= \mlvec{U}_{\mathsf{embed}} \begin{bmatrix}\mlvec{U}^{(1)}_{l} & 0 \\ 0 & \mlvec{I}_{d-\deff\times d-\deff}\end{bmatrix} \mlvec{U}_{\mathsf{embed}}^{\dagger},\quad\forall l \in [p]
\end{align}
And the ansatz takes the form
\begin{align}
  \mlvec{U}(\mlvec\theta) = \big(\prod_{l=1}^{p}\mlvec{U}_{l}\exp(-i\theta_{l}\mlvec{H}) \big) \mlvec{U}_{0}
\end{align}
where $\deff\times \deff$ unitaries $\{\mlvec{U}^{(1)}_{l}\}$ are sampled i.i.d from the
Haar measure over $SU(\deff)$.
In Figure~\ref{fig:toy-thresh-full}, we plot the success rate versus the number
of parameters for various $\deff$ and $\kappaeff$ that are used to generate Figure~\ref{fig:toy-thresh}.
\begin{figure}
  \centering
  \subfigure[Varying effective dimension $\deff$]{
    \includegraphics[width=0.8\linewidth]{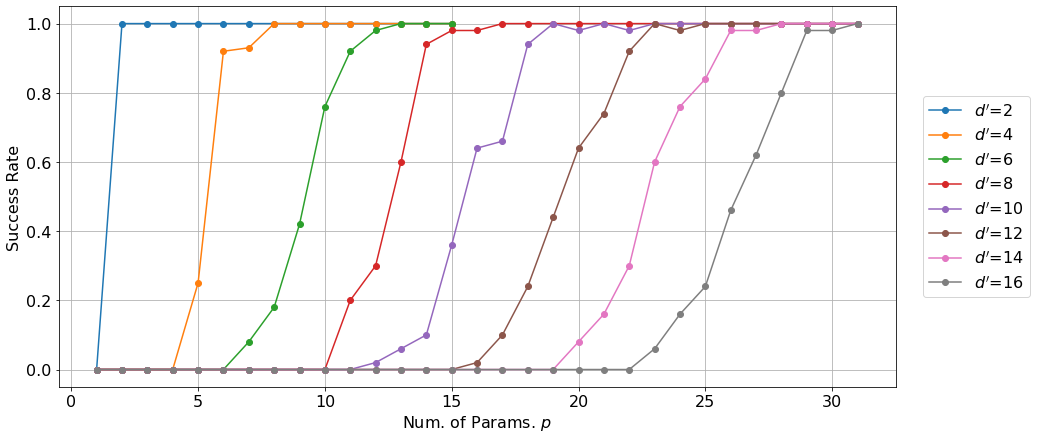}
  }\\
  \subfigure[Varying effective ratio $\kappaeff$]{
    \includegraphics[width=0.8\linewidth]{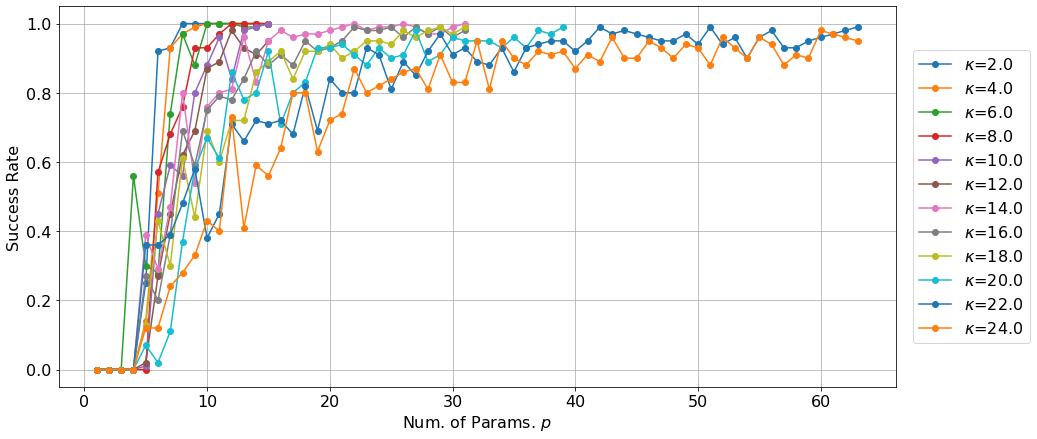}
  }
  \caption{
    The success rate for achieving a $0.01$-approximation for the ground state
    as a function of number of parameters. Each curve corresponds to a synthetic
    instance with dimension $16$ and with varying $(\deff, \kappaeff)$.
    Success rates are estimated over 100 random initializations.
    Top: Fixing $d=16, \kappaeff=4.0$ for $\deff=2, 4, 6, \cdots, 16$. The threshold increases as the system dimension increases.
    Bottom: Fixing $d=16, \deff=4$ for $\kappaeff=2.0, 4.0, 6.0, \cdots, 24.0$.
    The threshold
    is positively correlated to the spectral ratio of the system.
  }
  \label{fig:toy-thresh-full}
\end{figure}

\paragraph{Estimating the invariant subspace for TFI and XXZ models} Similar to
the Kitaev model in Section~\ref{subsec:estimate-eff}, we numerically confirm
that the TFI and XXZ models involved are all compatible. The convergences of the
empirical estimatino of projection $\hat{\mlvec{\Pi}}$ are summarized in
Figure~\ref{fig:app_invar_sub_tfi2}, Figure~\ref{fig:app_invar_sub_tfi3},
Figure~\ref{fig:app_invar_sub_xxz4} and Figure~\ref{fig:app_invar_sub_xxz6}. For
each of the plots, the x-axes corresponds to the indexes of the eigenvalues
sorted in the ascending orders. The value of $R$ in
Equation~\ref{eq:estimate_pi} ranges from $0$ to $100$ and is color-coded,
increasing from blue to red.

\begin{figure}[!htbp]
  \centering
  \subfigure[N=4]{
    \includegraphics[width=0.4\linewidth]{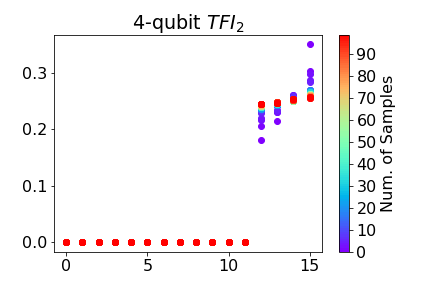}
  }
  \subfigure[N=6]{
    \includegraphics[width=0.4\linewidth]{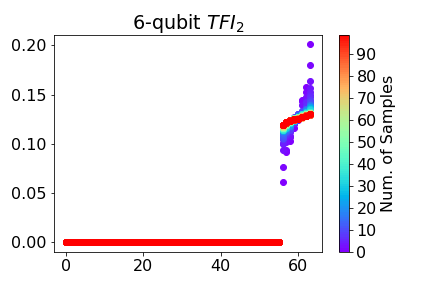}
  }\\
  \subfigure[N=8]{
    \includegraphics[width=0.4\linewidth]{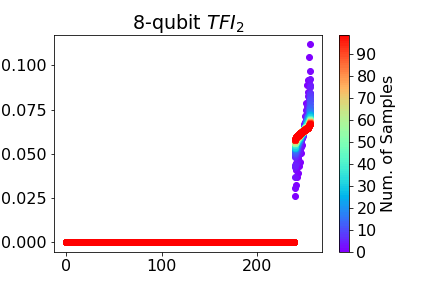}
  }
  \subfigure[N=10]{
    \includegraphics[width=0.4\linewidth]{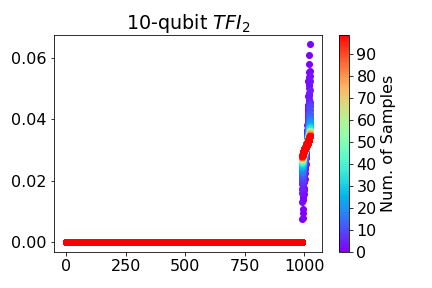}
  }
  \caption{Specturm of $\hat{\Pi}$ for $TFI_{\alt{2}}$ model with $4, 6, 8, 10$ qubits}
  \label{fig:app_invar_sub_tfi2}
\end{figure}

\begin{figure}[!htbp]
  \centering
  \subfigure[N=4]{
    \includegraphics[width=0.4\linewidth]{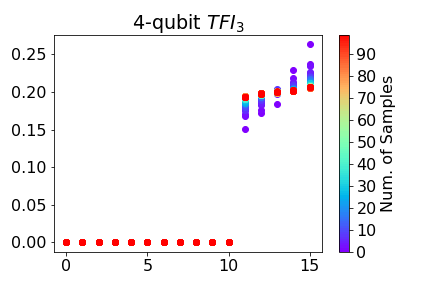}
  }
  \subfigure[N=6]{
    \includegraphics[width=0.4\linewidth]{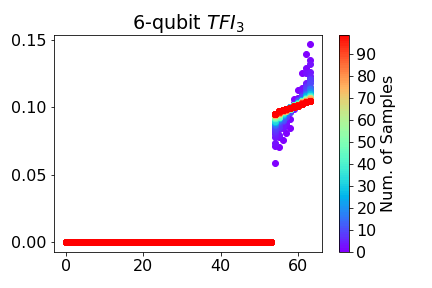}
  }\\
  \subfigure[N=8]{
    \includegraphics[width=0.4\linewidth]{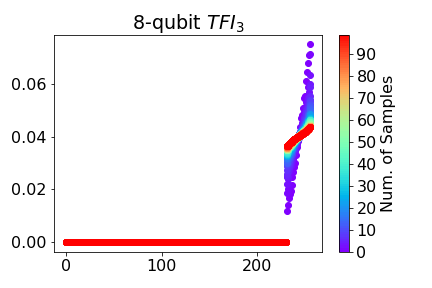}
  }
  \subfigure[N=10]{
    \includegraphics[width=0.4\linewidth]{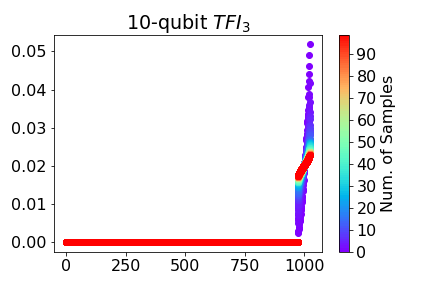}
  }
  \caption{Specturm of $\hat{\Pi}$ for $TFI_{\alt{3}}$ model with $4, 6, 8, 10$ qubits}
  \label{fig:app_invar_sub_tfi3}
\end{figure}

\begin{figure}[!htbp]
  \centering
  \subfigure[N=4]{
    \includegraphics[width=0.4\linewidth]{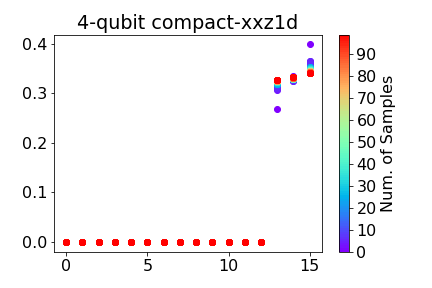}
  }
  \subfigure[N=6]{
    \includegraphics[width=0.4\linewidth]{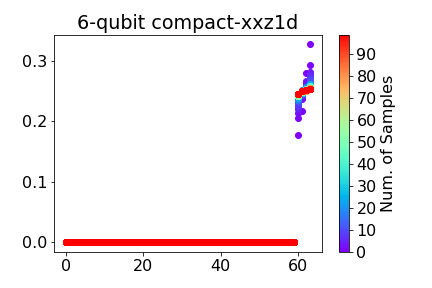}
  }\\
  \subfigure[N=8]{
    \includegraphics[width=0.4\linewidth]{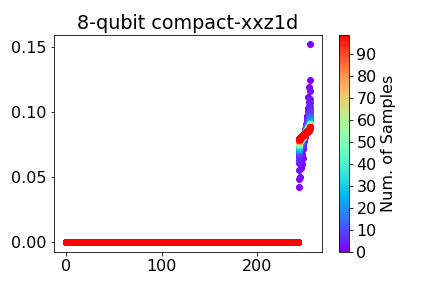}
  }
  \subfigure[N=10]{
    \includegraphics[width=0.4\linewidth]{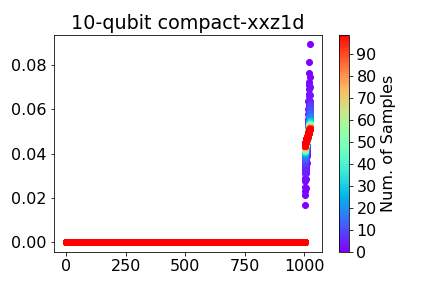}
  }
  \caption{Specturm of $\hat{\Pi}$ for $XXZ_{\alt{4}}$ model with $4, 6, 8, 10$ qubits}
  \label{fig:app_invar_sub_xxz4}
\end{figure}

\begin{figure}[!htbp]
  \centering
  \subfigure[N=4]{
    \includegraphics[width=0.4\linewidth]{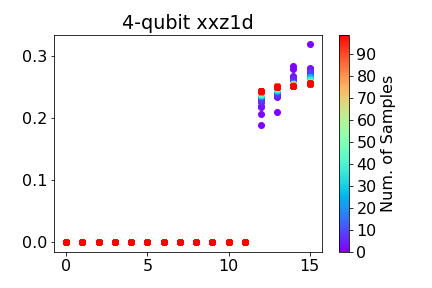}
  }
  \subfigure[N=6]{
    \includegraphics[width=0.4\linewidth]{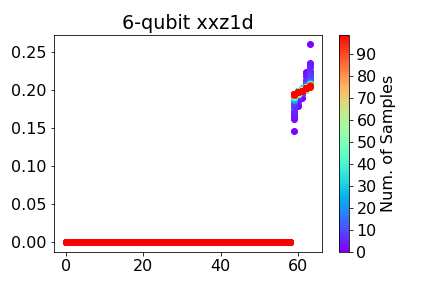}
  }\\
  \subfigure[N=8]{
    \includegraphics[width=0.4\linewidth]{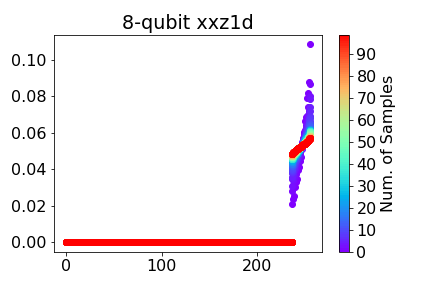}
  }
  \subfigure[N=10]{
    \includegraphics[width=0.4\linewidth]{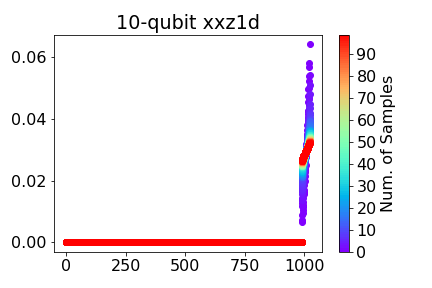}
  }
  \caption{Specturm of $\hat{\Pi}$ for $XXZ_{\alt{6}}$ model with $4, 6, 8, 10$ qubits}
  \label{fig:app_invar_sub_xxz6}
\end{figure}


\end{document}